\documentclass[
 amsmath,amssymb,
 aps, reprint,
prl
]{revtex4-2}

\usepackage{graphicx}% Include figure files
\usepackage{dcolumn}% Align table columns on decimal point
\usepackage{bm}% bold math
\usepackage{physics}
\usepackage{amsfonts, amsmath}
\usepackage{dsfont}
\usepackage{amsthm}
\usepackage{bm}
\usepackage{hyperref}
\usepackage{bbold}
\usepackage{color}
\usepackage{lipsum,babel}
\usepackage[normalem]{ulem}
\usepackage{tabularx}

\newcommand{\T}[0]{{\text{T}}}

\newcommand{\hro}{\hat{\rho}}

\newcommand{\hror}{\hat{\rho}^{(U)}}

\newtheorem{theorem}{Theorem}

\newtheorem{definition}{Definition}
\newtheorem{lemma}{Lemma}

\begin{document}
\definecolor{navy}{RGB}{46,72,102}
\definecolor{pink}{RGB}{219,48,122}
\definecolor{grey}{RGB}{184,184,184}
\definecolor{yellow}{RGB}{255,192,0}
\definecolor{grey1}{RGB}{217,217,217}
\definecolor{grey2}{RGB}{166,166,166}
\definecolor{grey3}{RGB}{89,89,89}
\definecolor{red}{RGB}{255,0,0}

\preprint{APS/123-QED}

\title{Exponential advantage in continuous-variable quantum state learning}
\author{Eugen Coroi}

\author{Changhun Oh}
\email{changhun0218@gmail.com}
\affiliation{Department of Physics, Korea Advanced Institute of Science and Technology, Daejeon 34141, Korea}

\begin{abstract}
We consider the task of learning quantum states in bosonic continuous-variable~(CV) systems.
We present an experimentally feasible protocol that uses entangled measurements and reflected states to efficiently learn, up to sign, the characteristic function of CV quantum states, with sample complexity independent of the number of modes~$n$. 
We prove that any adaptive scheme without entangled measurements requires exponentially many samples in $n$ for this learning task, thereby demonstrating an exponential advantage from entangled measurements.
Remarkably, we also prove that any entanglement-assisted scheme that does not have access to reflected states requires exponentially many samples in $n$.
Together, these results establish a rigorous exponential advantage that jointly relies on entangled measurements and reflected states.
Finally, we relate the sample‑complexity bounds to the input state's classicality, revealing a classicality–complexity tradeoff: entanglement‑free schemes suffice for highly classical states, whereas in the genuinely nonclassical regime the sample cost can be reduced only by providing both entanglement and reflected states.
\end{abstract}

\maketitle

%\section{Introduction}
Quantum systems often provide significant advantages in various information processing tasks over their classical counterparts, such as quantum computing~\cite{nielsen2002quantum, shor1994algorithms, lloyd1996universal} and quantum sensing~\cite{giovannetti2006quantum, giovannetti2011advances}.
Due to their fundamental and practical importance, such advantages have been both analyzed theoretically and demonstrated experimentally~\cite{hangleiter2023computational,arute2019quantum, wu2021strong, morvan_phase_2024, zhong2020quantum, zhong2021phase, madsen2022quantum,deng2023gaussian,decross2024computational}.
Beyond computing and sensing, recent works have identified analogous advantages in learning quantum systems~\cite{huang2021information, chen2022exponential, anshu2024survey}.

More specifically, it has been proven that using a controllable coherent quantum memory significantly reduces the sample complexity for learning a quantum system, offering an exponential advantage compared to any scheme without such a memory~\cite{huang2021information, chen2022exponential, aharonov2022quantum, anshu2024survey, king2024exponential, chen2022quantum, chen2023futility, chen2024tight}.
In many cases, subsequent experiments have demonstrated the predicted exponential advantage~\cite{huang2022quantum, seif2024entanglement,liu2025quantum}.
Despite this progress, most results to date concern discrete-variable~(DV) systems.
Extending these studies to CV quantum systems has been hindered by significant technical challenges posed by infinite-dimensional Hilbert spaces, despite the far-reaching applications of CV quantum systems in quantum information processing and the growing interest in learning these systems~\cite{wu2023quantum, gandhari2024precision, becker2024classical, mele2024learning}.

Very recently, an exponential advantage has been successfully extended to CV quantum systems by overcoming this technical barrier~\cite{oh2024entanglement}. 
The protocol in Ref.~\cite{oh2024entanglement} employs entangled probes and measurements with quantum memory to learn CV channels efficiently.
It also shows that any scheme without quantum memory requires exponentially more samples in the number of modes~$n$, thereby establishing an exponential advantage for channel learning in the CV setting.

\begin{figure}[t] % 't' for top of the column
  \centering
  \includegraphics[width=0.4\textwidth]{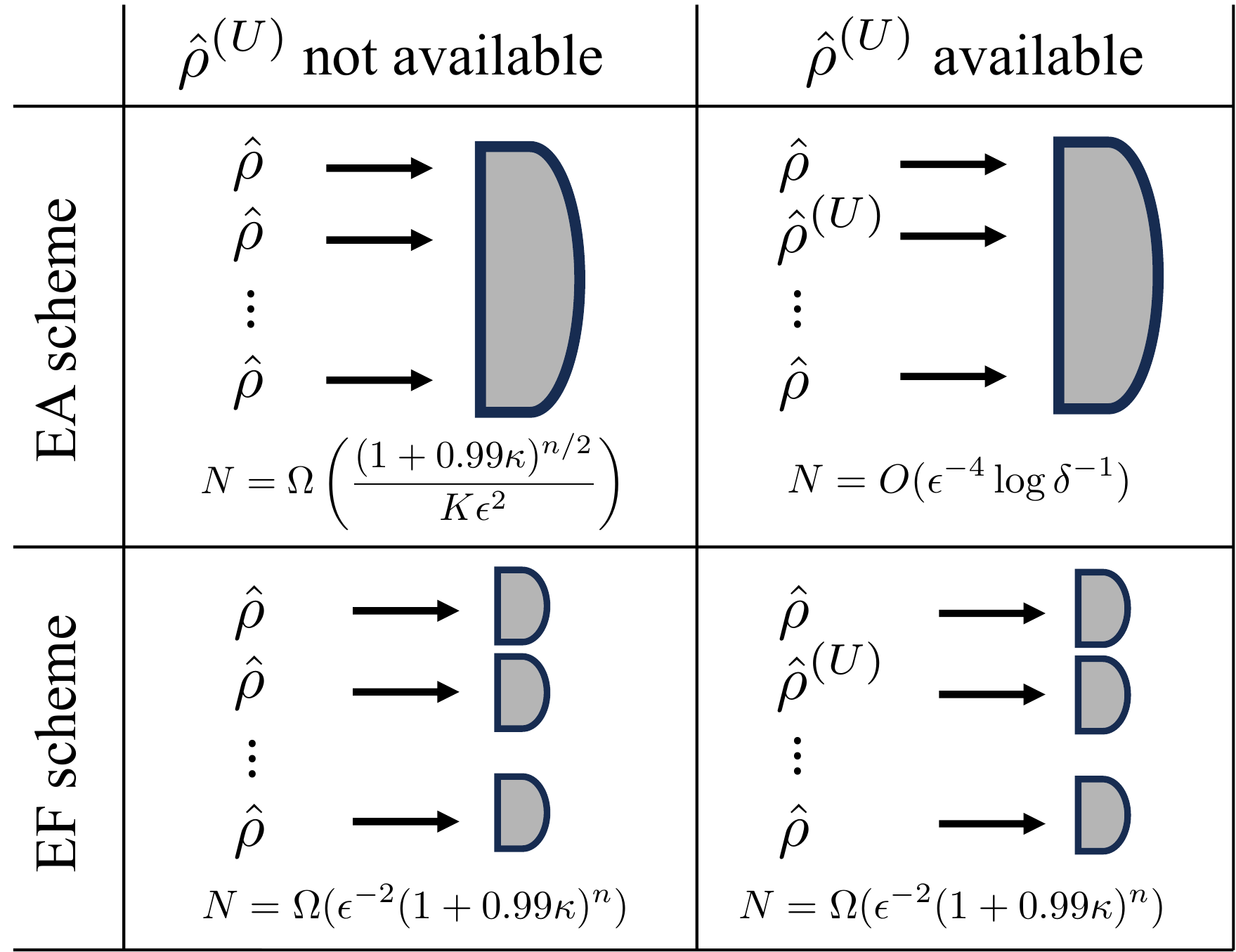} 
  \caption{Sample complexity for the learning task in Def.~\ref{def::learning-task}, categorized by (i) access to reflected states and (ii) access to entangled measurements.}
  \label{fig::quadrants}
\end{figure}

Building on these \emph{channel}‑level results, we prove an exponential advantage for learning CV quantum \emph{states} with quantum memory.
Inspired by recent results~\cite{wu2023quantum, king2024exponential}, we present a Bell‑measurement (BM) scheme that learns the state’s characteristic function (up to a sign) with sample complexity independent of the number of modes~$n$, via joint measurements on the state and its reflected state (defined below).
Notably, we rigorously prove that when entangled measurements are prohibited, any scheme must use exponentially many samples in $n$, even if adaptive strategies and access to reflected states are allowed.
Thus, we establish an exponential separation between entanglement‑assisted~(EA) and entanglement‑free~(EF) schemes for learning CV quantum states~(see Fig.~\ref{fig::quadrants}).

For reflection-symmetric states, the BM scheme on two copies of the state yields an exponential advantage over any EF scheme. 
We further show that when symmetry is absent, access to a reflected state is essential: without it, the task remains exponentially hard even with entangled measurements.
Therefore, our work fully characterizes the sample complexity of learning CV quantum states from the perspective of entangled measurements and access to reflected states.
Lastly, we establish bounds linking state classicality to the sample complexity of the learning task, revealing a tradeoff between classicality and the need for entanglement and reflected states.

%\section{Problem Setup}
 
\textit{Problem Setup.{\textemdash}}
Consider an $n$-mode bosonic CV system, defined by the bosonic annihilation operators $\{\hat{a}_i\}_{i=1}^n$, and the task of learning an $n$-mode CV quantum state~$\hro$.
Any CV quantum state can be expressed in the displacement-operator basis $\{ \hat{D}(\alpha) \equiv \bigotimes_{i=1}^{n} e^{\alpha_i \hat{a}^{\dagger}_i -  \alpha_i^* \hat{a}_i } \}_{\alpha\in \mathbb{C}^{n}}$, as~\cite{cahill1969density, ferraro2005gaussian, serafini2017quantum}
\begin{align}\label{eq:expansion}
    \hro=\frac{1}{\pi^n}\int d^{2n}\alpha\, \chi_{\hro}(\alpha)\,\hat{D}^\dagger(\alpha),
\end{align}
where $\chi_{\hro}(\alpha)\equiv \Tr[\hro\hat{D}(\alpha)]$ is the characteristic function of the quantum state and its Fourier transform yields the Wigner function~\cite{cahill1969density, ferraro2005gaussian, serafini2017quantum}.
Since $\chi_{\hro}$ fully specifies the state and suffices to compute physical observables~\cite{ferraro2005gaussian, serafini2017quantum}, we define learning a CV quantum state $\hro$ as learning its characteristic function.
% Since $\chi_{\hro}$ fully specifies the state, it is an equivalent description and suffices to compute physical observables~\cite{ferraro2005gaussian, serafini2017quantum}.
% Hence, in this work, we define learning a CV quantum state $\hro$ as learning its characteristic function.
More specifically, we define the learning task as follows~\cite{huang2021information,huang2022quantum,chen2022quantum,chen2022exponential,chen2024tight,oh2024entanglement}:
\begin{definition}[Learning Task]\label{def::learning-task}
Let $\hro$ be an $n$-mode CV quantum state, and fix $\epsilon,\delta,\kappa>0$. 
A learning scheme using $N$ copies of $\hro$ (and, when allowed, associated resource states) performs possibly joint measurements on these states and records classical outcomes.
Post-measurement, given an arbitrary query $\alpha\in\mathbb{C}^n$ satisfying $|\alpha|^2\leq \kappa n$, the scheme outputs an estimate $\tilde{u}$ such that
\begin{equation}
\Pr\left[\min_{\tau\in\{\pm1\}}|\tau\tilde{u}-\chi_{\hro}(\alpha)|\leq \epsilon\right]\geq 1-\delta.
\end{equation}
\end{definition}
The protocol utilizes the reflected state $\hror$ as a resource, defined via the characteristic function~\cite{wu2023quantum}
\begin{align}\label{eq:char0}
    \chi_{\hror}(\alpha) \equiv \chi_{\hro}(U\alpha^*),
\end{align}
where $U \in \mathbb{C}^{n\times n}$ is a symmetric unitary matrix ($U^{-1}=U^{*}=U^{\dagger}$).
For example, $U=-I$ yields the complex-conjugate state ${\hro^{(-I)}=\hro^*}$ (see SM Sec.~S2~\cite{supple}).
Physically, the map $\alpha \mapsto -\alpha^*$ corresponds to a reflection along the momentum quadrature in phase space; generally, $U$ specifies the axes of reflection.
This reflected state is a key resource enabling the exponential advantage demonstrated below.
We say that $\hro$ has reflection symmetry if there exists a symmetric unitary matrix $U$ such that $\hro=\hror$.
Moreover, for a given $U$, there exists a passive unitary operator $\hat{\mathcal{U}}$ implemented with beam splitters and phase shifters~\cite{reck1994experimental}, such that~\cite{wu2023quantum}
\begin{align}\label{eq:passive}
    \chi_{\hat{\mathcal{U}}\hro\hat{\mathcal{U}}^\dagger}(\alpha) = \chi_{\hro} (U^{\T}\alpha),
\end{align}
which immediately follows from Eq.~\eqref{eq:expansion} and the definition of the displacement operator.

%\section{Quantum BM Protocol}
 
\textit{BM Protocol.{\textemdash}}
We first present a concrete protocol, inspired by and closely related to Refs.~\cite{wu2023quantum, king2024exponential}, that solves the learning task (see Fig.~\ref{fig:BM}).
Consider an $n$-mode CV state $\hro$ and suppose we have access to a reflected state $\hro^{(U)}$ of $\hro$, with the corresponding symmetric unitary $U$ known.
To learn $\chi_{\hro}(\alpha)$ (up to sign), we apply a linear-optical circuit $\hat{\mathcal{U}}$ to the reflected state to obtain $\hat{\mathcal{U}}\hro^{(U)}\hat{\mathcal{U}}^\dagger$, and then measure the joint state $\hro\otimes \hat{\mathcal{U}}\hro^{(U)}\hat{\mathcal{U}}^\dagger$ using the CV Bell-measurement (BM), implemented by 50:50 beam splitters followed by homodyne detection of the $\hat{x}$ and $\hat{p}$ quadratures, respectively~\cite{serafini2017quantum}.
Here, the POVM elements of the BM are given by $\{\hat{\Pi}(\zeta)\}_{\zeta\in\mathbb{C}^n}$ and have the following form: $(I\otimes \hat{D}(\zeta))|\Psi\rangle\langle \Psi|(I\otimes \hat{D}^\dagger(\zeta))/\pi^n$, where $|\Psi\rangle$ denotes the tensor product of $n$ infinitely squeezed TMSV states, each proportional to $\sum_{k=0}^{\infty}\ket{k}\ket{k}$ when expressed in the Fock basis \footnote{The Bell-basis description is a formal representation of the POVM. 
The measurement is physically implemented using linear optics and homodyne detection 
and does not require preparing infinitely squeezed states.}.

% Place figure at the top of the second column
\begin{figure}[t] 
  \centering
  \includegraphics[width=0.35\textwidth]{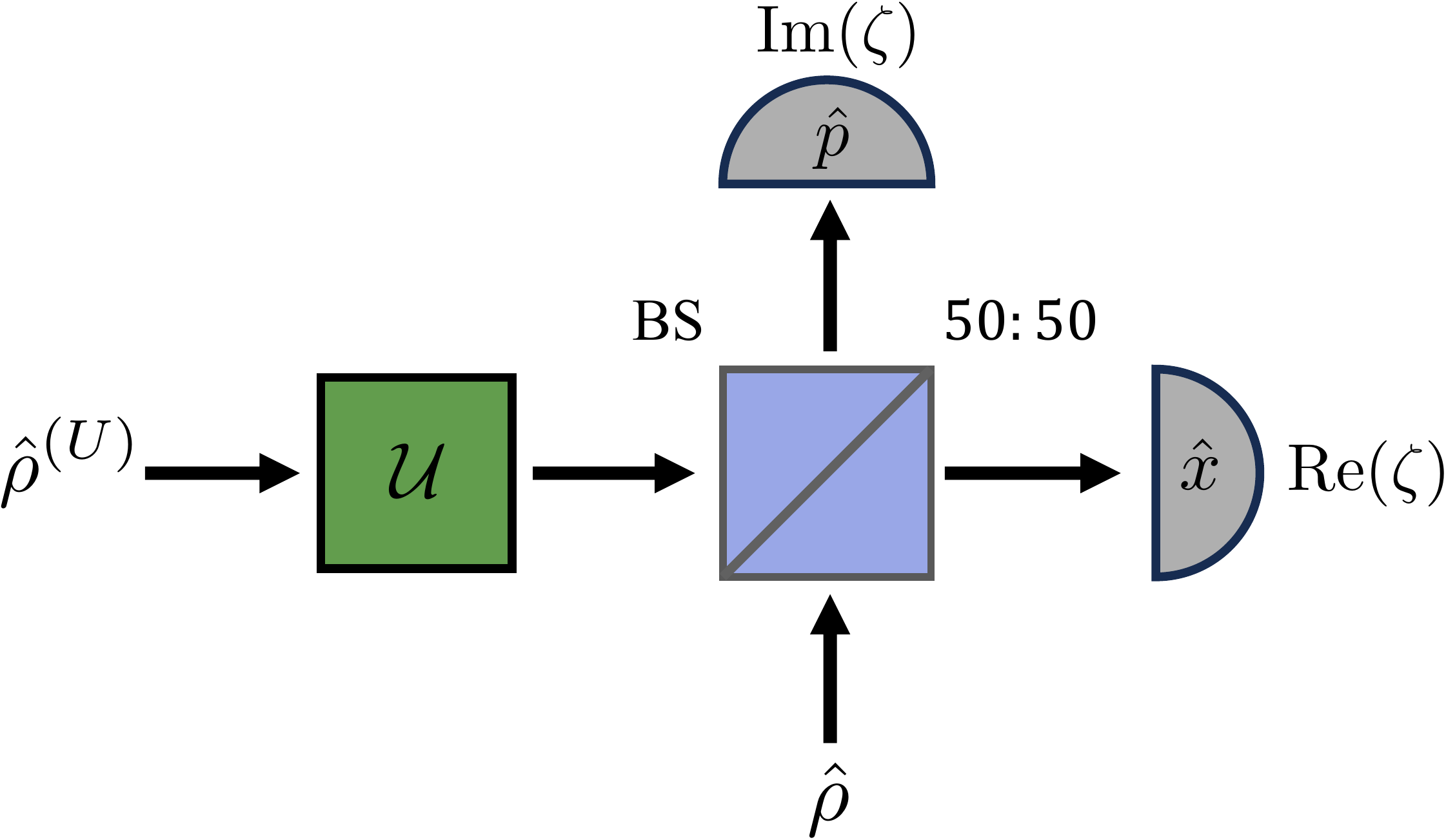} % replace with your image
  \caption{BM protocol. See the main text for the details.}
  \label{fig:BM}
\end{figure}

\iffalse
Using Eqs.~\eqref{eq:char0} and \eqref{eq:passive} and the POVM form above, the probability of outcome $\zeta\in\mathbb{C}^n$ is (see SM Sec.~S3~\cite{supple})
\begin{align}
p(\zeta) = \int \frac{d^{2n}\alpha }{\pi^{2n}}~\chi_{\hro}^2(\alpha) e^{\zeta\cdot\alpha - \zeta^*\cdot\alpha^*}.    
\end{align}

Using Eqs.~\eqref{eq:char0} and \eqref{eq:passive} and the POVM form above, one computes the outcome probability distribution $p(\zeta)$ (see SM Sec.~S3~\cite{supple}).
Inverting the Fourier transform gives the squared characteristic function:
\begin{align}
\chi_{\hro}^2(\alpha) = \int d^{2n}\zeta~p(\zeta) e^{-(\zeta\cdot\alpha - \zeta^*\cdot\alpha^*)}.
\end{align}
\fi
Eqs.~\eqref{eq:char0} and \eqref{eq:passive} together with the POVM above show that the squared characteristic function of the state $\chi^2_{\hro}(\alpha)$ is expressed in terms of the BM outcome distribution $p(\zeta)$ (see SM Sec.~S3~\cite{supple}):
\begin{align}
  \chi_{\hat\rho}^2(\alpha)
  = \int d^{2n}\zeta\, p(\zeta)\, e^{-(\zeta\cdot\alpha - \zeta^*\cdot\alpha^*)}.
\end{align}
Averaging $e^{-(\zeta^{(i)}\cdot \alpha-(\zeta^{(i)})^* \cdot \alpha^*)}$ over $N$ samples $\{\zeta^{(i)}\}_{i=1}^{N}$ yields an unbiased estimator of $\chi_{\hat{\rho}}^{2}(\alpha)$. 
Notably, the measurement data is independent of the query $\alpha$, allowing the same samples to be reused for multiple queries.
Applying Hoeffding's inequality and the union bound yields:
\begin{lemma}
    There exists an EA reflected state-assisted protocol 
    that learns $\chi^2_{\hro}(\alpha)$ within $\epsilon$-accuracy and $1-\delta$ success probability, for any $\alpha \in \mathbb{C}^{n}$, given 
    $N=O(\epsilon^{-2}\log (1/\delta))$ copies of $\hro$ and $\hror$.
\end{lemma}
Increasing the sample complexity to~$O(\epsilon^{-4})$ enables the estimation of the characteristic function itself (up to sign).
Hence, the learning task can be solved with a complexity independent of the system size~$n$:
\begin{theorem}[BM]\label{thm:BM}
    There exists an EA reflected state-assisted protocol that accomplishes the learning task given $N=O(\epsilon^{-4}\log (1/\delta))$ copies of $\hro$ and $\hror$.
\end{theorem}
Thus, the BM protocol solves the task using only two-copy joint measurements with $n$-independent sample complexity.
This is optimal in $n$ and near-optimal in $\epsilon$ due to the $N=\Omega(\epsilon^{-2})$ lower bound (see SM Sec.~S11~\cite{supple}).
Crucially, we show that this $n$-independence relies on having access to both entangled measurements and reflected states: removing either resource implies exponential sample complexity in $n$.
Finally, in SM Sec.~S4~\cite{supple}, we analyze the relationship between our BM state-learning protocol and the random-displacement channel-learning framework, clarifying both the overlap and the fundamental distinctions between the two tasks.
\iffalse
Determining the signs is a challenging task in and of itself, especially in CV systems.
We resolve these issues and learn the characteristic function together with its sign without extra cost for Gaussian states.
We also propose a heterodyne-based \textit{oblivious} protocol that works for arbitrary inputs $\hro$, albeit with exponential sample complexity (see SM Sec.~S3~\cite{supple} for further details on the challenges and our solutions).

Unlike the EA CV channel learning protocol~\cite{oh2024entanglement}, where the probe state's squeezing parameter is required to scale with the system size for the sample complexity to be independent of $n$, our CV state learning protocol does not require such an experimentally stringent condition.
\fi

%\section{Without entangled measurement}
 
\textit{Without entangled measurements.{\textemdash}}
To establish the advantage of the BM protocol, we prove that prohibiting entangled measurements forces an exponential sample complexity for the learning task, 
even for fully adaptive strategies that update measurement bases based on prior outcomes.
% To establish the advantage of the BM protocol, we prove that prohibiting entangled measurements forces an exponential sample complexity for the learning task.
% This lower bound holds even for adaptive strategies where measurement bases are updated based on prior outcomes.
Critically, we allow the learner access to the reflected state $\hat{\rho}^{(U)}$, ensuring that the separation arises solely from the inability to perform entangled measurements.

% Interestingly, even a restricted version of the task, estimating $\chi_{\hro}(\alpha)$ only on the $|\alpha|^2\leq \kappa n$ phase space region, already requires exponential sample complexity.

\begin{theorem}[Single-copy lower bound]\label{thm::scaling-single-copy}
Let $\hro$ be an $n$-mode quantum state ($n\geq 8$) and let $\hror$ denote a reflected state of $\hro$. 
Fix $\kappa>0$ and $\epsilon\in(0,0.245)$. 
Any EF scheme succeeding in the learning task with $2/3$ success probability must use
\begin{equation}
N=\Omega(\epsilon^{-2}\,(1+0.99\kappa)^{n}).
\end{equation}
\end{theorem}

\begin{proof}[Proof Sketch]
    (See SM Sec.~S6~\cite{supple} for the full proof).
    We demonstrate the lower bound by restricting the learning task to a specific family of states, the $n$-mode ``three-peak states", $\hat{\rho}_{s\gamma}$, defined by the characteristic function:
\begin{align}
        &\chi_{\hro_{s\gamma}}(\alpha) \nonumber\\ 
        &=e^{-\frac{|\alpha|^2}{2\Sigma^2}}\left[e^{-\frac{|\alpha|^2}{2\sigma^2}}+2is\epsilon_0 e^{-\frac{|\gamma|^2}{2\Sigma^2}}\left(e^{-\frac{|\gamma-\alpha|^2}{2\sigma^2}}-e^{-\frac{|\gamma+\alpha|^2}{2\sigma^2}}\right)\right],
\end{align}
    where $\gamma \in \mathbb{C}^{n}$, $\epsilon_0\leq 1/4$, $s\in \{\pm1\}$, and the variances $\sigma^2\equiv (1/\nu-\nu)/2$ and $\Sigma^2\equiv (1+\nu)/(1-\nu)$ are set by a constant $\nu \in (0,1)$ determined by $\kappa$ ($\nu$ is typically close to 1).
    Its operator representation is provided in SM Sec.~S5~\cite{supple}.
    Figure~\ref{fig::state-display} illustrates the characteristic and Wigner functions of a single-mode three-peak state.
    % We illustrate the characteristic function and Wigner function of a three-peak state in Fig.~\ref{fig::state-display}.
    % For ease of visualization, we display single-mode states.
    We formulate a hypothesis-testing game between Alice and Bob, in which Bob succeeds with high probability if he learns a characteristic function well using $N$ copies of $\hro$ and $\hro^{(U)}$.
    
    More specifically, Alice samples $s$ uniformly and $\gamma\in \mathbb{C}^n$ from a multivariate normal distribution with variance $\sigma_\gamma^2 = 0.99\kappa/2$. 
    With equal probability, she then prepares $N$ copies of one of the two hypotheses 1)~$\hro_0$ or $\hro_0^{(U)}$, or 2)~$\hro_{s\gamma}$ or $\hro_{s\gamma}^{(U)}$. 
    Bob may request any configuration of the original and reflected states for Alice's chosen hypothesis (e.g., $\hro,\hror,\hror,\dots,\hro$).
    Bob learns the provided states using his EF scheme. Once measurements are complete, Alice reveals $\gamma$.
    If Bob's scheme accurately estimates $\chi_{\hat{\rho}}(\pm\gamma)$, he can distinguish the hypotheses, as the characteristic functions differ significantly at $\alpha=\pm\gamma$.
    Successful discrimination indicates that Bob's learning scheme produces a sufficiently large total variation distance~(TVD) between the output probability distributions of the two hypotheses.
    
    Meanwhile, we directly compute an upper bound on the TVD for any EF scheme and show that the per-copy increase in TVD is exponentially small. 
    Hence, the required number of copies $N$ is exponentially large. 
\end{proof}

\begin{figure}[t] % 't' for top of the column
  \centering
  \includegraphics[width=0.37\textwidth]{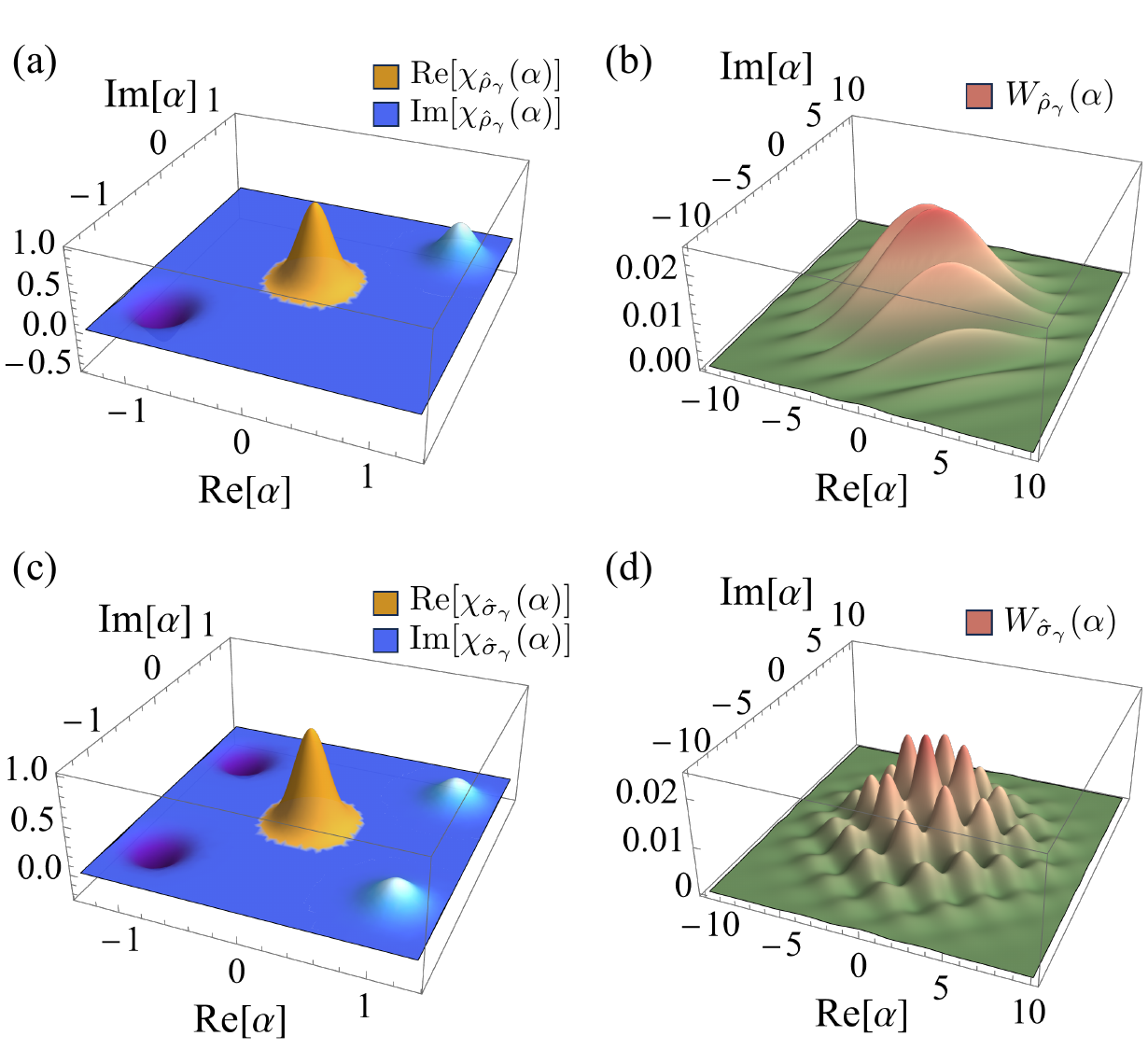} % replace with your image
  \caption{Characteristic~(left) and Wigner functions~(right) of three-peak~(top) and five-peak~(bottom) states with ${\gamma_R = 1}, {\gamma_I = 0.8}, {\nu =0.98}, {\epsilon_0 = 0.99/4}$. For the five-peak state, $U=-I$ (See SM Sec.~S5~\cite{supple}).} 
  \label{fig::state-display}
\end{figure}
Although our proof focuses on the specific three-peak family, any valid learning scheme must succeed on all $n$-mode input states. Thus, exhibiting even a single subfamily on which the sample complexity is exponential already yields a general exponential lower bound.
%Although the proof focuses on the specific three-peak family, any valid learning scheme must work for all $n$-mode input states, so this example already yields a general exponential lower bound on the sample complexity.
% Although the proof focuses on the specific three-peak family, a general learning scheme must provide guarantees for any $n$-mode input state. 
% Thus, the difficulty of this specific family lower-bounds the sample complexity of the general task, making the bound fully general.
Theorems~\ref{thm:BM} and~\ref{thm::scaling-single-copy} collectively establish an exponential separation between EF and EA schemes for learning CV states.
Since the BM scheme solves the learning task for arbitrary $\alpha\in\mathbb{C}^n$, this gap can be made arbitrarily large by increasing the query range $\kappa$.
% Notably, the BM scheme uses only two-copy joint measurements, i.e., minimal quantum memory (see also Ref.~\cite{king2024exponential}).

For the exponential gap above, we assumed access to a reflected state $\hror$ even when entangled measurements were prohibited. This assumption might seem artificial; however, it has been suggested that such scenarios arise in practice~\cite{king2024exponential}. 
We now show the same exponential advantage without any additional resource states by considering inputs with intrinsic reflection symmetry, $\hro=\hro^{(U)}$ for some symmetric unitary $U$~\cite{wu2023quantum}. 
In this case, Thm.~\ref{thm:BM} is unchanged, whereas any EF scheme is again lower bounded by an exponential cost:

\begin{theorem}[Reflection-symmetric inputs]\label{thm::scaling-rho-equal-reflected}
Let $\hro$ be an $n$-mode quantum state ($n\geq 8$) with reflection symmetry $\hro=\hro^{(U)}$. 
Fix $\kappa>0$ and $\epsilon\in(0,0.1225)$. 
Any EF learning scheme succeeding in the learning task with $2/3$ success probability must use
\begin{equation}
    N=\Omega(\epsilon^{-2}\,(1+0.66\kappa)^{n}).
\end{equation}
\end{theorem}
The proof is similar to that of Thm.~\ref{thm::scaling-single-copy} except that instead of three-peak states, we introduce a family of reflection-symmetric five-peak states~(see SM Sec.~S7~\cite{supple}).
The characteristic and Wigner functions of a single-mode five-peak state are depicted in Fig.~\ref{fig::state-display}. As pointed out in Ref.~\cite{wu2023quantum},  reflection-symmetric states arise frequently in quantum optics, such as Fock states, a subset of Gaussian states~\cite{ferraro2005gaussian}, cat states~\cite{yurke1986generating}, and GKP states~\cite{gottesman2001encoding}.
Hence, the assumption for proving the exponential advantage is not stringent and the exponential advantage may be of practical importance.

\iffalse
\begin{align}
    \hat{\sigma}_{s\gamma}&=(1-\nu^2)^{n}\nu^{\hat{N}} 
  \bigg[\hat{D}^{\dagger}(0) + is\epsilon_0 \bigg( \hat{D}^\dagger(\gamma) - \hat{D}^\dagger(-\gamma) \nonumber
  \\ &+ \hat{D}^\dagger(U^{\T}\gamma^*) - \hat{D}^\dagger(-U^{\T}\gamma^*) \bigg) \bigg] \nu^{\hat{N}}.
\end{align}
\fi

%\section{Without reflected-states}
 
\textit{Without reflected states.{\textemdash}}
In general, $\hro$ need not have reflection symmetry, while the BM protocol requires access to a reflected state $\hror$. 
We ask whether the exponential gap persists when entangled measurements across $K$ copies are allowed, but no reflected state is available. 
Notably, we prove that $\hror$ is indeed crucial: even with joint measurements across $K$ copies of the input state, $\hro^{\otimes K}$, the sample complexity remains exponential in $n$.

\begin{theorem}[No reflected state lower bound]\label{thm:no-reflected}
Let $\hro$ be an $n$-mode quantum state ($n\geq 8$) without reflection symmetry. 
Fix $\epsilon\in(0,0.245)$, $K\leq 0.22/\epsilon$, and $\kappa\geq 1/0.99$. 
Any EA scheme, performing joint measurements on up to $K$ input states, $\hro^{\otimes K}$, succeeding in the learning task with $2/3$ success probability must use
\begin{equation}
N=\Omega  (K^{-1}\epsilon^{-2} (1+0.99\kappa)^{n/2}).
\end{equation}
\end{theorem}
Similarly, the proof proceeds via a hypothesis-testing task (see SM Sec.~S8~\cite{supple}) in which the learner employs a general EA adaptive scheme, without access to reflected states, to distinguish between two hypotheses. 
% When $\hro^{(U)}$ is available, the BM protocol completes the learning task using only two-copy joint measurements ($K=2$).
This result confirms that the exponential advantage stems from the combination of entangled measurements and reflected states; neither resource alone suffices.

\begin{figure}[t] 
  \centering
  \includegraphics[width=0.37\textwidth]{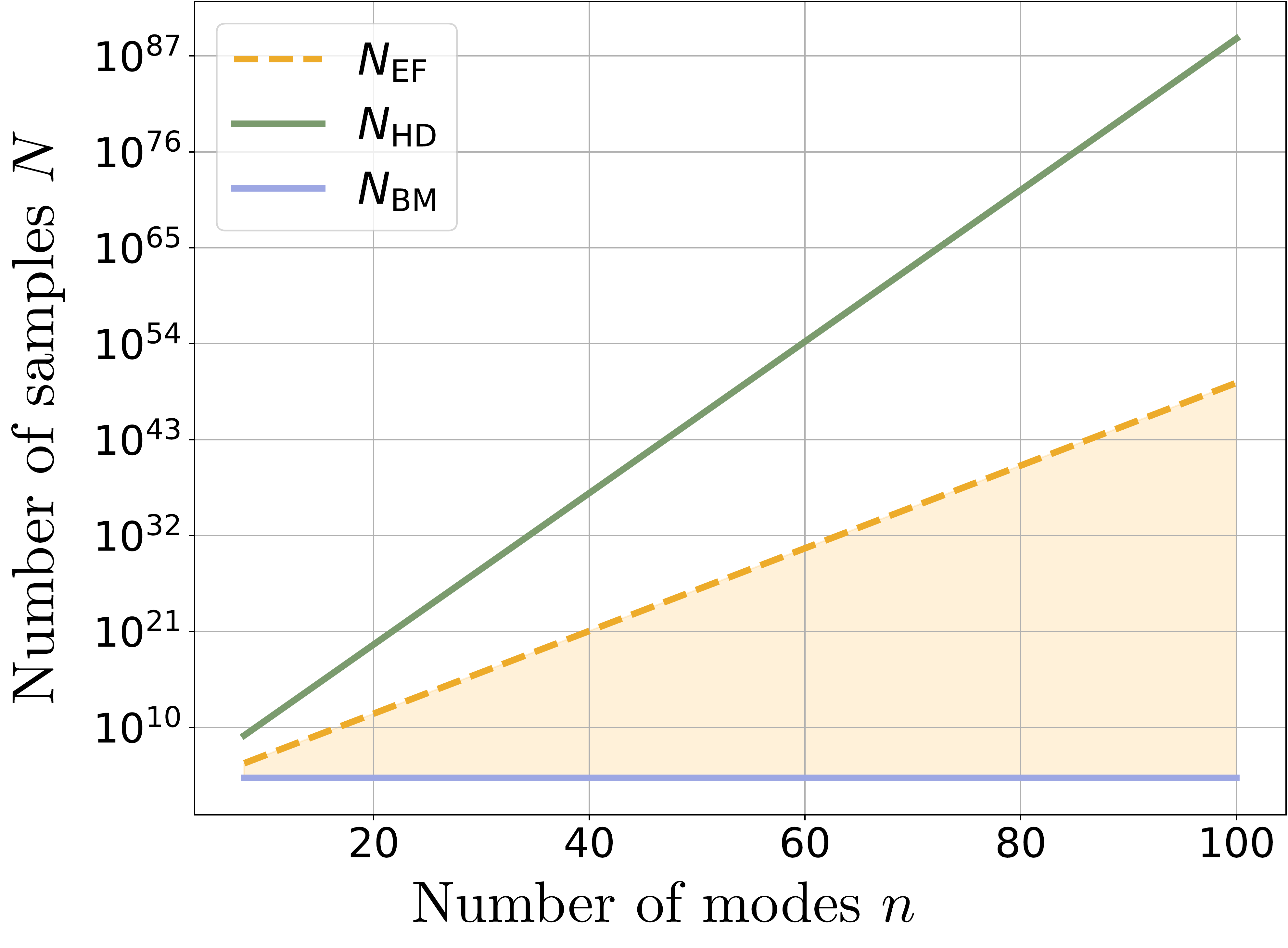} 
  \caption{Sample complexity for the BM, HD~(upper bounds), and EF~(lower bound). The orange-shaded area represents the region of advantage of the BM over any EF scheme. Here, $\epsilon=0.09$, $\delta=1/3$, and $\kappa=2$.}
  \label{fig::bound-comparision}
\end{figure}

\textit{Heterodyne detection~(HD).{\textemdash}}
To provide a concrete EF upper bound, we analyze a heterodyne detection scheme~\cite{serafini2017quantum, ferraro2005gaussian}, implemented via homodyne measurements of orthogonal quadratures after 50:50 beam splitters.
The POVM elements are given by $\{\hat{\Pi}_\zeta\}_{\zeta\in\mathbb{C}^n}$, where $\hat{\Pi}_\zeta\equiv |\zeta\rangle\langle\zeta|/\pi^n$ with $|\zeta\rangle$ representing coherent states.
The characteristic function of the state can be written as the inverse Fourier transform of the outcome distribution $p(\zeta)$~(see Sec.~S3 in SM~\cite{supple}) as
\begin{align}
    \chi_{\hro}(\alpha)=e^{\frac{|\alpha|^2}{2}}\int d^{2n}\zeta~p(\zeta)e^{\zeta^\dagger\alpha-\alpha^\dagger\zeta}.
\end{align}
Given $N$ outcomes~$\{\zeta^{(i)}\}_{i=1}^N$, $\tilde{\chi}_{\hro}(\alpha)\equiv\frac{1}{N}\sum_{i=1}^N e^{\frac{|\alpha|^2}{2}}e^{(\zeta^{(i)})^\dagger\alpha-\alpha^\dagger(\zeta^{(i)})}$ is an unbiased estimator for the characteristic function. The number of samples required to achieve $\epsilon$-accuracy with probability $1-\delta$ is
\begin{align}\label{eq::heterodyne-bound}
    N_{\text{HD}}=O(\epsilon^{-2} e^{|\alpha|^2}\log(1/\delta)).
\end{align}
Restricting queries to $|\alpha|^2 \leq \kappa n$ yields
$N_{\text{HD}} = O(\epsilon^{-2} e^{\kappa n}\log(1/\delta))$.
As with the general EF lower bounds, the HD protocol also exhibits exponential scaling (with a different base), suggesting either that tighter EF lower bounds may exist or that more efficient EF schemes than HD are possible.
Figure~\ref{fig::bound-comparision} illustrates that, for reasonable values of $\epsilon$, $\delta$, and $\kappa$, the sample complexity gap is enormous, e.g., $10^{19}$ ($10^{39}$) for $n=50$ ($100$) modes. 
Additionally, the EF lower bound and the EA BM scheme differ by a factor of $10^{21}$ ($10^{45}$) samples for $n=50$ ($n=100$) modes.

\textit{Classicality–complexity tradeoff.{\textemdash}}
\iffalse
The lower and upper bounds above quantify the sample complexity required to guarantee accuracy and success when learning an arbitrary $n$-mode input state.
However, one would expect the learning cost to differ between highly \textit{quantum} and highly \textit{classical} states. \fi
The general EF results capture the exponential cost for learning {\it arbitrary} states.
One naturally wonders whether specific classes of states avoid this prohibitive scaling.
We therefore analyze the connection between learning cost and the state's degree of (non)classicality.
Among various indicators capturing nonclassicality~\cite{kenfack2004negativity, albarelli2018resource,takagi2018convex, tan2020negativity}, we focus on the \emph{nonclassical depth} $\tau \equiv \tfrac{1-s_{\max}}{2}$~\cite{lee1991measure}, where $s_{\max}$ is the largest $s$ for which the $s$-ordered QPD $W_{\hro}(s,\beta)$ is a valid probability distribution~\cite{barnett2002methods}.
This regularity ensures the Gaussian tail bound $|\chi_{\hro}(\alpha)| \leq e^{- s_{\text{max}}\frac{|\alpha|^2}{2}}$ (see SM Sec.~S10~\cite{supple}).
For $s_{\max}>0$, this decay defines a cutoff radius $|\alpha|^2 \leq L_{\epsilon}(s_{\text{max}})\equiv \tfrac{2}{s_{\text{max}}}\log(1/\epsilon)$ beyond which the characteristic function magnitude is negligible, effectively bounding the required sample complexity.

\begin{theorem}[Classicality-aware upper bound]\label{thm:heterodyne-classicality}
Fix $\epsilon,\delta\in(0,1)$, $\kappa>0$, and $S\in (0,1]$.
Let $\hro$ be an $n$–mode state with classicality $s_\text{max}\in [S,1]$.
Then, there exists an EF protocol that solves the learning task using
\begin{align}   
N=
\begin{cases}
O(\epsilon^{-2}e^{\kappa n}\log(1/\delta)), & \text{if }\kappa n\leq L_{\epsilon}(S),\\
O(\epsilon^{-2\left(1+\frac{1}{S}\right)}\log(1/\delta)), & \text{if }\kappa n\geq L_{\epsilon}(S).
\end{cases}
\end{align}
\end{theorem}

\begin{proof}[Proof idea.]
Perform HD on each copy.
Given $s_{\text{max}}$, use the standard unbiased HD estimator for $|\alpha|^2 \leq L_{\epsilon}(s_{\text{max}})$ (as in the general case).
For $|\alpha|^2 \geq L_{\epsilon}(s_{\text{max}})$, the tail bound ensures $|\chi_{\hro}(\alpha)|\leq \epsilon$; output $\tilde{u}=0$, which is $\epsilon$‑accurate by design.
A full proof and additional implications are provided in SM Sec.~S10~\cite{supple}.
\end{proof}
%Note that statistical mixtures of coherent states are called classical states~($s_{\text{max}}=1$)~\cite{serafini2017quantum}, in which case the bound gives $N=O(\epsilon^{-4}\log(1/\delta))$, matching the EA BM scaling but achieved here without entanglement or reflected states.
Note that the bound for classical states~($s_{\text{max}}=1$)~\cite{serafini2017quantum}, statistical mixtures of coherent states, becomes $N=O(\epsilon^{-4}\log(1/\delta))$, matching the EA BM scaling but achieved here without entanglement or reflected states.
Hence, a classicality–complexity tradeoff emerges: in the more classical regime the HD scheme suffices, whereas in the highly nonclassical regime, the sample complexity can be reduced only by providing entanglement and reflected states. 
We complement the heterodyne upper bound with a classicality‑aware lower bound.

\begin{theorem}[Classicality-aware lower bound]\label{thm::scaling-single-copy2-classicality}
Let $\hro$ be an $n$-mode quantum state ($n\geq 8$) with classicality $s_{\text{max}}\in[S,1)$, where $S\in(0,1)$.
Fix $\epsilon\in(0,0.245)$ and let $\kappa$ lie in an admissible domain.
Any EF learning scheme succeeding in the learning task with $2/3$ success probability must use
\begin{equation}    
N=\Omega 
( \epsilon^{-2}  
\exp\left[-2\kappa' nf(S)\right] 
\left(1+0.99\kappa'\right)^n
 ),
\end{equation}
where 
$f(S) \equiv \frac{S}{1+\sqrt{1-S^2}}$, and $\kappa' =\kappa'(\kappa, n,S,\epsilon)$ (see Sec.~S10~\cite{supple}).
\end{theorem}

\begin{figure}[t]
\centering
\includegraphics[width=0.38\textwidth]{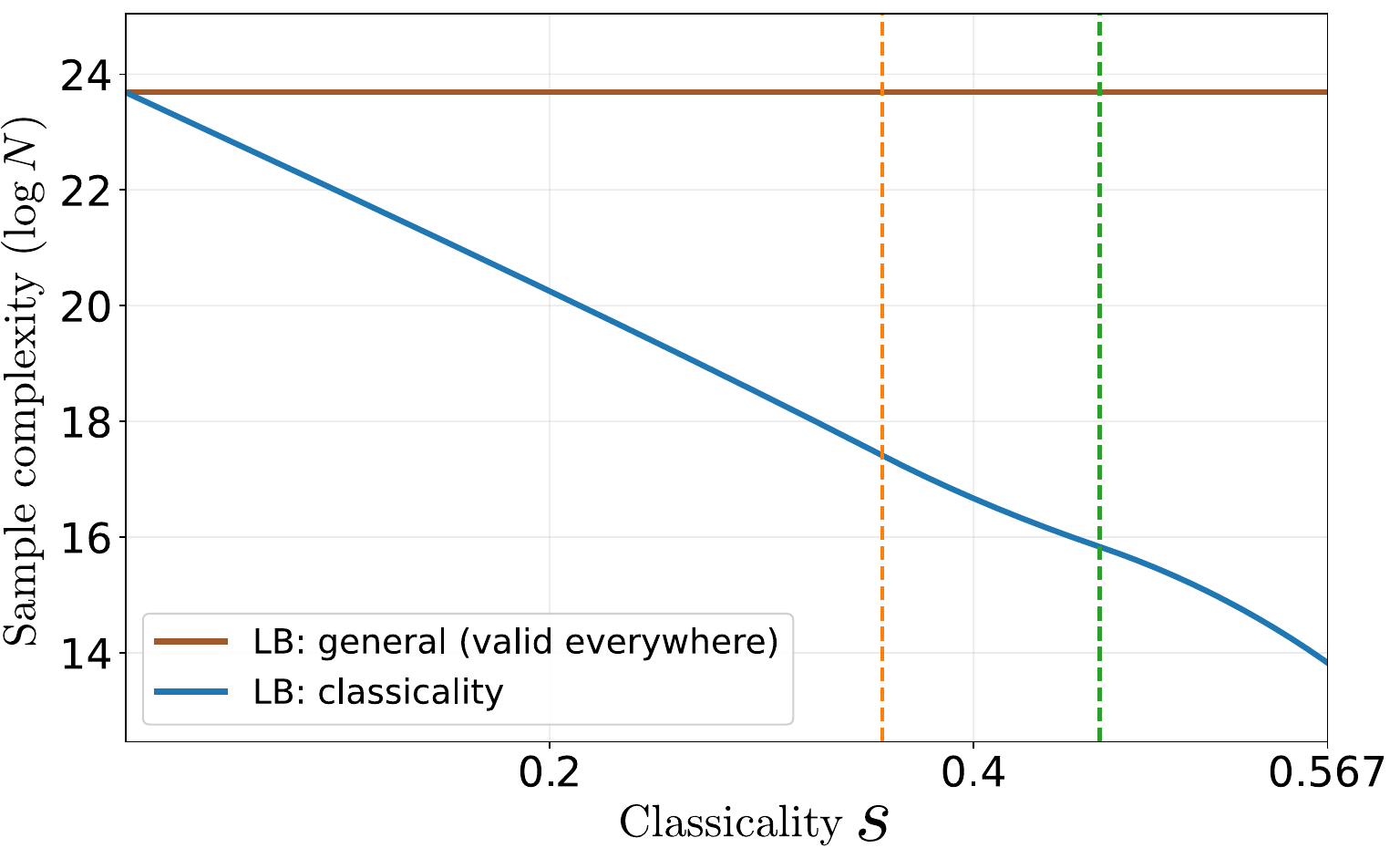}
\caption{
Sample complexity lower bound as a function of the input‑state's minimum classicality, $S$.
Parameters: $n=10$, $\epsilon=10^{-3}$, $\kappa=1.7$.}
\label{fig::phase-transition}
\end{figure}

We use a similar revealed‑hypothesis testing task as in the proof of Thm.~\ref{thm::scaling-single-copy}, with the additional promise that, if Alice selects the second hypothesis, all states are at least $S$‑classical.
Note that the admissible domain of $\kappa$ is determined by $S,n,$ and $\epsilon$ and may be empty~(see SM Sec.~S10~\cite{supple}).
%This construction holds when $\kappa$ lies in the domain $[\kappa_{\text{min}}(S), \kappa_{\text{max}}(S,n,\epsilon)]$, which is non-empty for $S \leq s_{\text{cap}}(n,\epsilon)$ (see SM Sec.~S10~\cite{supple}).
% Since $f(S)$ increases monotonically with $S$, the base $(1+0.99\kappa')e^{-2\kappa' f(S)}$ and hence the lower bound decrease as classicality increases; this effect is visible in Fig.~\ref{fig::phase-transition}.
Since $f(S)$ increases monotonically with $S$ and $\kappa'$ is chosen as in Sec.~S10~\cite{supple}, one can show that the effective base $(1+0.99\kappa')e^{-2\kappa' f(S)}$ monotonically decreases in $S$, and hence the lower bound decreases as the input classicality increases; this behavior is illustrated in Fig.~\ref{fig::phase-transition}.
% Thus, when the input has known positive classicality ($s_{\text{max}}>0$), Thms.~\ref{thm:heterodyne-classicality} and \ref{thm::scaling-single-copy2-classicality} imply that the same learning guarantees are achievable with \emph{lower} sample complexity.
% By contrast, for unknown or non‑positive classicality ($s_{\text{max}}\leq 0$), the guarantees revert to Thm.~\ref{thm::scaling-single-copy} and Eq.~\eqref{eq::heterodyne-bound}.
Thus, known positive classicality ($s_{\text{max}}>0$) allows for lower sample complexity (Thms.~\ref{thm:heterodyne-classicality} and \ref{thm::scaling-single-copy2-classicality}).
By contrast, for non-positive classicality ($s_{\text{max}}\leq 0$), refining complexity via $s_{\text{max}}$ remains an open question; hence, we revert to the general bounds of Thm.~\ref{thm::scaling-single-copy} and Eq.~\eqref{eq::heterodyne-bound}.
Taken together with our EF lower bounds, these results show that in the genuinely nonclassical regime the \textit{only} way to reduce sample complexity while maintaining learning guarantees is to provide both entanglement and reflected states (Thm.~\ref{thm:BM}).

\textit{Conclusion.{\textemdash}}
We established a rigorous exponential quantum advantage in learning CV quantum states that jointly relies on reflected states and entangled measurements, thereby identifying these two resources as precisely what is needed to overcome the exponential sample cost in generic bosonic state learning.
% We established a rigorous exponential quantum advantage in learning CV quantum states that jointly relies on reflected states and entangled measurements.
Since learning high-dimensional CV quantum states is a key primitive in many protocols, we anticipate that the BM scheme can be employed in various applications, such as quantum state certification~\cite{aolita2015reliable, chabaud2019building, eisert2020quantum,wu2021efficient,da2011practical}, Wigner function reconstruction~\cite{weinbub2018recent,buvzek1996reconstruction}, and shadow estimation~\cite{huang2020predicting,chen2021robust,becker2024classical,gandhari2024precision}.
Two natural directions for future work are tightening the EF lower bounds to close the gap to heterodyne-based upper bounds, and analyzing how noise and experimental imperfections affect the separation in practice, to guide experimental demonstrations of the exponential quantum advantage.
% A few intriguing open questions remain to be investigated. 
% First, the current lower bounds for EF schemes may not be tight, motivating sharper bounds to close the gap to HD upper bounds. 
% Second, analyzing the impact of noise and experimental imperfections will clarify the practical robustness of our results and guide experimental demonstrations of the exponential quantum advantage.
% First, the current lower bounds for EF schemes may not be tight, warranting further analysis to establish sharper bounds and to close the gap to known upper bounds such as heterodyne detection.
% Additionally, analyzing the impact of noise and experimental imperfections could provide further evidence supporting the practicality of our results and inspire experimental demonstrations of the exponential quantum advantage.

\bigskip

\begin{acknowledgements}
We thank Robbie King for interesting and fruitful discussions.
This research was supported by the National Research Foundation of Korea Grants (No. RS-2024-00431768 and No. RS-2025-00515456) funded by the Korean government (Ministry of Science and ICT~(MSIT)) and the Institute of Information \& Communications Technology Planning \& Evaluation (IITP) Grants funded by the Korea government (MSIT) (No. IITP-2025-RS-2025-02283189 and IITP-2025-RS-2025-02263264).
\end{acknowledgements}

\bibliography{reference.bib}

\end{document}

\definecolor{navy}{RGB}{46,72,102}
\definecolor{pink}{RGB}{219,48,122}
\definecolor{grey}{RGB}{184,184,184}
\definecolor{yellow}{RGB}{255,192,0}
\definecolor{grey1}{RGB}{217,217,217}
\definecolor{grey2}{RGB}{166,166,166}
\definecolor{grey3}{RGB}{89,89,89}
\definecolor{red}{RGB}{255,0,0}

\preprint{APS/123-QED}

\title{Supplemental Material for ``Exponential advantage in continuous-variable quantum state learning"}
\author{Eugen Coroi}

\author{Changhun Oh}
\email{changhun0218@gmail.com}
\affiliation{Department of Physics, Korea Advanced Institute of Science and Technology, Daejeon 34141, Korea}

\maketitle
\tableofcontents
\newpage

\section{Preliminaries}\label{sec:pre}
In this section, we review the formalism of bosonic systems and introduce the key identities and definitions used throughout this work (see Refs.~\cite{ferraro2005gaussian, weedbrook2012gaussian, serafini2017quantum} for comprehensive reviews).
An $n$-mode bosonic system is defined by its annihilation and creation operators $\hat{a}\equiv(\hat{a}_1,\dots,\hat{a}_n)^\text{T}$ and $\hat{a}^\dagger\equiv (\hat{a}_1^\dagger,\dots,\hat{a}_n^\dagger)^\text{T}$, which satisfy the canonical commutation relations $[\hat{a}_i,\hat{a}_j^\dagger]=\delta_{ij}$.
The $n$-mode displacement operator is defined as $\hat{D}(\alpha)\equiv \bigotimes_{i=1}^n e^{\alpha_i \hat{a}_i^\dagger-\alpha_i^* \hat{a}_i}$, parameterized by the complex vector $\alpha\equiv (\alpha_1,\dots,\alpha_n)^\text{T}\in\mathbb{C}^n$.
The displacement operators form an orthogonal operator basis, allowing any operator $\hat{O}$ to be expanded as~\cite{cahill1969density, ferraro2005gaussian}
\begin{align}\label{eq:displacement_fourier}
    \hat{O}=\frac{1}{\pi^n}\int d^{2n}\alpha \Tr[\hat{O}\hat{D}(\alpha)]\hat{D}^\dagger(\alpha),
\end{align}
where $\chi_{\hat{O}}(\alpha) \equiv \Tr[\hat{O}\hat{D}(\alpha)]$ is called the characteristic function of the operator $\hat{O}$.
The displacement operator satisfies the following identities:
\begin{align}
    \hat{D}^\dagger(\alpha)=\hat{D}(-\alpha),~~~\hat{D}^*(\alpha)=\hat{D}(\alpha^*),~~~\hat{D}^\text{T}(\alpha)=\hat{D}(-\alpha^*),~~~\Tr[\hat{D}(\alpha)]=\pi^n \delta^{(2n)}(\alpha),  \label{eq::dis1}
    \\ 
    \hat{D}(\alpha_1)\hat{D}(\alpha_2)=\hat{D}(\alpha_1+\alpha_2)e^{(\alpha_2^\dagger \alpha_1-\alpha_1^\dagger\alpha_2)/2},~~~\hat{D}(\alpha)\hat{D}^\dagger(\beta)\hat{D}^\dagger(\alpha)=\hat{D}^\dagger(\beta)e^{\alpha^\dagger\beta-\beta^\dagger\alpha} \label{eq::dis2}.
\end{align}
We also make frequent use of the integral representation of the delta function:
\begin{align}\label{eq:delta}
    \delta^{(2n)}(\alpha)&=\frac{1}{\pi^{2n}}\int d^{2n}\beta e^{\beta^\dagger\alpha-\alpha^\dagger\beta}.
\end{align}
We define the photon number operator for mode $i$ as $\hat{n}_i=\hat{a}_i^\dagger\hat{a}_i$, with the total photon number given by $\hat{N}=\sum_{i=1}^n\hat{n}_i$.
For a system comprising $K$ copies of an $n$-mode state, we denote the total photon number across all copies as $\hat{N}_{1:K}\equiv \sum_{j=1}^K \hat{N}_j$, where $\hat{N}_j$ acts on the $j$-th copy. We omit the copy index when the context is unambiguous.

An unnormalized thermal state $\nu^{\hat{N}}$, with $\nu\in [0,1)$, admits the following expansion in the displacement basis:
\begin{align}\label{eq:thermal_char}
    \nu^{\hat{N}}
    =\int \frac{d^{2n}\alpha}{\pi^n}\Tr[\nu^{\hat{N}}\hat{D}(\alpha)]\hat{D}^\dagger(\alpha)
    =\int \frac{d^{2n}\alpha}{\pi^n(1-\nu)^n}e^{-\frac{1}{2}\frac{1+\nu}{1-\nu}|\alpha|^2}\hat{D}^\dagger(\alpha).
\end{align}
Next, we recall the beam splitter unitary $\hat{\mathcal{U}}(\theta)$ acting on two modes with annihilation operators $\hat{a}$ and $\hat{b}$:
\begin{align}
\hat{\mathcal{U}}(\theta) = \exp\left(-i\theta (\hat{a}^{\dagger}\hat{b} + \hat{b}^{\dagger}\hat{a}) \right).
\end{align}
The Heisenberg evolution of the annihilation operators is given by
\begin{align}
\hat{c} = \hat{\mathcal{U}}^{\dagger}\hat{a}\hat{\mathcal{U}} = \hat{a} \cos\theta +  \hat{b}\sin\theta, \quad
\hat{d} = \hat{\mathcal{U}}^{\dagger}\hat{b}\hat{\mathcal{U}} = \hat{b} \cos\theta - \hat{a}\sin\theta,
\end{align}
where $\hat{c}$ and $\hat{d}$ are the output mode operators.
The action on tensor products of displacement operators is
\begin{align}
\hat{\mathcal{U}}^{\dagger}(\theta) \left( \hat{D}(\alpha)\otimes \hat{D}(\beta) \right) \hat{\mathcal{U}}(\theta) = \hat{D}\left(\alpha \cos\theta+\beta \sin\theta\right) \otimes \hat{D}\left( \beta \cos\theta-\alpha \sin\theta\right).
\end{align}
For the specific case of a balanced (50:50) beam splitter ($\theta=\pi/4$), we also use the inverse relation:
\begin{align}
\hat{\mathcal{U}}\left(\pi/4\right) \left( \hat{D}(\alpha)\otimes \hat{D}(\beta) \right) \hat{\mathcal{U}}^{\dagger}\left(\pi/4\right)
= \hat{D}\left(\frac{\alpha - \beta}{\sqrt{2}}\right) \otimes
\hat{D}\left( \frac{\alpha+\beta }{\sqrt{2}}\right). \label{eq:5050bs}
\end{align}
We now formalize the learning task studied in this work.
\begin{definition}[Learning task]\label{def::learning-task}
    Let $\hro$ be an $n$-mode continuous-variable~(CV) quantum state, and fix parameters $\epsilon,\delta,\kappa>0$. A learning scheme using $N$ copies of $\hro$ is any protocol that:
    \begin{itemize}
        \item Performs (possibly joint) measurements on the state $\hro^{\otimes N}$ and records classical outcomes.
        \item Given an arbitrary query $\alpha \in \mathbb{C}^{n}$ satisfying $|\alpha|^2 \leq \kappa n$, outputs an estimate $\tilde{u}$ of the characteristic function such that $\Pr[\min_{\tau\in\{\pm1\}}|\tau \tilde{u} - \chi_{\hro}(\alpha)|\leq \epsilon] \geq 1-\delta$.
    \end{itemize}
\end{definition}
Crucially, the measurement phase must be completed before the query $\alpha$ is revealed. The goal is to extract sufficient classical information to estimate the characteristic function (up to a global sign) at any point within the phase space ball defined by $\kappa$. We analyze the sample complexity $N$ as a function of the available quantum resources (e.g., entanglement, reflected states) and the input state's degree of classicality.

To quantify classicality, we introduce the generalized $s$-ordered characteristic functions and quasiprobability distributions (QPDs).
\begin{definition}\label{def::s-ordered}
    For $s\in[-1,1]$, the $s$-ordered QPD of an $n$-mode state $\hro$, denoted $W_{\hro}(s,\beta)$, is defined via the Fourier transform of the $s$-ordered characteristic function $\chi_{\hro}(s,\alpha)$:
    \begin{align} 
        W_{\hro}(s,\beta) &= \frac{1}{\pi^{2n}} \int d^{2n}\alpha \chi_{\hro}(s,\alpha) e^{\alpha^{\dagger}\beta-\beta^{\dagger}\alpha }, \label{eq::qpds1} \\
        \chi_{\hro}(s,\alpha)& = e^{s \frac{|\alpha|^2}{2}}\Tr\left[ \hro \hat{D}(\alpha) \right] =  e^{s \frac{|\alpha|^2}{2}} \chi_{\hro}(\alpha). \label{eq::qpds2}
    \end{align}
\end{definition}
The term ``quasiprobability'' reflects that $W_{\hro}(s,\beta)$ may be negative or singular. The parameter $s$ interpolates between key distributions: $s=0$ yields the Wigner function, $s=-1$ the Husimi $Q$-function, and $s=1$ the Glauber-Sudarshan $P$-function. The $P$ function underlies the standard notion of classicality, a fact we formalize below.

\begin{definition}[Classical state]\label{def::classical}
An $n$-mode CV quantum state $\hro$ is \emph{classical} (i.e., a mixture of coherent states) if its Glauber--Sudarshan $P$-function is non-negative for all $\beta\in\mathbb{C}^n$ and is not more singular than a delta function. Equivalently, $P_{\hro}(\beta)=W_{\hro}(1,\beta)\geq0$ and is no more singular than $\delta$. 
\end{definition}

Beyond the classical regime, we would like to quantify the degree of nonclassicality. Common indicators include the negativity of the Wigner function~\cite{kenfack2004negativity}, the Mandel $Q$ parameter~\cite{mandel1979sub}, and measures of non-Gaussianity~\cite{albarelli2018resource,takagi2018convex}; in what follows we adopt the \emph{nonclassical depth} $\tau_m \equiv \tfrac12(1-s_{\text{max}})$~\cite{lee1991measure}.

\begin{definition}\label{def::nonclassicality}
    The classicality parameter $s_{\text{max}}$ of a state $\hro$ is the largest value $s \in [-1,1]$ such that $W_{\hro}(s,\beta) \geq 0$ for all $\beta \in \mathbb{C}^n$.
\end{definition}

Finally, we present a key technical lemma stating that states with positive classicality exhibit Gaussian decay in their characteristic functions. This tail bound allows us to truncate the effective phase space in our sample complexity proofs (Secs.~\ref{sec::Channel-State-Equivalence} and~\ref{sec::classicality-section}).

\begin{lemma}[Tail bound from $s$-ordered classicality]\label{lem:tail}
Let $\hat\rho$ be an $n$-mode state for which the $s$-ordered QPD is a valid probability distribution for some $s\in(0,1]$. Then~\cite{barnett2002methods}:
\begin{equation}
|\chi_{\hat\rho}(\alpha)|\leq e^{-\frac{s}{2}|\alpha|^2}\qquad\forall\,\alpha\in\mathbb{C}^n.
\end{equation}
\begin{proof}
Substituting the definition of the $s$-ordered characteristic function into Eq.~\eqref{eq::qpds2} and taking the absolute value yields:
\begin{align}
|\chi_{\hro}(\alpha)|
= e^{-\frac{s}{2}|\alpha|^2}\left|\int d^{2n}\beta\, W_{\hro}(s,\beta)e^{\beta^{\dagger}\alpha-\alpha^{\dagger}\beta}\right|
\leq e^{-\frac{s}{2}|\alpha|^2} \int d^{2n}\beta\, |W_{\hro}(s,\beta)|.
\end{align}
Since $W_{\hro}(s,\beta)$ is non-negative and normalized to unity (being a valid probability measure), the integral equals 1, yielding the bound.
\end{proof}
\end{lemma}

\section{Reflected states}\label{sec::reflected-states}
Consider an $n$-mode quantum state $\hat{\rho}$.
The quantum state admits an expansion in displacement operator basis, as presented in Sec.~\ref{sec:pre}:
\begin{align}
    \hat{\rho}
    =\frac{1}{\pi^n}\int d^{2n}\alpha \chi_{\hat{\rho}}(\alpha)\hat{D}^\dagger(\alpha),
\end{align}
where  $\chi_{\hat{\rho}}(\alpha)=\Tr[\hat{\rho}\hat{D}(\alpha)]$ is the characteristic function of the state.

We first define the complex conjugate state of $\hat{\rho}$, where the conjugation is applied in the photon number basis, and then generalize to reflected state.
Using the displacement operator basis, the complex-conjugate state expands as
\begin{align}
    \hat{\rho}^*
    =\frac{1}{\pi^n}\int d^{2n}\alpha \Tr[\hat{\rho}\hat{D}(\alpha)]^*\hat{D}^{\T}(\alpha)
    =\frac{1}{\pi^n}\int d^{2n}\alpha \Tr[\hat{\rho}\hat{D}(-\alpha^*)]\hat{D}^{\dagger}(\alpha)
    =\frac{1}{\pi^n}\int d^{2n}\alpha \chi_{\hat{\rho}^*}(\alpha)\hat{D}^{\dagger}(\alpha),
\end{align}
so its characteristic function is 
\begin{align}\label{eq::conjugate_chi}
    \chi_{\hat{\rho}^*}(\alpha)=\chi_{\hat{\rho}}(-\alpha^*).
\end{align}

Inspired by Ref.~\cite{wu2023quantum}, we define the reflected state $\hat{\rho}^{(U)}$ as a generalization of the complex conjugate state.
Let  $U\in \mathbb{C}^{n\times n}$ be a symmetric unitary matrix ($U^{-1} = U^{\dagger} = U^{*}$).
Then, for all $\alpha \in \mathbb{C}^{n}$, the reflected state's characteristic function takes the following form:
\begin{align}\label{eq::reflected}
    \chi_{\hat{\rho}^{(U)}}(\alpha) = \chi_{\hat{\rho}}(U\alpha^*).
\end{align}

A simple example of a reflected state is the complex conjugate state, for which $U=-I$. 
From Eq.~\eqref{eq::reflected}, it follows that conjugation corresponds to $n$ reflections over the imaginary axes in phase space ($\Re(\alpha)\rightarrow-\Re(\alpha), \Im(\alpha)\rightarrow\Im(\alpha)$). 
More generally, to allow reflections about arbitrary sets of orthogonal axes in phase space, we use a symmetric unitary $U$ and define $\hror$ accordingly.
The condition $UU^*=I$ ensures that two consecutive reflections about the same axes act as identity on phase space.

To further understand reflections in phase space, factorize the symmetric unitary $U$ via the Takagi decomposition $U=VV^{\T}$ with $V$ unitary.
The reflection process proceeds as follows:\\
$1.$ The unitary matrix $V$ actively acts on $\alpha$, transforming it as $\alpha\rightarrow V^{\dagger}\alpha$.\\
$2.$ The new vector is reflected across all real axes: $V^{\dagger}\alpha\rightarrow V^{\T}\alpha^*$.\\
$3.$ The initial transformation is inverted: $ V^{\T}\alpha^* \rightarrow VV^{\T}\alpha^* = U\alpha^*$. \\
The same reflection can be understood passively as a change of basis for the axes of reflection.

If a state $\hro$ is equal to its reflected state $\hror$ for some symmetric unitary $U$, we say that $\hro$ has reflection symmetry about the corresponding axes in phase space. Equivalently, $\chi_{\hro}(\alpha)=\chi_{\hro}(U\alpha^{*})$.
Lastly, note that for a given matrix $U$, there exists a passive unitary operator $\hat{\mathcal{U}}$, which can be implemented by beam splitters and phase shifters such that
\begin{align}\label{eq:LON}
    \chi_{\hat{\mathcal{U}}\hat{\rho}\hat{\mathcal{U}}^\dagger}(\alpha) = \chi_{\hat{\rho}} (U^{\T}\alpha),
\end{align}
which follows directly from $\hat{\mathcal{U}}^{\dagger}\hat{D}(\alpha)\hat{\mathcal{U}} = \hat{D}(U^{\T} \alpha)$ the definition of characteristic function.

\iffalse
\cor{Is learning reflection axis difficult?\\
Say one has $\hro$ and $\hror$, and an oracle that can learn $\tilde{U}$ such that $|\tilde{U}-U|_{\infty}\leq \epsilon$, then one could apply $\mathcal{\tilde{U}}$ on the reflected state and obtain $\mathcal{\tilde{U}}\hror\mathcal{\tilde{U}}^{\dagger}$ which should allow for learning of $\chi_{\hro}^{2}(\alpha)$ with some extra cost because of discrepancy between $U$ and $\tilde{U}$. 

Next, we can say that if one has such oracle, then they can distinguish very well between two reflected state $\hro^{(U_1)}$ and $\hro^{(U_2)}$, under some assumptions (the unitaries are different enough and such things). By bounding TVD, we can get sample complexity lower bound for learning the reflection unitary $U$ within some $\epsilon$-error.

If learning reflection axis is exponentially hard, then it does not help our result. If it is constant, then that is a lower bound. If we can make it tight, it becomes meaningful. But, I expect lower bound is exponentially large in $n$. So the question becomes, should we pursue this question?
}
\fi

\section{Bell-measurement scheme}\label{sec::Appendix_BM}
In this section we present a concrete learning scheme for the task in Def.~\ref{def::learning-task}.
Given multiple copies of an $n$-mode CV state $\hat{\rho}$ together with an associated resource state (e.g., its reflected state $\hror$), the scheme performs only two-copy joint measurements and outputs classical data.
The data is post-processed to estimate the characteristic function (up to a sign) at an arbitrary set of $M$ phase space points $\mathcal{S}=\{\alpha_i \in \mathbb{C}^{n}\}_{i=1}^{M}$ chosen \emph{after} the measurements.
Concretely, the processed outcomes produce estimates $\{\hat{u}_i\}_{i=1}^{M}$ such that
$\min_{\tau\in\{\pm 1\}}|\tau \hat{u}_i - \chi_{\hro}(\alpha_i)| \leq \epsilon$ with high probability.
The protocol supports any $M\geq1$ and does not assume a phase space truncation ($|\alpha|^2\nleq \kappa n$).

This differs from traditional estimation methods (e.g., homodyne tomography), which must be designed around a predetermined set $\mathcal{S}$.
Here, $\mathcal{S}$ can be specified \emph{after} collecting measurement outcomes, aligning with the notion of quantum \emph{learning} rather than property estimation.

We now detail the Bell-measurement~(BM) scheme designed to learn $\chi^{2}_{\hat{\rho}}(\alpha)$ for arbitrary CV quantum states. 
Our protocol is inspired by and very similar to the ones in Refs.~\cite{wu2023quantum, king2024exponential}. 
In particular, we generalize the concept of complex conjugate state to the reflected state which serves as a resource in the protocol. 

The scheme proceeds as follows. 
We first fix an $n$-mode linear-optical circuit $\hat{\mathcal{U}}$ implementing the unitary $U$ from Eq.~\eqref{eq:LON}.
Although $\hat{\rho}$ is unknown, $U$ is assumed given so that the circuit can be constructed.
Any $n$-mode linear-optical unitary can be realized with at most ${n\choose 2}$ beam splitters~\cite{reck1994experimental}, hence the circuit is efficient.
We then send the reflected state $\hror$ through $\hat{\mathcal{U}}$, and perform a CV BM on $\hro$ and $\hat{\mathcal{U}}\hror\hat{\mathcal{U}}^{\dagger}$.

\subsection{Learning the characteristic function up to an unknown $\pm$ sign}
\begin{theorem}\label{SM-thm:BM-theorem}
For any $n$-mode quantum state $\hat{\rho}$, there exists a learning scheme that, given $N$ copies of $\hro\otimes\hror$, outputs classical information enabling accurate estimation of the squared characteristic function $\chi^2_{\hro}(\alpha_i)$ at $M$ arbitrary points $\{\alpha_{i}\}_{i=1}^{M}$.
The estimates $\{\hat{v}_{i}\}_{i=1}^{M}$ satisfy $|\hat{v}_{i}-\chi^2_{\hro}(\alpha_i)|\leq \epsilon$ with success probability at least $1-\delta$, using
$N = O\left(\epsilon^{-2} \log(M/\delta)\right)$ copies.
\end{theorem}

\begin{proof}
The idea is to use the BM whose POVM elements are given by $\{\hat{\Pi}(\zeta)\}_{\zeta\in\mathbb{C}^{n}}$~\cite{ferraro2005gaussian, serafini2017quantum}, where
\begin{align} \label{eq::bell-basis}
    \hat{\Pi}(\zeta)
    \equiv \frac{1}{\pi^n}(I\otimes \hat{D}(\zeta))|\Psi\rangle\langle \Psi|(I\otimes \hat{D}^\dagger(\zeta))
    =\int \frac{d^{2n}\alpha}{\pi^{2n}} e^{\zeta\cdot \alpha-\zeta^* \cdot \alpha^*}\hat{D}(\alpha)\otimes \hat{D}(\alpha^*).
\end{align}
Here, $|\Psi\rangle =\sum_{k=0}^\infty |k\rangle|k\rangle$ is the product of $n$ two-mode infinitely squeezed vacuum states,  which is formally written as~\cite{ferraro2005gaussian}
\begin{align}
    |\Psi\rangle\langle \Psi|
    \equiv \int \frac{d^{2n}\alpha}{\pi^{n}} \hat{D}(\alpha)\otimes \hat{D}(\alpha^*).
\end{align}

Although the POVM elements in Eq.~$\ref{eq::bell-basis}$ are formally written in terms of projectors onto idealized CV Bell states, this does not imply the physical realization of infinite-energy states. 
In practice, a CV BM is implemented by sending any two-mode state $\hro_{AB}$ through a $50:50$ beam-splitter, followed by homodyne detection of orthogonal quadratures at the output ports. 
This measurement is experimentally feasible and has been widely used in CV teleportation protocols~\cite{furusawa1998unconditional,braunstein1998teleportation}.
After sending the reflected state through the linear optical circuit $\hat{\mathcal{U}}$ corresponding to the unitary transformation $U$ as seen in Eq.~\eqref{eq:LON}, we perform the BM on $\hat{\rho}$ and the resultant state $\hat{\mathcal{U}}\hror\hat{\mathcal{U}}^{\dagger}$.
The output probability of the BM is then written as:
\begin{align}
    p(\zeta)
    &=\int \frac{d^{2n}\alpha}{\pi^{2n}} e^{\zeta\cdot\alpha - \zeta^*\cdot\alpha^*}\Tr[(\hat{\rho}\otimes \hat{\mathcal{U}}\hat{\rho}^{(U)}\hat{\mathcal{U}}^{\dagger})(\hat{D}(\alpha)\otimes \hat{D}(\alpha^*))] \\
    %&=\int \frac{d^{2n}\alpha}{\pi^{2n}} e^{\zeta\cdot\alpha - \zeta^*\cdot\alpha^*}\Tr[\hat{\rho}\hat{D}(\alpha)]\Tr[\hat{\mathcal{U}}\hat{\rho}^{(U)}\hat{\mathcal{U}}^{\dagger}\hat{D}(\alpha^*)] \\
    &=\int \frac{d^{2n}\alpha}{\pi^{2n}} e^{\zeta\cdot\alpha - \zeta^*\cdot\alpha^*}\chi_{\hat{\rho}}(\alpha)\chi_{\hat{\mathcal{U}}\hat{\rho}^{(U)}\hat{\mathcal{U}}^\dagger}(\alpha^*) \\ 
    &=\int \frac{d^{2n}\alpha}{\pi^{2n}} e^{\zeta\cdot\alpha - \zeta^*\cdot\alpha^*}\chi_{\hat{\rho}}^2(\alpha).
\end{align}
For the last equality, we utilized the definition of the reflected state, Eq.~\eqref{eq::reflected}, and the property of the linear optical circuit $\hat{\mathcal{U}}$ from Eq.~\eqref{eq:LON}.
By taking the Fourier transform of the above expression, we obtain
\begin{align}
    \chi^2_{\hat{\rho}}(\alpha)=\int d^{2n}\zeta~p(\zeta)e^{-(\zeta\cdot\alpha - \zeta^*\cdot\alpha^*)}.
\end{align}
Hence, if we take $\hat{v}_i\equiv \sum_{j=1}^N e^{-\bigl(\zeta^{(j)}\cdot \alpha_i -(\zeta^{(j)})^* \cdot\alpha_i^*\bigr)}/N$ to be an estimator for $N$ samples $\{\zeta^{(j)}\}_{j=1}^N$, it is guaranteed to be unbiased. 
More precisely, when we have access to $\hat{\rho}$ and $\hat{\rho}^{(U)}$, the above BM scheme and the unbiased estimator allow us to estimate $\chi_{\hat{\rho}}^{2}(\alpha)$ with precision $\epsilon$ and probability $1-\delta$ using $N=O(\epsilon^{-2}\log1/\delta)$ samples. Using the union bound, we can extend this result to $M$ points of $\chi_{\hat{\rho}}^{2}(\alpha)$ with sample complexity $N=O(\epsilon^{-2}\log(M/\delta))$.
\end{proof}
We now extend this result to learning the characteristic function up to an unknown sign, $\pm\chi_{\hro}(\alpha)$.
\begin{corollary}\label{cor:BM-corollary}
For any $n$-mode state $\hat{\rho}$, there is a scheme that, given $N$ copies of $\hro\otimes\hror$, outputs classical data from which one can estimate $\chi_{\hro}(\alpha_i)$, up to a global $\pm$ sign, at $M$ arbitrary phase space points chosen post-measurement.
The estimates $\{\hat{u}_{i}\}_{i=1}^{M}$ satisfy
$\underset{\tau\in \{ \pm1\}}{\min}|\tau\hat{u}_{i}-\chi_{\hro}(\alpha_i)|\leq \epsilon$
with probability at least $1-\delta$, using
$N = O\left(\epsilon^{-4} \log(M/\delta)\right)$ copies.
\end{corollary}
\begin{proof}
    Assuming the BM scheme succeeds $\left( \left|\hat{v}_i - \chi^2_{\hro}(\alpha_i) \right|\leq \epsilon^2/3, \forall i=1\ldots M \right)$, we construct estimators $\{\hat{u}_{i}\}_{i=1}^{M}$ for the characteristic function at $M$ arbitrary phase space points, achieving an additive error at most $\epsilon$, up to an unknown global sign factor.
    We consider two cases based on the magnitude of $\hat{v}_{i}$, the estimator of $\chi_{\hro}^2(\alpha_i)$.
    If $|\hat{v}_i|\leq 2/3 \epsilon^2$, then assign $\hat{u}_i = 0$ (case 1). 
    Otherwise, $\hat{u}_i = \sqrt{\hat{v}_i}$ is the principal square root of $\hat{v}_{i}$ (case 2). 
    In the following, we show that $\hat{u}_i$ has at most an additive error $\epsilon$.
    
\textbf{Case 1: $|\hat{v}_i|\leq \frac{2}{3}\epsilon^2$}. Output $\hat{u}_i = 0$.
\begin{align}
    \left| \chi_{\hat{\rho}}^{2}(\alpha_i) \right|
    \leq | \hat{v}_i |+ \left| \chi_{\hat{\rho}}^{2}(\alpha_i) - \hat{v}_i \right|
    \leq \frac{2}{3}\epsilon^2+\frac{1}{3}\epsilon^2
    =\epsilon^2.
\end{align}
Since our estimator $\hat{u}_i = 0$:
\begin{align}
     \left|\hat{u}_i - \chi_{\hat{\rho}}(\alpha_i) \right| \leq \epsilon.
\end{align}
\textbf{Case 2: $|\hat{v}_i| > \frac{2}{3}\epsilon^2$.} Output the principal square root $\hat{u}_i = \sqrt{\hat{v}_i}$.
Without loss of generality, we assume that $\hat{u}_i$ lies in the same half-plane (shares the same complex phase direction) as $\chi_{\hro}(\alpha_i)$. 
If this assumption does not hold, simply replace $\hat{u}_i$ with $-\hat{u}_{i}$.
The estimator satisfies:
\begin{align}
    \left| \hat{u}_i - \chi_{\hro}(\alpha_i) \right| 
    = \frac{\left| \hat{v}_i - \chi^2_{\hro}(\alpha_i) \right|}{\left| \hat{u}_i + \chi_{\hro}(\alpha_i) \right|}
    \leq \frac{\left| \hat{v}_i - \chi^2_{\hro}(\alpha_i) \right|}{|\hat{u}_i|}
    \leq \frac{\epsilon}{\sqrt{6}},
\end{align}
since $|\hat{u}_{i} + \chi_{\hro}(\alpha_i)| \geq |\hat{u}_{i}| \geq \sqrt{\tfrac{2}{3}}\epsilon$ holds under our assumption for $\hat{u}_{i}$.
Consequently, $\hat{u}_{i}$ approximates $\chi_{\hro}(\alpha_i)$ within an error of $\epsilon$, modulo a sign factor.
\end{proof}

\subsection{Sign ambiguity}
In this subsection, we review the literature and analyze prior solutions to the sign ambiguity problem in both the DV and CV settings. 
We highlight the fundamental challenges preventing the direct application of these techniques to the efficient learning of CV quantum states.

Our learning scheme exploits that the set of operators $\{\hat{D}(\alpha)\otimes\hat{D}(\alpha^{*})\}_{\alpha\in\mathbb{C}^{n}}$ is a commuting set, and that the CV BM POVM elements, being weighted integrals over this commuting set, enable simultaneous measurement (see Sec.~\ref{sec::Channel-State-Equivalence} for more details). 
In direct analogy, the tensored Pauli operators $\{\hat{P}_{i}\otimes \hat{P}_{i}\}_{i\in\mathbb{Z}_2^{2n}}$ form a commuting set in the DV $n$-qubit case, and so do the discrete displacement operators $\{\hat{D}_{q,p}\otimes\hat{D}_{-q,p}\}_{(q,p)\in \mathbb{Z}_d^{2}}$ in the $d$-dimensional DV case. 
The general strategy is to expand the Hilbert space so that one can \emph{simultaneously} measure a commuting family (e.g., $\{\hat{D}(\alpha)\otimes\hat{D}(\alpha^{*})\}$) via an appropriate \textit{Bell} measurement; this approach was popularized in Ref.~\cite{huang2021information}, where the goal was to learn Pauli expectation values $\{\Tr[\hro \hat{P}_{i}]\}_{i}$. 
Because learning characteristic functions, or Pauli expectations, \emph{up to a sign}, can be done efficiently, this strategy has become standard. 
However, going a step further to also resolve the \emph{signs} within a bona fide learning (as opposed to property-estimation) protocol has remained challenging.

For recovering the signs of $M$ non-negligible Pauli expectation values in an $n$-qubit system, Ref.~\cite{huang2021information} performs $M$ consecutive two-outcome (gentle) measurements. 
Intuitively, these sequential measurements disturb the state, $\hro^{\otimes O(\log M)}$, only mildly and thereby reveal $O(M)$ signs using just $O(\log M)$ copies. 
Building on this idea, Ref.~\cite{wu2023quantum} adapts a similar strategy to recover the signs of the CV characteristic function, utilizing joint measurements across $O(\log M)$ copies of $\hro$ stored in quantum memory. 
Crucially, the settings for the two-outcome measurements can only be fixed \emph{after} the set $\mathcal{S}$ of interest (the $M$ Pauli strings or $M$ phase space points) is specified. This dependency violates the defining requirement of a learning protocol—that measurements must precede the queries—effectively reducing the task to property estimation rather than true learning.

A different recent approach, proposed in Ref.~\cite{king2024exponential} for learning displacement amplitudes in $d$-dimensional DV systems, introduces a reference, or hypothesis, state. 
After first learning the target amplitudes up to sign at candidate points, the algorithm selects the $M$ points with non-negligible values and then engineers a reference state $\tilde{\rho}$ whose displacement amplitudes are known (with signs) and closely match the true magnitudes at those points. 
Crucially, this state-engineering subroutine consumes $O(\epsilon^{-4}\log d)$ copies of $\hro$. 
By performing a BM on $\hro\otimes \tilde{\rho}^{*}$, one can check whether $\tilde{y}_{q,p}=\Tr[\tilde{\rho}\hat{D}_{q,p}]$ has the same or opposite sign as $y_{q,p}=\Tr[\hro\hat{D}_{q,p}]$, thereby resolving the sign ambiguity.

The fundamental limitation of this method lies in the nature of the bosonic phase space.
In the bosonic limit $d\to\infty$, the operator basis becomes continuous and the phase space unbounded. Consequently, directly applying this approach to the CV setting is unfeasible, as engineering a reference $\tilde{\rho}$ that maintains sufficient overlap with the target state across arbitrary post-selected points of interest would effectively require covering the entire relevant phase space. 
Engineering such a ``universal" reference without incurring a sample overhead that diverges in the CV limit (scaling with $\log d$) remains an unsolved problem.
Consequently, resolving the sign ambiguity efficiently for general CV states within a bona fide learning framework remains an open challenge.

\section{Relation and distinctness of state and channel learning protocols} \label{sec::Channel-State-Equivalence}

In this section, we provide a unifying perspective connecting the efficient CV state-learning protocol introduced in Sec.~\ref{sec::Appendix_BM} to the random-displacement channel-learning protocol of Ref.~\cite{oh2024entanglement}. 
Specifically, we demonstrate that both schemes are special instances of measuring commuting operators in the CV Bell basis. 
Furthermore, we elucidate the relationship between learning a channel's characteristic function and learning that of a state via the Choi–Jamiołkowski~(CJ) isomorphism, thereby clarifying when channel-learning protocols can and cannot simulate state learning.

The core observation is that the operators $\{\hat{D}(\alpha)\otimes\hat{D}(\alpha^*)\}_{\alpha\in\mathbb{C}^n}$ form a commuting set, allowing their expectation values to be estimated simultaneously.
The displaced, infinitely-squeezed \textit{two-mode squeezed vacuum}~(TMSV) vectors $\{\ket{\Psi(\zeta)}=I\otimes \hat{D}(\zeta)\ket{\Psi}\}_{\zeta\in\mathbb{C}^n}$ constitute a non-normalized eigenbasis for this set.
This follows from identity~Eq.~\eqref{eq::dis2} and the transpose trick, noting that $\ket{\Psi}=\sum_{j=0}^\infty \ket{j,j}$ is an unnormalized maximally entangled state.
This structure motivates the use of the CV BM, defined by the POVM elements in Eq.~\eqref{eq::bell-basis}.
Leveraging this commuting structure, we construct an efficient estimator for the characteristic function of an arbitrary $2n$-mode state $\hro_{AB}$ at phase space points of the form $(\alpha,\alpha^*)$.

To be more specific, any $2n$-mode quantum state admits the expansion in terms of displacement operators as
\begin{align}
    \hro_{AB}
    = &\frac{1}{\pi^{2n}} \int d^{2n}\alpha d^{2n}\beta\ \chi_{\hro_{AB}}(\alpha,\beta) \hat{D}^\dagger(\alpha)\otimes \hat{D}^\dagger(\beta) \\
    = &\frac{1}{\pi^{2n}} \int d^{2n}\alpha \chi_{\hro_{AB}}(\alpha,\alpha^*) \hat{D}^\dagger(\alpha)\otimes \hat{D}^\dagger(\alpha^*) \\
     + &\frac{1}{\pi^{2n}} \int d^{2n}\alpha d^{2n}\beta\ \bigl(1-\delta^{(2n)}(\alpha-\beta^*)\bigr) \chi_{\hro_{AB}}(\alpha,\beta) \hat{D}^\dagger(\alpha)\otimes \hat{D}^\dagger(\beta).
\end{align}
Under a BM, only the first term contributes to the outcome probability density:
\begin{align}  
p(\zeta) = \Tr\bigl[\hat{\Pi}(\zeta)\hro_{AB}\bigr]
= \frac{1}{\pi^{2n}} \int d^{2n}\alpha\ e^{\zeta\cdot\alpha - \zeta^*\cdot\alpha^*} \chi_{\hro_{AB}}(\alpha,\alpha^*).
\end{align}
Inverting the Fourier transform yields
\begin{align}
\chi_{\hro_{AB}}(\alpha,\alpha^*) = \int d^{2n}\zeta\ e^{\zeta^*\cdot\alpha^*-\zeta\cdot\alpha} p(\zeta).
\end{align}
We define the estimator
\begin{align}
\tilde{\chi}_{\hro_{AB}}(\alpha,\alpha^*) \equiv \frac{1}{N}\sum_{j=1}^{N} e^{\zeta_j^*\cdot\alpha^*-\zeta_j\cdot\alpha},
\end{align}
which is unbiased. Applying Hoeffding's inequality to the real and imaginary parts yields the following result:

\begin{theorem}\label{thm::thm-two}
For any $2n$-mode CV state, one can learn $\chi_{\hro_{AB}}(\alpha,\alpha^*)$ at $M$ arbitrary points $\alpha\in\mathbb{C}^n$ such that
$|\tilde{\chi}_{\hro_{AB}}(\alpha,\alpha^*)-\chi_{\hro_{AB}}(\alpha,\alpha^*)|\leq \epsilon$
holds with probability at least $1-\delta$, using
$N=O\left(\epsilon^{-2}\log(M/\delta)\right)$ copies.
\end{theorem}
We now apply this theorem to two scenarios of interest: the efficient state-learning protocol and the channel-learning protocol.

\paragraph{Recovering the BM state-learning scheme.}
Let $\hro_{AB}=\hro\otimes \hat{\mathcal{U}}\hror\hat{\mathcal{U}}^{\dagger}$. Then, the state expands as
\begin{align}
\hro_{AB}
= \frac{1}{\pi^{2n}} \int d^{2n}\alpha d^{2n}\beta\ \chi_{\hro}(\alpha)\chi_{\hror}(\beta) \hat{D}^\dagger(\alpha)\otimes \hat{D}^\dagger(U^{\T}\beta).
\end{align}
Consequently, the joint characteristic function evaluated at $(\alpha, \alpha^*)$ becomes
\begin{align}
\chi_{\hro_{AB}}(\alpha,\alpha^*)=\chi_{\hro}(\alpha)\chi_{\hror}(U^{\dagger}\alpha^*)=\chi_{\hro}^2(\alpha)
\end{align}
(see Sec.~\ref{sec::reflected-states} for details on reflected states).
Thus, Thm.~\ref{thm::thm-two} directly recovers the BM protocol for CV state learning.

\paragraph{Channel learning via the Choi–Jamiołkowski isomorphism.}
We now consider the task of learning a CV random-displacement channel as introduced in Ref.~\cite{oh2024entanglement}.
That work proposed an efficient scheme utilizing $n$ TMSV pairs with squeezing parameter $r$, quantum memory, and entangled measurements.
In this protocol, half of each TMSV state is evolved by the channel while the idler half is stored in quantum memory.
The resulting joint state is then measured in the CV Bell basis.

A bosonic random-displacement channel $\Lambda$ acts as
\begin{align}\label{eq::rd-displacement-channel}
    \Lambda(\hat\rho) 
    = \int d^{2n}\alpha p(\alpha) \hat{D}(\alpha)\hat\rho \hat{D}^{\dagger}(\alpha)
    = \frac{1}{\pi^{n}}\int d^{2n}\beta \lambda(\beta) \chi_{\hat\rho}(\beta) \hat{D}^{\dagger}(\beta),
\end{align}
where $\lambda$ is the Fourier transform of the classical probability distribution $p$ and serves as the channel's characteristic function (distinct from the state characteristic function $\chi_{\hat\rho}$).

Let us now employ the CJ isomorphism for infinite-dimensional systems for a quantum channel $\Phi$~\cite{serafini2017quantum}:
\begin{align}
\hat\rho^{r}_{\Phi} &= (\mathcal{I} \otimes \Phi)\left(\ket{\Psi_{r}}\bra{\Psi_{r}}\right),\\
\Phi(\hat\rho) &= \lim_{r\to\infty} \cosh^{4}(r) \bra{\Psi_{r}}_{13} \hat\rho^{r}_{\Phi}\otimes\hat\rho \ket{\Psi_{r}}_{13},
\end{align}
where $\ket{\Psi_{r}}=\frac{1}{\cosh r}\sum_{j=0}^\infty \tanh^{j} r\ket{j,j}$ denotes the TMSV state with squeezing parameter $r$.

The channel-learning protocol of Ref.~\cite{oh2024entanglement} effectively prepares a finite-energy Choi state before measuring it in Bell basis.
Injecting $n$ TMSV pairs into $\Lambda$ yields the $2n$-mode state:
\begin{align}\label{eq:choi-state}
\hat\rho_{\Lambda}(r) 
= (\mathcal{I}\otimes \Lambda)\left(\ket{\Psi(r)}\bra{\Psi(r)}\right)
= \frac{1}{\pi^{2n}}\int d^{2n}\alpha d^{2n}\beta\ g(\alpha,\beta,r) \lambda(\beta) \hat{D}^{\dagger}(\alpha)\otimes \hat{D}^{\dagger}(\beta),
\end{align}
with
\begin{align}
g(\alpha,\beta,r) 
&= \exp\left[-\frac{e^{2r}}{4}\bigl|\alpha-\beta^{*}\bigr|^{2}\right]
\exp\left[-\frac{e^{-2r}}{4}\bigl|\alpha+\beta^{*}\bigr|^{2}\right]  \label{eq::channel-form-1} \\
&= \exp \left[-\frac{\cosh2r}{2}|\alpha|^2 \right] 
\exp \left[ -\frac{\cosh2r}{2}|\beta|^2  \right]
\exp \left[\frac{\sinh2r}{2} (  \alpha\cdot \beta + \alpha^* \cdot\beta^*)  \right] \label{eq::channel-form-2}
\end{align}
Hence, the characteristic function of the Choi state factorizes as
\begin{align}
\chi_{\hat\rho_{\Lambda}(r)}(\alpha,\beta)=g(\alpha,\beta,r)\lambda(\beta).
\end{align}
Evaluating this at the specific $(\alpha^{*},\alpha)$ point yields
\begin{align}
\chi_{\hat\rho_{\Lambda}(r)}(\alpha^{*},\alpha)=\exp\left[-e^{-2r}|\alpha|^{2}\right]\lambda(\alpha).
\end{align}
Invoking Thm.~\ref{thm::thm-two}, one can estimate $\chi_{\hat\rho_{\Lambda}(r)}(\alpha^{*},\alpha)$ with error $\epsilon$ and success probability $1-\delta$ at $M$ points using
\begin{align}
N = O\left(\epsilon^{-2}\log(M/\delta)\right)
\end{align}
copies of $\hat\rho_{\Lambda}(r)$. Constructing an estimator $\tilde\lambda(\alpha)$ for the channel’s characteristic function introduces a multiplicative factor $\exp[e^{-2r}|\alpha|^{2}]$ into the error propagation. Thus, to guarantee $|\tilde\lambda(\alpha)-\lambda(\alpha)|\leq\epsilon$, the sample complexity becomes
\begin{align}
N = O\left(\exp\bigl[2e^{-2r}|\alpha|^{2}\bigr] \epsilon^{-2}\log(M/\delta)\right).
\end{align}
This effectively recovers the result of Ref.~\cite{oh2024entanglement}, wherein the channel is learned by measuring its Choi state in the Bell basis.

\paragraph{BM state learning is not a special case of BM channel learning.}
One may ask whether the BM scheme for learning $\chi_{\hro}^{2}(\alpha)$ can be recovered by reinterpreting it as a special instance of the BM channel-learning framework described above.
At the level of resource states, however, the two schemes are structurally different.
From Eq.~\eqref{eq::channel-form-2}, the envelope $g(\alpha,\beta,r)$ that determines the Choi state characteristic function $\chi_{\hro_{\Lambda}(r)}(\alpha,\beta)$ contains a correlation term proportional to $\sinh2r$.
For any strictly positive squeezing parameter $r>0$, this cross-term prevents $g(\alpha,\beta,r)$, and hence $\chi_{\hro_{\Lambda}(r)}(\alpha,\beta)$, from factorizing into the product form $\chi_{\hro_1}(\alpha)\,\chi_{\hro_2}(\beta)$ for any pair of $n$-mode states $\hro_1,\hro_2$.
Consequently, for $r>0$ no random-displacement channel admits a product-state Choi state.
In contrast, the BM state-learning protocol applies the CV BM to the product state $\hro\otimes \hat{\mathcal{U}}\hro^{(U)}\hat{\mathcal{U}}^{\dagger}$ which can never coincide with the underlying Choi state.

Although the BM is applied on strictly distinct states ($\hro_{\Lambda}(r) \neq \hro \otimes \hror$ for any $r>0$), we investigate whether the measurement outcome statistics coincide.
After all, both efficient protocols are described by a unified BM framework.
Furthermore, the BM on a given input only probes the slice $(\alpha,\beta)=(\alpha^*,\alpha)$ of the characteristic function, $\chi_{\hro_{AB}}$, making it possible for the outcome statistics to match even if the input states do not.
We now determine the extent of this statistical overlap to answer whether one scheme can be obtained by reinterpreting the settings of the other.

\begin{theorem}\label{thm::overlap-extent}
    For any $n$-mode quantum state $\hat{\rho}$ of positive classicality $s_{\text{max}}>0$ and its reflected state $\hror$, there exists a squeezing parameter $r \geq r_{*}(\hro)$ and a random-displacement channel $\Lambda_{\hro,r}$ such that the output statistics of the BM state-learning scheme on $\hro \otimes \hror$ and the output statistics of the BM channel-learning scheme for $\Lambda_{\hro,r}$ (with TMSV input $|\Psi_{r}\rangle \langle\Psi_{r}|$) are identical.
    The threshold squeezing is given by:
    \begin{equation}
    r_{*}(\hro) \equiv \frac{1}{2}\ln\left(\frac{1}{s_{\text{max}}}\right).
\end{equation}
\end{theorem}
\begin{proof}
    From Eq.~\eqref{eq::channel-form-1}, the Choi state's characteristic function restricted to the $(\alpha^*,\alpha)$ slice is $\chi_{\hro_{\Lambda}(r)}(\alpha^{*},\alpha) = \exp \left[-e^{-2r}|\alpha|^2\right]\lambda(\alpha)$, where $\lambda$ is the channel’s characteristic function.
    If the BM channel-learning protocol is to reproduce the BM state-learning statistics for some input state $\hro$ at all query points, we must have $\chi_{\hro}^{2}(\alpha) \equiv \exp \left[-e^{-2r}|\alpha|^2\right] \lambda(\alpha)$ for all $\alpha \in \mathbb{C}^{n}$.
    Equivalently, we require:
\begin{equation}\label{eq::lambda-rho-r-def}
    \lambda_{\hro,r}(\alpha) \equiv \exp\left[e^{-2r}|\alpha|^2\right]\chi_{\hro}^{2}(\alpha),
\end{equation}
where $\lambda_{\hro,r}(\alpha)$ must be a valid classical characteristic function~\cite{ushakov1999selected}.
Recall that $s_{\text{max}}$ is the largest ordering parameter for which the $s$-ordered QPD $W_{\hro}(s,\beta)$ is a valid probability density.
For $s_{\text{max}}>0$, $W_{\hro}(s_{\text{max}},\beta)$ is non-negative and no more singular than a delta function (see Def.~\ref{def::nonclassicality}).

Define the threshold $r_{*}(\hro) \equiv \frac{1}{2}\ln(1/s_{\text{max}})$, so that for any $r\geq r_{*}(\hro)$ we have $s_{\text{max}} - e^{-2r} \geq 0$.
Consider the classical random variable $Z \equiv X_1 + X_2 + N$, where $X_1$ and $X_2$ are i.i.d. with distribution $W_{\hro}(s_{\text{max}},\beta)$, and $N$ is a zero-mean isotropic Gaussian on $\mathbb{R}^{2n}$ with covariance $(s_{\text{max}} - e^{-2r})I_{2n\times 2n}$.
The probability distribution of $Z$ is the convolution 
\begin{align}
p(\beta) \equiv \left( W_{\hro}(s_{\text{max}}) \ast W_{\hro}(s_{\text{max}}) \ast \mathcal{N}\right)(\beta), 
\end{align}
which is a valid classical probability density since $W_{\hro}(s_{\text{max}},\beta)\geq 0$ and $\mathcal{N}$ has non-negative covariance.
By independence, the characteristic function of $Z$ is:
\begin{align}\label{eq::channel-existence}
    \lambda(\alpha)
    &= \chi_{\hro}^{2}(s_{\text{max}},\alpha)
       \exp\left[-(s_{\text{max}} - e^{-2r})|\alpha|^2\right] 
    = \chi_{\hro}^{2}(\alpha)
       \exp \left[e^{-2r}|\alpha|^2\right]
       = \lambda_{\hro,r}(\alpha).
\end{align}
Hence, $\lambda_{\hro,r}$ is a classical characteristic function for all $r\geq r_{*}(\hro)$, and the BM channel-learning protocol for the random-displacement channel $\Lambda_{\hro,r}$ with TMSV input $|\Psi_{r}\rangle \langle \Psi_{r}|$ reproduces exactly the BM state-learning statistics for $\hro$. 
\end{proof}

The above theorem establishes that for any $n$-mode state of positive classicality, there exists a random-displacement channel and squeezing parameter such that both schemes yield identical outcome statistics. 
If this reduction were to hold generally for all states (including those of non-positive classicality), we could claim that the state-learning protocol is merely a reinterpretation of the channel-learning one. 
However, we now prove that this is not the case.
The factorization in Eq.~\eqref{eq::lambda-rho-r-def} relies crucially on the term $\exp[-(s_{\text{max}} - e^{-2r})|\alpha|^2]$ being a valid Gaussian characteristic function.
For the convolution argument to hold, we require the variance to be non-negative, i.e., $s_{\text{max}} - e^{-2r} \geq 0$.
In the genuinely nonclassical regime $s_{\text{max}}\leq 0$, there is no finite $r$ for which $e^{-2r}\leq s_{\text{max}}$, so the constructive reduction in Thm.~\ref{thm::overlap-extent} cannot be implemented.
This does not, by itself, rule out the existence of \emph{some} random-displacement channel realizing $\lambda_{\hro,r}$, since Eq.~\eqref{eq::lambda-rho-r-def} could in principle describe a classical characteristic function depending on the specific form of $\chi_{\hro}(s_{\text{max}},\alpha)$.
However, we now show that this possibility fails for a simple Wigner-negative state.

Non-zero Fock states have negative Wigner functions and therefore non-positive classicality ($s_{\text{max}}<0$).
Let $\hro=|1\rangle\langle1|$ (single mode) and assume, for the sake of contradiction, that there exist $r$ and $\Lambda$ such that:
\begin{align}
\lambda(\alpha)
= \exp\left[e^{-2r}|\alpha|^{2}\right]\chi_{|1\rangle}^{2}(\alpha)
= \exp(-c|\alpha|^{2})(1-|\alpha|^{2})^{2},
\end{align}
where $c\equiv 1-e^{-2r}\in[0,1)$.
In the domain $0<r<\infty$, the probability distribution corresponding to this characteristic function is:
\begin{align}
p(\beta)
= \frac{1}{\pi^{2}}\int d^{2}\alpha\, \lambda(\alpha)\, e^{\alpha^{*}\beta-\beta^{*}\alpha}
= \frac{1}{\pi c^{5}} e^{-|\beta|^{2}/c}\Bigl(|\beta|^{4}-(4c-2c^{2})|\beta|^{2}+(c^{4}-2c^{3}+2c^{2})\Bigr).
\end{align}
This distribution becomes negative in the annulus:
\begin{align}
c\left(2-c-\sqrt{2(1-c)}\right) < |\beta|^{2} < c\left(2-c+\sqrt{2(1-c)}\right).
\end{align}
The existence of this negative annulus implies that $p(\beta)$ cannot be interpreted as a probability distribution over displacement operators.
Consequently, no physical random-displacement channel $\Lambda$ exists that can reproduce the learning of the single-photon state $\ket{1}$.
While the channel-learning reduction covers any $n$-mode state of positive classicality, failure in the regime of non-positive classicality further separates the channel and state-learning schemes.
For the limit case $r=0$ (i.e., $c=0$), one can verify directly that $\lambda$ is not positive-definite; for instance, evaluating on the finite set $\{|\alpha|^{2}=0, |\alpha|^{2}=4\}$ reveals a violation of Bochner’s theorem~\cite{ushakov1999selected}.
Thus, no such channel $\Lambda$ exists for $\hro=|1\rangle\langle1|$, and the BM state-learning protocol is strictly more general than the channel-learning reduction.

Theorem~\ref{thm::overlap-extent} hinted at a potential resource equivalence: having access to a reflected state, $\hror$, in the state-learning scheme appeared statistically similar to having access to a squeezed TMSV state, $|\Psi_{r}\rangle \langle \Psi_{r}|$, in the channel-learning scheme.
However, the counter-example of the Fock state demonstrates that this statistical equivalence breaks down for states of non-positive classicality ($s_{\text{max}} \leq 0$). 
Because the channel-learning protocol is formulated specifically for random-displacement channels, it is impossible to reproduce the state-learning statistics for these highly nonclassical states within that framework. 
It remains an open question for future work whether a broader class of bosonic channels (beyond random displacements) could restore this statistical equivalence for all states, or if a fundamental resource separation persists.

\section{Input states: three-peak and five-peak states}\label{sec:input_states}

In this section we define the three-peak and five-peak states and derive their characteristic functions. These formulas are later used to obtain sample complexity lower bounds for learning schemes with restricted resources (see Secs.~\ref{sec::no-entanglement}\textendash\ref{sec::classicality-section}), and for learning schemes with access to both entanglement and reflected states (see Sec.~\ref{sec::both-resources}).

\paragraph{Three-peak state.}
We define the three-peak state as
\begin{align}
    \hat{\rho}_\gamma \equiv (1-\nu^2)^{n}\nu^{\hat{N}}\left(\hat{D}^\dagger(0)+2i\epsilon_0\hat{D}^\dagger(\gamma)-2i\epsilon_0\hat{D}^\dagger(-\gamma)\right)\nu^{\hat{N}},
\end{align}
where $0<\nu<1$ and $\hat{N}=\sum_{i=1}^{n}\hat{n}_i$. The operator is positive semidefinite because
\begin{align}
    \left|2i\epsilon_0\left(\langle \psi|\hat{D}^\dagger(\gamma)|\psi\rangle-\langle \psi|\hat{D}^\dagger(-\gamma)|\psi\rangle\right)\right|
    \leq 1
\end{align}
whenever $0<\epsilon_0\leq 1/4$, which we assume throughout.

\medskip
\noindent\textit{Characteristic function.}
Using the $\hat{D}(\alpha)\hat{D}(\beta)=\hat{D}(\alpha+\beta)e^{(\beta^\dagger\alpha-\alpha^\dagger\beta)/2}$ identity~\eqref{eq::dis2} and the $\nu^{\hat{N}}$ operator expansion
\begin{align}
    \nu^{\hat{N}}=\int \frac{d^{2n}\alpha}{\pi^{n}(1-\nu)^n}e^{-\frac{1}{2}\frac{1+\nu}{1-\nu}|\alpha|^2}\hat{D}^{\dagger}(\alpha),
\end{align}
we obtain
\begin{align}\label{eq::displacement-helper}
    \Tr\left[\nu^{\hat{N}}\hat{D}^\dagger(\gamma)\nu^{\hat{N}}\hat{D}(\alpha)\right]
    =\frac{1}{(1-\nu^2)^{n}}\exp\left[-\frac{1-\nu}{2(1+\nu)}\left(|\alpha|^2+|\gamma|^2\right)-\frac{\nu}{1-\nu^2}|\gamma-\alpha|^2\right].
\end{align}
Hence, the characteristic function of $\hat{\rho}_{\gamma}$ is
\begin{align}\label{eq::chi-three-peak}
    \chi_{\hat{\rho}_\gamma}(\alpha)
    &=e^{-\frac{1-\nu}{2(1+\nu)}|\alpha|^2}\left[e^{-\frac{\nu}{1-\nu^2}|\alpha|^2}+2i\epsilon_0 e^{-\frac{1-\nu}{2(1+\nu)}|\gamma|^2}\left(e^{-\frac{\nu}{1-\nu^2}|\gamma-\alpha|^2}-e^{-\frac{\nu}{1-\nu^2}|\gamma+\alpha|^2}\right)\right] \\ 
    &=e^{-\frac{|\alpha|^2}{2\Sigma^2}}\left[e^{-\frac{|\alpha|^2}{2\sigma^2}}+2i\epsilon_0 e^{-\frac{|\gamma|^2}{2\Sigma^2}}\left(e^{-\frac{|\gamma-\alpha|^2}{2\sigma^2}}-e^{-\frac{|\gamma+\alpha|^2}{2\sigma^2}}\right)\right] \\ 
    &= \chi_{\hat{\rho}_0}(\alpha)+\chi_{\hat{\rho}_\gamma}^{\text{(add)}}(\alpha),
\end{align}
where we defined
\begin{align}
    \chi_{\hat{\rho}_0}(\alpha)\equiv e^{-\frac{|\alpha|^2}{2\Sigma^2}}e^{-\frac{|\alpha|^2}{2\sigma^2}},\qquad
    \chi_{\hat{\rho}_\gamma}^{\text{(add)}}(\alpha)\equiv 2i\epsilon_0e^{-\frac{|\alpha|^2+|\gamma|^2}{2\Sigma^2}}\left[e^{-\frac{|\gamma-\alpha|^2}{2\sigma^2}}-e^{-\frac{|\gamma+\alpha|^2}{2\sigma^2}}\right],
    \label{eq::chi-three-peak}
\end{align}
and introduced two characteristic variances
\begin{align}
    \Sigma^2\equiv \frac{1+\nu}{1-\nu}\approx\infty,\qquad
    \sigma^2\equiv \frac{1}{2}\left(\frac{1}{\nu}-\nu\right)\approx 0,
\end{align}
with the approximations valid for $\nu\approx 1$ case.
We are primarily interested in the regime $\sigma\approx 0$ and $\Sigma\approx\infty$. 
If $\nu \geq2-\sqrt{3}\approx 0.2679$, then $\Sigma\ge\sigma$, with equality at $\Sigma^2=\sqrt{3}$. 
By tuning $\nu$, one can make $\sigma$ arbitrarily small and $\Sigma$ arbitrarily large.

\paragraph{Five-peak state.}
We similarly define the five-peak state
\begin{align}
    \hat{\sigma}_\gamma \equiv (1-\nu^2)^n \nu^{\hat{N}} \left(\hat{D}^{\dagger}(0)+i\epsilon_0\hat{D}^{\dagger}(\gamma)-i\epsilon_0\hat{D}^\dagger(-\gamma)+i\epsilon_0\hat{D}^{\dagger}(U^{\T}\gamma^*)-i\epsilon_0\hat{D}^\dagger(-U^{\T}\gamma^*)\right)\nu^{\hat{N}},
\end{align}
where $0<\nu<1$. The five-peak state exhibits a reflection symmetry and is specified by $\gamma$ and a symmetric unitary $U\in\mathbb{C}^{n\times n}$ (the $U$-dependence is implicit in $\hfive$). 
Although $U^{\T}=U$, we leave the definition in that form to highlight the reflection symmetry, $\hfive=\hfive^{(U)}$~\eqref{eq::reflected}.
As for the three-peak state, positivity holds when $0<\epsilon_0\leq1/4$. Its characteristic function is
\begin{align}
    \chi_{\hfive}(\alpha)
    &=e^{-\frac{|\alpha|^2}{2\Sigma^2}}\left[e^{-\frac{|\alpha|^2}{2\sigma^2}}+i\epsilon_0 e^{-\frac{|\gamma|^2}{2\Sigma^2}}\left(e^{-\frac{|\gamma-\alpha|^2}{2\sigma^2}}-e^{-\frac{|\gamma+\alpha|^2}{2\sigma^2}}+e^{-\frac{|U^{\T}\gamma^*-\alpha|^2}{2\sigma^2}}-e^{-\frac{|U^{\T}\gamma^*+\alpha|^2}{2\sigma^2}}\right)\right] \\
    &\equiv \chi_{\hat{\sigma}_0}(\alpha)+\chi_{\hat{\sigma}_\gamma}^{\text{(add)}}(\alpha),
\end{align}
with
\begin{align}
    \chi_{\hat{\sigma}_0}(\alpha)\equiv e^{-\frac{|\alpha|^2}{2\Sigma^2}}e^{-\frac{|\alpha|^2}{2\sigma^2}},\qquad
    \chi_{\hfive}^{\text{(add)}}(\alpha)\equiv i\epsilon_0e^{-\frac{|\alpha|^2+|\gamma|^2}{2\Sigma^2}}\left(e^{-\frac{|\gamma-\alpha|^2}{2\sigma^2}}-e^{-\frac{|\gamma+\alpha|^2}{2\sigma^2}} + e^{-\frac{|U^{\T}\gamma^*-\alpha|^2}{2\sigma^2}}-e^{-\frac{|U^{\T}\gamma^*+\alpha|^2}{2\sigma^2}}\right).
\end{align}
When $\gamma=0$, both $\hthree$ and $\hfive$ reduce to $\hro_0=\hat{\sigma}_0=(1-\nu^2)^n\nu^{2\hat{N}}$, the $n$-mode thermal state with mean photon number $n\nu^2/(1-\nu^2)$.

\paragraph{Wigner functions.}
The Wigner function, the Fourier transform of the characteristic function~\cite{ferraro2005gaussian, weedbrook2012gaussian, serafini2017quantum}, is
\begin{align}
    W_{\hro}(\beta) \equiv \frac{1}{\pi^{2n}}\int d^{2n}\alpha \chi_{\hro}(\alpha) e^{\alpha^{\dagger}\beta-\beta^{\dagger}\alpha}.
\end{align}
For the three-peak and five-peak states we obtain
\begin{align}
    W_{\hthree} (\beta) &=\frac{1}{\pi^{n}a ^{n}} 
    e^{-\frac{|\beta|^{2}}{a }} 
    \left[ 1 + 4\epsilon_{0} 
    e^{- \left( \frac{1}{2\Sigma^{2}} + \frac{1}{2(\Sigma^{2} + \sigma^{2})} \right) |\gamma|^{2} } \sin\left(\frac{2\Im(\beta^{\dagger}\gamma) }{1+\frac{\sigma^2}{\Sigma^2}} \right) \right], \\ 
    W_{\hfive}(\beta) &=\frac{1}{\pi^{n}a^{n}} 
    e^{-\frac{|\beta|^{2}}{a }} \left[1 + 2\epsilon_0 e^{-\left(\frac{1}{2{\Sigma^2}} + \frac{1}{2(\Sigma^2 + \sigma^2)}\right)|\gamma|^{2}}\left( \sin\left(\frac{2\Im(\beta^{\dagger}\gamma)}{1+\frac{\sigma^2}{\Sigma^2}} \right) + \sin\left(\frac{2\Im(\beta^{\dagger}U^{\T}\gamma^*)}{1+\frac{\sigma^2}{\Sigma^2}} \right) \right) \right],
\end{align}
where
\begin{align}
    a\equiv \frac{1}{2\sigma^2}+\frac{1}{2\Sigma^2}> 1/2.
\end{align}

\paragraph{$s$-ordered QPDs.}
By Def.~\ref{def::s-ordered}, the $s$-ordered characteristic function and QPD of the thermal state $\hro_0$ are
\begin{align}
    \chi_{\hro_{0}}(s,\alpha) =  
    e^{-\left(a-\tfrac{s}{2}\right)|\alpha|^2},\qquad 
    W^{(0)}(s,\beta) =  \frac{1}{\pi^n \left(a-\tfrac{s}{2} \right)^n } \exp\left(- \frac{|\beta|^2}{a-\frac{s}{2}}\right).
\end{align}
For the three-peak state,
\begin{align}
    \chi_{\hro_{\gamma}}(s,\alpha) &=  
    \chi_{\hro_{0}}(s,\alpha)
    + 2i\epsilon_0 e^{-\left(\tfrac{1}{2\Sigma^2}-\tfrac{s}{2}\right)|\alpha|^2} e^{-\frac{|\gamma|^2}{2\Sigma^2}}
    \left( e^{-\tfrac{|\gamma - \alpha|^2}{2\sigma^2}} -  e^{-\tfrac{|\gamma + \alpha|^2}{2\sigma^2}} \right),\\
    W_{\gamma}(s,\beta) &=   W^{(0)}(s,\beta)\left( 1+ 4\epsilon_0 
    e^{f_1(s)|\gamma|^2}  \sin\big(f_2(s)\Im[\beta^{\dagger}\gamma] \big) \right),
\end{align}
where
\begin{align}
    f_1(s) &= \frac{1}{2} - \frac{1-s}{(1-s) + \nu^2 (1+s)},\qquad
    f_2(s) = \frac{4\nu}{(1-s)+\nu^2 (1+s)}.
\end{align}

\paragraph{Energy and classicality.}
Both $\hro_0$ and $\hro_\gamma$ have the same mean photon number and energy:
\begin{align}
    \langle \hat{N}\rangle_{\hro_0} = \langle \hat{N} \rangle_{\hro_{\gamma}} = n\left(a-\tfrac{1}{2}\right),
    \qquad
    \langle \hat{E}\rangle_{\hro_0} = \langle \hat{E} \rangle_{\hro_{\gamma}} = na.
\end{align}
However, they may differ in classicality (Def.~\ref{def::nonclassicality}). 
For the three-peak state, $s_{\text{max}}(\nu,\gamma,\epsilon_0)$ is the largest $s$ such that $W_{\gamma}(s,\beta)$ is a valid probability measure. This requires $f_1(s)\leq c$, where $c\equiv \frac{1}{|\gamma|^2}\log\big(\tfrac{1}{4\epsilon_0}\big)$. Hence,
\begin{align} \label{eq::s_max}
    s\leq \frac{(1-\nu^2) + 2c(1+\nu^2)}{(1+\nu^2) + 2c(1-\nu^2)},
    \qquad
    s_{\text{max}}(\nu,\gamma,\epsilon_0) = \min\left(1,\ \frac{(1-\nu^2) + 2c(1+\nu^2)}{(1+\nu^2) + 2c(1-\nu^2)} \right),
\end{align}
whereas $s_{\text{max}}=1$ for the thermal (classical) state $\hro_0$. For practical purposes we introduce a lower bound on the classicality that depends only on $\nu$,
\begin{align}\label{eq::s-s'-relation}
    s'(\nu) \equiv \frac{1-\nu^2}{1+\nu^2} = \frac{1}{2a}\leq s_{\text{max}}(\nu,\gamma,\epsilon_0),
\end{align}
and use it to express the characteristic variances:
\begin{align}
    \Sigma^2 &= \frac{1 + \sqrt{1 - (s')^2}}{s'}  \label{eq::S_lower},\\
    \sigma^2 &= \frac{s'}{\sqrt{1 - (s')^2}}. \label{eq::s_lower}
\end{align}
Three-peak states with $s_{\text{max}}\in(0,1)$ are studied in Sec.~\ref{subsec::9-C}. In this domain, Eq.~\eqref{eq::s_max} further implies the constraint $2c<s_{\text{max}}$. Rearranging gives
\begin{align}
    c= \frac{(1+s_{\text{max}})\nu^2 - (1-s_{\text{max}})}{(1+s_{\text{max}})\nu^2 +(1-s_{\text{max}})},
\end{align}
and since $c\geq 0$, any three-peak state of classicality $s_{\text{max}}$ is parametrized by $\nu \in \left[ \sqrt{\frac{1-s_{\text{max}}}{1+s_{\text{max}}}},1 \right)$. 
Hence, any state with $\nu \in \left(0, \sqrt{\frac{1-s_{\text{max}}}{1+s_{\text{max}}}}~ \right]$ is guaranteed to be at least $s_{\text{max}}$ classical. 
Intuitively, decreasing $\nu$ strengthens the Gaussian envelope $e^{-\frac{|\alpha|^2}{2\Sigma^2}}$ in Eq.~\eqref{eq::chi-three-peak}, suppressing high-frequency components and yielding a smoother, more classical Wigner function. Moreover, the bound on $c$ introduces a lower bound on $|\gamma|^2$:
\begin{align}\label{eq::gamma-lower-bound}
    |\gamma|^2 > \frac{2}{s_{\text{max}}} \log\left( \frac{1}{4\epsilon_0}\right),
\end{align}
for any state of classicality $s_{\text{max}}$. Larger $|\gamma|^2$ correspond to higher-frequency features and hence reduced classicality.

\section{Lower bound for entanglement-free schemes}\label{sec::no-entanglement}
The BM scheme (see main text) uses both entangled measurements and reflected states to learn efficiently. 
In this section we quantify how removing entangled measurements affects sample complexity.
We prove rigorous lower bounds in the EF setting and show that they grow exponentially with $n$, the number of modes. 
Consequently, we obtain an exponential separation between schemes with simultaneous access to both resources and those restricted to EF measurements.

\begin{theorem}\label{thm::scaling-single-copy2} 
Let $\hat{\rho}$ be an arbitrary $n$-mode quantum state ($n\geq8$) and let $\hror$ denote its reflected state for some symmetric unitary $U\in\mathbb{C}^{n\times n}$. Assume $\epsilon\in(0,0.245)$ and $\kappa>0$. Consider any learning scheme on single copies of $\hro$ and $\hror$ (without entangled measurements). Suppose $\chi_{\hat{\rho}}(\alpha)$ is learned using $N$ copies of $\hro$ and $\hror$, and the scheme then receives a query $\alpha\in\mathbb{C}^n$ with $|\alpha|^2\leq\kappa n$. If the scheme outputs an estimate $\tilde{\chi}_{\hro}(\alpha)$ satisfying $|\tilde{\chi}_{\hro}(\alpha)-\chi_{\hat{\rho}}(\alpha)|\leq\epsilon$ with probability at least $2/3$, then
\begin{align}
N=\Omega\!\left(\epsilon^{-2}\,\bigl(1+0.99\,\kappa\bigr)^{n}\right).
\end{align}
\end{theorem}

\noindent\textit{Revealed hypothesis testing and reduction to learning.}
A standard route to lower bounds is to formulate a (partially revealed) hypothesis-testing task~\cite{chen2024tight,oh2024entanglement} and reduce that test to our original learning goal.
We introduce such a task—a distinguishing game between Alice and Bob—whose hardness (under EF measurements) translates into lower bounds for the characteristic-function learning task.
Although different quantum states may vary in difficulty, any general learning scheme must succeed uniformly over all inputs. Therefore, proving exponential hardness even for a specific family of states (e.g., three-peak states) suffices to imply exponential lower bounds for the general task. 
In fact, we grant Bob extra information by restricting the possible states to this family; despite this help, distinguishing remains exponentially hard, which only strengthens the general lower bound.

\subsection{Revealed hypothesis testing for three-peak states} \label{sec::revealed_hyp_test}
In our setup for the hypothesis testing task, Alice samples $s\in\{\pm 1\}$ uniformly and $\gamma\in\mathbb{C}^n$ from the multivariate Gaussian
\begin{equation}\label{eq::multivariate-gaussian}
    q(\gamma) = \left( \frac{1}{2\pi\sigma_{\gamma}^{2}} \right)^n
    e^{-\frac{|\gamma|^{2}}{2\sigma^{2}_{\gamma}}}.
\end{equation}
She then, with equal probability:
\begin{enumerate}
\item prepares $N$ copies of $\hat{\rho}_0$ and its reflected state $\hat{\rho}_0^{(U)}$; or
\item prepares $N$ copies of $\hat{\rho}_{s\gamma}$ and its reflected state $\hat{\rho}_{s\gamma}^{(U)}$.
\end{enumerate}
Alice first reveals the symmetric unitary $U$ to Bob. She then supplies the $N$ copies in any order of $\hro$/$\hror$ that Bob specifies (e.g., $\hro,\hror,\hror,\dots,\hro$). Bob receives and measures the copies, knowing for each position whether it is the original or reflected state, using any adaptive scheme, but with no entangled measurements. 
After Bob completes all quantum measurements and retains only classical data, Alice reveals $\gamma$ (but not $s$). 
Bob must now decide which hypothesis holds. The proof remains unchanged even if Bob can, in the beginning of each round, choose whether to measure the original or the reflected state.

\textit{Reduction to the learning task.}
Assume Bob has a protocol that, using $N$ copies, learns $\chi_{\hat{\rho}}$ (up to an overall sign) with precision $\epsilon$ at success probability at least $2/3$. Then Bob can win the above game with high probability. 
Intuitively, the two hypotheses produce noticeably different values of $\chi$ at the revealed point $\alpha=\gamma$ (or at $\alpha=U\gamma^*$ for the reflected state):
\begin{align}
    \underset{\tau_1,\tau_2 \in \{\pm 1\}}{\min} |\tau_1\chi_{\hat{\rho}_{s\gamma}}(\gamma)-\tau_2\chi_{\hat{\rho}_0}(\gamma)| 
    &\geq |\chi_{\hat{\rho}_{s\gamma}}(\gamma)-\chi_{\hat{\rho}_0}(\gamma)| \notag\\
    &\geq 2\epsilon_0 e^{-\frac{|\gamma|^2}{\Sigma^2}}\left[1-\exp\left(-\frac{2|\gamma|^2}{\sigma^2}\right)\right], \label{eq::chi_diff_1} \\
    \underset{\tau_1,\tau_2 \in \{\pm 1\}}{\min} |\tau_1 \chi_{\hat{\rho}^{(U)}_{s\gamma}}(U\gamma^*)- \tau_2 \chi_{\hat{\rho}^{(U)}_0}(U\gamma^*)|
    &\geq |\chi_{\hat{\rho}^{(U)}_{s\gamma}}(U\gamma^*)-  \chi_{\hat{\rho}^{(U)}_0}(U\gamma^*)| \notag\\
    &\geq 2\epsilon_0 e^{-\frac{|\gamma|^2}{\Sigma^2}}\left[1-\exp\left(-\frac{2|\gamma|^2}{\sigma^2}\right)\right], \label{eq::chi_diff_1_reflected}
\end{align}
where the $\tau_{j}$ account for the unknown global sign in the learned quantity. 
However, one can notice from the above inequalities that the lower bound is the same regardless of whether the learning scheme outputs an estimate of the characteristic function, or of the characteristic function up to sign. 
Hence, moving forward, we focus on learning the characteristic function with the correct sign while the same conclusions apply to learning up to an unknown sign.
% Since the lower bounds in \eqref{eq::chi_diff_1}–\eqref{eq::chi_diff_1_reflected} are sign-invariant, learning up to sign is sufficient; for simplicity we henceforth phrase statements for learning with the correct sign.
Using $\epsilon_0\leq \tfrac{1}{4}$, whenever $2\sigma^2<|\gamma|^2\leq \kappa n$ we obtain
\begin{align}
    |\chi_{\hat{\rho}_{s\gamma}}(\gamma)-\chi_{\hat{\rho}_0}(\gamma)|> 1.96e^{-\frac{\kappa n}{\Sigma^2}}\epsilon_0,
    \qquad
    |\chi_{\hat{\rho}^{(U)}_{s\gamma}}(U\gamma^*)-\chi_{\hat{\rho}^{(U)}_0}(U\gamma^*)|> 1.96 e^{-\frac{\kappa n}{\Sigma^2}}\epsilon_0,
\end{align}
since $e^{-\frac{2|\gamma|^2}{\sigma^2}}\leq e^{-4}\leq 0.0184$. 
Therefore, if Bob's estimate at the revealed phase space point has error $\epsilon\leq 0.98 e^{-\frac{\kappa n}{\Sigma^2}}\epsilon_{0}$, he distinguishes the hypotheses with probability at least $2/3$. 
Note that \eqref{eq::chi_diff_1}-\eqref{eq::chi_diff_1_reflected} are invariant under the $\gamma\mapsto -\gamma$ mapping, so Bob can succeed even though Alice never reveals~$s$.

\textit{Success probability of the hypothesis-testing task.}
Bob proceeds as follows: if $2\sigma^2<|\gamma|^2\leq \kappa n$, he uses the learned value at the revealed phase space point; otherwise, he guesses uniformly at random. The success probability obeys
\begin{align}\label{eq::sec6-tvd-1/6-one}
    \Pr(\text{success}) 
    &\geq 1/2\big(1-\Pr_\gamma(2\sigma^2<|\gamma|^2 \leq \kappa n)\big)+
    2/3 \Pr_\gamma(2\sigma^2<|\gamma|^2 \leq \kappa n) \notag\\
    &=2/3-\frac{1}{6}\big(1-\Pr_\gamma(2\sigma^2<|\gamma|^2 \leq\kappa n)\big) \notag\\ 
    &=2/3-\frac{1}{6}\Big(\Pr_\gamma(|\gamma|^2\leq 2\sigma^2)+\Pr_\gamma(|\gamma|^2>\kappa n)\Big) \notag\\ 
    &=2/3-\frac{1}{6}\left[\left(1-\frac{\Gamma(n,\sigma^2/\sigma_\gamma^2)}{\Gamma(n)}\right)+\left(\frac{\Gamma(n,\kappa n/(2\sigma_\gamma^2))}{\Gamma(n)}\right)\right].
\end{align}
Here
\begin{align}\label{eq::sec6-tvd-1/6-two}
    \Pr_\gamma(|\gamma|^2\leq 2\sigma^2)=1-\frac{\Gamma(n,\sigma^2/\sigma_\gamma^2)}{\Gamma(n)},\qquad
    \Pr_\gamma(|\gamma|^2>\kappa n )=\frac{\Gamma(n,\kappa n/(2\sigma_\gamma^2))}{\Gamma(n)}.
\end{align}
When $\kappa n\geq 2n\sigma_\gamma^2/0.99$ and $\sigma^2\leq \sigma_\gamma^2$, the right-hand side is bounded below by
\begin{align} \label{eq::sec6-tvd-1/6-three}
    2/3-\frac{1}{6}\left[\left(1-\frac{\Gamma(n,\sigma^2/\sigma_\gamma^2)}{\Gamma(n)}\right)+\left(\frac{\Gamma(n,n/0.99)}{\Gamma(n)}\right)\right]
    &\geq 2/3-\frac{1}{6}\left[\left(1-\frac{\Gamma(n,\sigma^2/\sigma_\gamma^2)}{\Gamma(n)}\right)+0.492\right] \notag\\
    &=2/3-\frac{1}{6}\left[\frac{\int_0^{\sigma^2/\sigma_\gamma^2} t^{n-1}e^{-t}dt}{(n-1)!}+0.492\right] \notag\\ 
    &\geq 1/2(1+1/6),
\end{align}
where for the first inequality we set $n\geq 8$ and utilized (see Sec.~S4 of Ref.~\cite{oh2024entanglement}):
\begin{align}
    \frac{\Gamma(n,n/0.99)}{\Gamma(n)}\leq 0.492.
\end{align}
For the second inequality, we again set $n\geq8$ such that 
\begin{align}
    \left[\frac{\int_0^{\sigma^2/\sigma_\gamma^2} t^{n-1}e^{-t}dt}{(n-1)!}+0.492\right] \leq 1/2.
\end{align}
To satisfy the above reduction with the success guarantees mentioned above, we ensure that $0.99\kappa \geq 2\sigma^2 = (1/\nu-\nu)$. 
For now we assume this to be true and will verify it below (see Eq.~\eqref{eq::kappa_min_6}).
We also introduce $\eta_1,\eta_2 \in (0,1]$ such that $\epsilon = 0.98\eta_1 e^{-\frac{\kappa n}{\Sigma^2}} \epsilon_0$ and $2\sigma_{\gamma}^2 = 0.99 \eta_2 \kappa$. 
Because a general learning scheme must work for all admissible values, we will eventually focus on the hardest regime and set $\eta_1=\eta_2=1$.

Hence, access to a learning scheme that estimates $\chi_{\hat{\rho}}(\gamma)$ for all $\gamma\in\mathbb{C}^n$ with $|\gamma|^2\leq \kappa n$, to $\epsilon = 0.98\eta_1 e^{-\frac{\kappa n}{\Sigma^2}} \epsilon_0$ accuracy and success probability at least $2/3$, guarantees
\begin{align}\label{eq::star-tvd-lowerbound}
    \mathbb{E}_{\gamma}\text{TVD}\left(p_0, p_{\pm\gamma}\right)\geq 1/6,
\end{align}
where
\begin{align}    
    \text{TVD}(p,q)
    \equiv\tfrac{1}{2}\sum_{x}|p(x)-q(x)|
    =\sum_{x}\max\{0,p(x)-q(x)\}.
\end{align}

We have thus reduced the distinguishing task to the learning task. 
It remains to show that any adaptive EF measurement protocol achieving sufficiently large TVD between the two hypotheses (i.e., 1/6) requires exponentially many samples.

\subsection{TVD upper bound}
In this subsection we derive the sample complexity lower bound required to achieve a sufficiently large TVD (i.e., 1/6), for the hypothesis testing task considered in the previous subsection. 
We allow for any general adaptive EF measurement scheme that has access to reflected states.

For simplicity, let us first focus on the case when we are only given $N$ copies of the state $\hat{\rho}$ and then argue that the same proof holds when the $N$ copies are made up of both $\hat{\rho}$ and $\hat{\rho}^{(U)}$.
We now demonstrate that any protocol capable of learning $\chi_{\hat{\rho}}(\alpha)$ up to an unknown sign within $\epsilon$ precision, without access to entangled measurements, requires an exponentially large sample complexity.

To begin, we use the inequality between the TVD and the success probability~\cite{cover1999elements}:
\begin{align}
    \Pr[\text{Success}|\gamma]\leq \frac{1}{2}(1+\text{TVD}(p_0,p_{\pm\gamma})),
\end{align}
and thus
\begin{align}
    \Pr[\text{Success}]\leq \frac{1}{2}(1+\mathbb{E}_\gamma\text{TVD}(p_0,p_{\pm\gamma})).
\end{align}
Let $p_0(o_{1:N})$ be the probability distribution for the measurement outcomes of the quantum state $\hat{\rho}_0$:
\begin{align}
    p_0(o_{1:N})=p_0(o_1)p_0(o_2|o_1)\cdots p_0(o_N|o_{1:N-1}),
\end{align}
where the notation $p(o_i|o_{1:i-1})$ is used to describe the adaptivity of the measurement apparatus, so
\begin{align}
    p_0(o_i|o_{1:i-1})
    =\Tr[\hat{\rho}_0 \hat{\Pi}_{o_{1:i-1}}(o_i)].
    %=\Tr[\hat{\rho}_0 |B_{o_i}^{o_{1:i-1}}\rangle\langle B_{o_i}^{o_{1:i-1}}|]
    %=\langle B_{o_i}^{o_{i-1}}|\hat{\rho}_0|B_{o_i}^{o_{1:i-1}}\rangle.
\end{align}
Hence, the POVM $\{\hat{\Pi}_{o_{1:i-1}}(o_i)\}_{o_i}$ may depend on the previous outcomes $o_{1:i-1}$.
For the three-peak state, we similarly define $p_{s\gamma}(o_{1:N})$ as,
\begin{align}
    p_{s\gamma}(o_i|o_{1:i-1})
    =\Tr[\hat{\rho}_{s\gamma} \hat{\Pi}_{o_{1:i-1}}(o_i)],
    %=\Tr[\hat{\rho}_{s\gamma} |B_{o_i}^{o_{1:i-1}}\rangle\langle B_{o_i}^{o_{i-1}}|]
    %=\langle B_{o_i}^{o_{1:i-1}}|\hat{\rho}_{s\gamma}|B_{o_i}^{o_{1:i-1}}\rangle 
\end{align}
and $p_{\pm\gamma}(o_{1:N})$ as
\begin{align}
    p_{\pm\gamma}(o_{1:N})\equiv \mathbb{E}_{s=\pm1}[p_{s\gamma}(o_{1:N})].
\end{align}
Noting that
\begin{align}
    p_{s\gamma}(o_{i}|o_{1:i-1})
    =\Tr[\hat{\rho}_{s\gamma} \hat{\Pi}_{o_{1:i-1}}(o_i)]
    &=\frac{1}{\pi^n}\int d^{2n}\alpha  \chi_{\hat{\rho}_{s\gamma}}(\alpha)\Tr[\hat{D}^\dagger(\alpha)\hat{\Pi}_{o_{1:i-1}}(o_i)],
    %&=\frac{1}{\pi^n}\int d^{2n}\alpha  (\chi_{\hat{\rho}_{0}}(\alpha)+\chi^\text{(add)}_{\hat{\rho}_{s\gamma}}(\alpha))\Tr[\hat{D}^\dagger(\alpha)\hat{\Pi}_{o_{1:i-1}}(o_i)].
\end{align}
we compute the difference between the probabilities:
\begin{align}
    p_0(o_{1:N})-p_{\pm\gamma}(o_{1:N})
    &=p_0(o_{1:N})\left(1-\frac{p_{\pm\gamma}(o_{1:N})}{p_0(o_{1:N})}\right) \\
    &=p_0(o_{1:N})\left(1-\mathbb{E}_{s=\pm1}\prod_{i=1}^N \left(\frac{\frac{1}{\pi^n}\int d^{2n}\alpha_i  \chi_{\hat{\rho}_{s\gamma}}(\alpha_i)\Tr[\hat{D}^\dagger(\alpha_i)\hat{\Pi}_{o_{1:i-1}}(o_i)]}{\frac{1}{\pi^n}\int d^{2n}\alpha_i  \chi_{\hat{\rho}_0}(\alpha_i)\Tr[\hat{D}^\dagger(\alpha_i)\hat{\Pi}_{o_{1:i-1}}(o_i)]}\right)\right) \\ 
    &=p_0(o_{1:N})\left(1-\mathbb{E}_{s=\pm1}\prod_{i=1}^N \left(1+\frac{\frac{1}{\pi^n}\int d^{2n}\alpha_i \chi_{\hat{\rho}_{s\gamma}^\text{(add)}}(\alpha_i)\Tr[\hat{D}^\dagger(\alpha_i)\hat{\Pi}_{o_{1:i-1}}(o_i)]}{\frac{1}{\pi^n}\int d^{2n}\alpha_i \chi_{\hat{\rho}_0}(\alpha_i)\Tr[\hat{D}^\dagger(\alpha_i)\hat{\Pi}_{o_{1:i-1}}(o_i)]}\right)\right) \\ 
    &=p_0(o_{1:N})\left(1-\mathbb{E}_{s=\pm1}\prod_{i=1}^N \left(1-4\epsilon_0 \Im G_\sigma^{o_{\leq i}}(s\gamma)\right)\right),
\end{align}
where we defined
\begin{align}\label{app-eq:G1}
    G_\sigma^{o_{\leq i}}(\gamma)
    &=\frac{\frac{1}{\pi^n}\int d^{2n}\alpha  e^{-\frac{|\alpha|^2}{2\Sigma^2}}e^{-\frac{|\gamma|^2}{2\Sigma^2}}e^{-\frac{|\gamma-\alpha|^2}{2\sigma^2}}G^{o_{\leq i}}(\alpha)}{\frac{1}{\pi^n}\int d^{2n}\alpha  e^{-\frac{|\alpha|^2}{2\Sigma^2}}e^{-\frac{|\alpha|^2}{2\sigma^2}}G^{o_{\leq i}}(\alpha)},
\end{align}
and
\begin{align}
    G^{o_{\leq i}}(\alpha)\equiv \Tr[\hat{D}^\dagger(\alpha)\hat{\Pi}_{o_{1:i-1}}(o_i)].
\end{align}
Note that $G^{o_{\leq i}}(\alpha)=G^{o_{\leq i}}(-\alpha)^*$.
We thus have
\begin{align}
    \mathbb{E}_\gamma\text{TVD}(p_0,p_{\pm\gamma})
    =\mathbb{E}_\gamma\sum_{o_{1:N}}p_0(o_{1:N})\max\left\{0,1-\mathbb{E}_{s=\pm1}\prod_{i=1}^N \left(1-4\epsilon_0 \Im G_\sigma^{o_{\leq i}}(s\gamma)\right)\right\}.
\end{align}
Here, to find the upper bound of the averaged TVD, we establish the following inequalities:
\begin{align}\label{app-eq:inequalities}
    \mathbb{E}_{s=\pm1}\prod_{i=1}^N \left(1-4\epsilon_0 \Im G_\sigma^{o_{\leq i}}(s\gamma)\right) 
    &\geq \prod_{i=1}^N\sqrt{\left(1-4\epsilon_0 \Im G_\sigma^{o_{\leq i}}(\gamma)\right)\left(1-4\epsilon_0 \Im G_\sigma^{o_{\leq i}}(-\gamma)\right)} \\
    &=\prod_{i=1}^N\sqrt{1-16\epsilon_0^2 \left(\Im G_\sigma^{o_{\leq i}}(\gamma)\right)^2} \\ 
    &\geq \prod_{i=1}^N\left[1-16\epsilon_0^2 \left(\Im G_\sigma^{o_{\leq i}}(\gamma)\right)^2\right] \\
    &\geq 1-16\epsilon_0^2\sum_{i=1}^N |G_\sigma^{o_{\leq i}}(\gamma)|^2, \label{eq::tvd-ineq-end}
\end{align}
where the first inequality uses the AM-GM inequality and the fact that the expression inside the bracket is the ratio of two conditional probabilities and is thus non-negative.
The first equality uses the fact that ${\Im G_\sigma^{o_{\leq k}}(\gamma)=-\Im G_\sigma^{o_{\leq k}}(-\gamma)}$, and the second inequality uses $\sqrt{1-x}\geq 1-x$ for $x\in [0,1]$, and the final inequality uses the inequality $\prod_i(1-x_i)\geq 1-\sum_i x_i$ if $x_i\in [0,1]$ for all $i$.
Thus, the expected TVD is upper bounded by
\begin{align}
    \mathbb{E}_\gamma\text{TVD}(p_0,p_{\pm\gamma})
    \leq \sum_{o_{1:N}}p_0(o_{1:N})\sum_{i=1}^N 16 \epsilon_0^2\mathbb{E}_\gamma |G_\sigma^{o_{\leq i}}(\gamma)|^2. \label{eq::star-tvd-upperbound}
\end{align}
Therefore, the remaining task is to upper bound the following term:
\begin{align}\label{eq::G_expectation}
    \mathbb{E}_\gamma |G_\sigma^{o_{\leq i}}(\gamma)|^2
    =\mathbb{E}_\gamma\left[\frac{\left|\frac{1}{\pi^n}\int d^{2n}\alpha e^{-\frac{|\alpha|^2}{2\Sigma^2}}e^{-\frac{|\gamma|^2}{2\Sigma^2}}e^{-\frac{|\gamma-\alpha|^2}{2\sigma^2}}\Tr[\hat{D}^\dagger(\alpha)\hat{\Pi}_{o_{1:i-1}}(o_i)]\right|^2}{\left|\frac{1}{\pi^n}\int d^{2n}\alpha e^{-\frac{|\alpha|^2}{2\Sigma^2}}e^{-\frac{|\alpha|^2}{2\sigma^2}} \Tr[\hat{D}^\dagger(\alpha)\hat{\Pi}_{o_{1:i-1}}(o_i)]\right|^2}\right].
\end{align}

First, by averaging over $\gamma$, we obtain
\begin{align}\label{eq::cc-invariance1}
    &\mathbb{E}_\gamma \left[\left|\frac{1}{\pi^n}\int d^{2n}\alpha  e^{-\frac{|\alpha|^2}{2\Sigma^2}}e^{-\frac{|\gamma|^2}{2\Sigma^2}}e^{-\frac{|\gamma-\alpha|^2}{2\sigma^2}}\Tr[\hat{D}^\dagger(\alpha)\hat{\Pi}_{o_{1:i-1}}(o_i)]\right|^2\right] \\
    %&=\frac{1}{\pi^{2n}}\int d^{2n}\alpha  d^{2n}\beta  e^{-\frac{|\alpha|^2+|\beta|^2}{2\Sigma^2}}\mathbb{E}_\gamma \left[e^{-\frac{|\gamma|^2}{\Sigma^2}}e^{-\frac{|\gamma-\alpha|^2}{2\sigma^2}}e^{-\frac{|\gamma-\beta|^2}{2\sigma^2}}\right]\langle B|\hat{D}^\dagger(\alpha)|B\rangle\langle B|\hat{D}(\beta)|B\rangle \\ 
    &=\frac{1}{\pi^{2n}\left(1+\frac{2\sigma_\gamma^2}{\sigma^2}+\frac{2\sigma_\gamma^2}{\Sigma^2}\right)^n}\int d^{2n}\alpha  d^{2n}\beta  e^{-\frac{|\alpha|^2+|\beta|^2}{2\Sigma^2}}e^{-\frac{|\alpha|^2+|\beta|^2-\frac{|\alpha+\beta|^2\Sigma^2\sigma_\gamma^2}{\sigma^2\Sigma^2+2\sigma_\gamma^2(\sigma^2+\Sigma^2)}}{2\sigma^2}}\Tr[\hat{D}^\dagger(\alpha)\hat{\Pi}_{o_{1:i-1}}(o_i)]\Tr[\hat{D}(\beta)\hat{\Pi}_{o_{1:i-1}}(o_i)].
\end{align}
Thus, noting that we have a prefactor and omitting the indices of $\hat{\Pi}_{o_{1:i-1}}(o_i)$ for brevity, the numerator becomes
\begin{align}
    &\int \frac{d^{2n}\alpha d^{2n}\beta}{\pi^{2n}} e^{-\left(\frac{1}{2\sigma^2}+\frac{1}{2\Sigma^2}\right)(|\alpha|^2+|\beta|^2)+\frac{|\alpha+\beta|^2\Sigma^2\sigma_\gamma^2}{2\sigma^2(\sigma^2\Sigma^2+2\sigma_\gamma^2(\sigma^2+\Sigma^2))}}\Tr[\hat{D}^\dagger(\alpha)\hat{\Pi}]\Tr[\hat{D}(\beta)\hat{\Pi}] \label{eq::cc-invariance2} \\ 
    &=\int \frac{d^{2n}\alpha d^{2n}\beta}{\pi^{2n}} e^{-\left(\frac{1}{2\sigma^2}+\frac{1}{2\Sigma^2}\right)(|\alpha|^2+|\beta|^2)+\frac{|\alpha+\beta|^2\Sigma^2\sigma_\gamma^2}{2\sigma^2(\sigma^2\Sigma^2+2\sigma_\gamma^2(\sigma^2+\Sigma^2))}}\Tr[(\hat{D}^\dagger(\alpha)\otimes \hat{D}(\beta))(\hat{\Pi}\otimes \hat{\Pi})] \\ 
    &=\int \frac{d^{2n}\alpha d^{2n}\beta}{\pi^{2n}} e^{-\left(\frac{1}{2\sigma^2}+\frac{1}{2\Sigma^2}\right)(|\alpha|^2+|\beta|^2)+\frac{|\alpha+\beta|^2\Sigma^2\sigma_\gamma^2}{2\sigma^2(\sigma^2\Sigma^2+2\sigma_\gamma^2(\sigma^2+\Sigma^2))}}\Tr[\left(\hat{D}^\dagger\left(\frac{\alpha-\beta}{\sqrt{2}}\right)\otimes \hat{D}\left(\frac{\alpha+\beta}{\sqrt{2}}\right)\right)\hat{U}_\text{BS}^\dagger(\hat{\Pi}\otimes \hat{\Pi})\hat{U}_\text{BS}] \\ 
    &=\int \frac{d^{2n}\alpha' d^{2n}\beta'}{\pi^{2n}} e^{-\left(\frac{1}{2\sigma^2}+\frac{1}{2\Sigma^2}\right)(|\alpha'|^2+|\beta'|^2)+\frac{|\beta'|^2\Sigma^2\sigma_\gamma^2}{\sigma^2(\sigma^2\Sigma^2+2\sigma_\gamma^2(\sigma^2+\Sigma^2))}}\Tr[(\hat{D}^\dagger(\alpha')\otimes \hat{D}(\beta'))\hat{\Omega}] \\ 
    &=\int \frac{d^{2n}\alpha' d^{2n}\beta'}{\pi^{2n}} e^{-a|\alpha'|^2-b|\beta'|^2}\Tr[(\hat{D}^\dagger(\alpha')\otimes \hat{D}(\beta'))\hat{\Omega}] \\ 
    &=\left(\frac{2}{2a+1}\right)^n\left(\frac{2}{2b+1}\right)^n\Tr[\left(\frac{2a-1}{2a+1}\right)^{\hat{N}_1}\otimes \left(\frac{2b-1}{2b+1}\right)^{\hat{N}_2}\hat{\Omega}],
\end{align}
where we introduced the 50:50 beam splitters $\hat{U}_\text{BS}$ whose action is defined as
\begin{align}
    \hat{U}_\text{BS}\left(\hat{D}^\dagger\left(\frac{\alpha-\beta}{\sqrt{2}}\right)\otimes \hat{D}\left(\frac{\alpha+\beta}{\sqrt{2}}\right)\right)\hat{U}_\text{BS}^\dagger=\hat{D}^\dagger(\alpha)\otimes \hat{D}(\beta).
\end{align}
We also defined $\hat{\Omega}\equiv \hat{U}_\text{BS}^\dagger(\hat{\Pi}\otimes \hat{\Pi})\hat{U}_\text{BS}$, and
\begin{align}
    &a\equiv\frac{1}{2\Sigma^2}+\frac{1}{2\sigma^2}>\frac{1}{2},~~
    b\equiv\frac{1}{2\sigma^2}+\frac{1}{2\Sigma^2} -\frac{\Sigma^2\sigma_\gamma^2}{\sigma^2(\sigma^2\Sigma^2+2\sigma_\gamma^2(\sigma^2+\Sigma^2))}>0.
\end{align}
For the last equality, we used the relation from Eq.~\eqref{eq:thermal_char} that 
\begin{align}
    \frac{1}{\pi^n}\int d^{2n}\alpha e^{-a|\alpha|^2}\hat{D}(\alpha)=\left(\frac{2}{2a+1} \right)^{n}\left(\frac{2a-1}{2a+1}\right)^{\hat{N}}.
\end{align}
The denominator in Eq.~\eqref{eq::G_expectation} is easily obtained by substituting $\sigma_\gamma=0$ in the expression above:
\begin{align}
    \left|\frac{1}{\pi^n}\int d^{2n}\alpha e^{-\frac{|\alpha|^2}{2\Sigma^2}}e^{-\frac{|\alpha|^2}{2\sigma^2}}\Tr[\hat{D}^\dagger(\alpha)\hat{\Pi}]\right|^2 
    =\left(\frac{2}{2a+1}\right)^n\left(\frac{2}{2a+1}\right)^n\Tr[\left(\frac{2a-1}{2a+1}\right)^{\hat{N}_1}\otimes \left(\frac{2a-1}{2a+1}\right)^{\hat{N}_2}\hat{\Omega}].
\end{align}
Therefore, we obtain the following expectation:
\begin{align}
    \mathbb{E}_\gamma |G_\sigma^{o_{\leq i}}(\gamma)|^2
    &=\frac{1}{\left(1+\frac{2\sigma_\gamma^2}{\sigma^2}+\frac{2\sigma_\gamma^2}{\Sigma^2}\right)^n}\left(\frac{2a+1}{2b+1}\right)^n\frac{\Tr[\left(\frac{2a-1}{2a+1}\right)^{\hat{N}_1}\otimes \left(\frac{2b-1}{2b+1}\right)^{\hat{N}_2}\hat{\Omega}]}{\Tr[\left(\frac{2a-1}{2a+1}\right)^{\hat{N}_1}\otimes \left(\frac{2a-1}{2a+1}\right)^{\hat{N}_2}\hat{\Omega}]} \\ 
    &=\frac{1}{\left(1+2\sigma_\gamma^2\right)^n}\frac{\Tr[\left(\frac{2a-1}{2a+1}\right)^{\hat{N}_1}\otimes \left(\frac{2b-1}{2b+1}\right)^{\hat{N}_2}\hat{\Omega}]}{\Tr[\left(\frac{2a-1}{2a+1}\right)^{\hat{N}_1}\otimes \left(\frac{2a-1}{2a+1}\right)^{\hat{N}_2}\hat{\Omega}]}.
\end{align}

To find an upper bound of $\mathbb{E}_\gamma |G_\sigma^{o_{\leq i}}(\gamma)|^2$, noticing that $b=1/2$ when $2\sigma^2_\gamma = 1$, we structure our analysis around two distinct regions.
Region 1 is defined by $0\leq 2\sigma_{\gamma}^{2} < 1 \iff b> 1/2$. When $b> 1/2$, the operators on the right-hand side are positive semidefinite, and it is not hard to see, by monotonicity, that: 
\begin{align}
    \frac{\Tr[\left(\frac{2a-1}{2a+1}\right)^{\hat{N}_1}\otimes \left(\frac{2b-1}{2b+1}\right)^{\hat{N}_2}\hat{\Omega}]}{\Tr[\left(\frac{2a-1}{2a+1}\right)^{\hat{N}_1}\otimes \left(\frac{2a-1}{2a+1}\right)^{\hat{N}_2}\hat{\Omega}]} \leq 1.
\end{align}
Region 2 is defined by $2\sigma_{\gamma}^{2} \geq 1 \iff 0< b \leq 1/2$. 
In this region the following holds:
\begin{align}
    \frac{\Tr[\left(\frac{2a-1}{2a+1}\right)^{\hat{N}_1}\otimes \left(\frac{2b-1}{2b+1}\right)^{\hat{N}_2}\hat{\Omega}]}{\Tr[\left(\frac{2a-1}{2a+1}\right)^{\hat{N}_1}\otimes \left(\frac{2a-1}{2a+1}\right)^{\hat{N}_2}\hat{\Omega}]} 
    \leq \frac{\Tr[\left(\frac{2a-1}{2a+1}\right)^{\hat{N}_1}\otimes \left|\frac{2b-1}{2b+1}\right| ^{\hat{N}_2}\hat{\Omega}]}{\Tr[\left(\frac{2a-1}{2a+1}\right)^{\hat{N}_1}\otimes \left(\frac{2a-1}{2a+1}\right)^{\hat{N}_2}\hat{\Omega}]} 
    =\frac{\Tr[\left(\frac{2a-1}{2a+1}\right)^{\hat{N}_1}\otimes \left(\frac{1-2b}{1+2b}\right)^{\hat{N}_2}\hat{\Omega}]}{\Tr[\left(\frac{2a-1}{2a+1}\right)^{\hat{N}_1}\otimes \left(\frac{2a-1}{2a+1}\right)^{\hat{N}_2}\hat{\Omega}]} \leq 1.
\end{align}
Since $\frac{2a-1}{2a+1} - \frac{1-2b}{1+2b}=\frac{2\nu^2}{1+2\sigma_\gamma^2} >0$, the above inequality holds for any value of $\nu$.
Hence, we conclude that the same upper bound holds in both regions. Since their union covers the whole domain of $\sigma_\gamma^2$, our proof holds for any value of~$\nu$ and any value of $\sigma_\gamma^2$.

Finally, we arrive at the following bound:
\begin{align} \label{eq::three-peak-bound}
   &\mathbb{E}_\gamma |G_\sigma^{o_{\leq i}}(\gamma)|^2 \leq \frac{1}{\left(1+2\sigma_\gamma^2\right)^n}.
\end{align}
By combining this with Eq.~\eqref{eq::star-tvd-upperbound}, we establish a sample complexity lower bound to accomplish the distinguishing task as a function of the expected TVD:
\begin{align}
    N \geq \frac{\mathbb{E}_{\gamma}\text{TVD}(p_0,p_{\pm \gamma})}{16}\epsilon_0^{-2}\left( 1 + 2\sigma_\gamma^2 \right)^n,
\end{align}
and conclude that distinguishing between the two hypotheses using any EF scheme requires exponential sample complexity.

\textit{Relation to the learning task.}
In subsection~\ref{sec::revealed_hyp_test} we showed that if a learning scheme estimates $\chi_{\hat{\rho}}(\alpha)$ for all $|\alpha|^2\leq \kappa n$ to accuracy
\begin{equation}
  \epsilon  = 0.98 \eta_1e^{-\frac{\kappa n}{\Sigma^2}}\epsilon_0
  \qquad\text{and we set}\qquad
  2\sigma_\gamma^2  = 0.99\eta_2\kappa,
\end{equation}
then the revealed hypothesis test has an expected total variation distance at least
\begin{equation}
  \mathbb{E}_{\gamma}\text{TVD}\left(p_0, p_{\pm\gamma}\right)\ \geq\ \frac{1}{6}.
\end{equation}
On the other hand, from~\eqref{eq::star-tvd-upperbound} and the bound~\eqref{eq::three-peak-bound} we obtained
\begin{equation}
  \mathbb{E}_\gamma\text{TVD}\left(p_0,p_{\pm\gamma}\right)
   \leq 16 N \epsilon_0^{2}(1+2\sigma_\gamma^2)^{-n}.
\end{equation}
Combining the two yields the fundamental sample lower bound
\begin{equation}\label{eq:tvd-to-samples}
  N\ \geq \frac{1}{96}\epsilon_0^{-2}(1+2\sigma_\gamma^2)^{n}.
\end{equation}

To instantiate this in terms of the learning parameters $(\kappa,\epsilon)$, focus on the three-peak state with
\begin{equation*}
  \Sigma^2  = \frac{\kappa n}{\log(1+\eta_3)}\qquad\text{for some arbitrarily small constant } \eta_3>0.
\end{equation*}
Then $e^{\kappa n/\Sigma^2}=1+\eta_3$ and hence
\begin{equation}
  \epsilon_0 = \frac{1+\eta_3}{0.98}\epsilon.
\end{equation}
We require $\epsilon_0\leq \tfrac{1}{4}$ (as assumed earlier), which is equivalent to
\begin{equation}
  \epsilon\ \leq\ \frac{0.245}{1+\eta_3}.
\end{equation}
Because $\eta_3>0$ can be chosen arbitrarily small, this condition is satisfied for any fixed accuracy $\epsilon<0.245$ by picking $\eta_3\in \bigl(0,\frac{0.245}{\epsilon}-1\bigr]$.
We return to the three-peak state considered, and compute its corresponding $\nu$ and $\sigma^2$ parameters,
\begin{equation}\label{eq::sigma-Sigma-auxiliary}
    \nu = \frac{\Sigma^2 -1 }{\Sigma^2+1} = 
    \frac{\kappa n - L}{\kappa n + L}, \quad 
    \sigma^2 = \frac{2\Sigma^2}{\Sigma^4-1} =\frac{2\kappa n L}{(\kappa n)^2 - L^2},
\end{equation}
for $L\equiv \log(1+\eta_3) < \kappa n$.
The previously introduced constraint that provides the success guarantees upon reduction, $2\sigma^2 \leq 0.99 \kappa$, is now expressed as:
\begin{equation}\label{eq::kappa_min_6}
    0.99\kappa \geq \frac{4\kappa nL}{(\kappa n)^2 - L^2},
\end{equation}
which is satisfied for any 
\begin{equation}\label{eq::eta-3-constraint}
    0<L = \log(1+\eta_3)\leq n \left( -\frac{2}{0.99} + \sqrt{(\frac{2}{0.99})^2 + \kappa ^2} \right).
\end{equation}
Thus, given any queried $\kappa>0$, and accuracy $\epsilon \in(0,0.245)$, we can always find a three-peak state with $\epsilon_0=\tfrac{(1+\eta_3)}{0.98}\epsilon$, $\nu = \frac{\kappa n - \log(1+\eta_3)}{\kappa n +\log(1+\eta_3)}$, for $\eta_3 \in (0, \frac{0.245}{\epsilon}-1]$ and satisfying the above constraint.

The sample complexity required to achieve the expected TVD, Eq.~\eqref{eq::star-tvd-lowerbound}, and by reduction to complete the characteristic function learning task, with success probability $2/3$, then takes the form 
\begin{align}\label{eq::sample-complex-6}
  N\ \geq \frac{0.98^2}{96(1+\eta_3)^2} \epsilon^{-2} \bigl(1+0.99\kappa\bigr)^n,
\end{align}
for any $n\geq 8$, where we tightened the bound by setting $\eta_1=\eta_2=1$. 
Since we take $\eta_3>0$ to be an arbitrarily small constant in the allowed domain, thus independent of $n$, this implies the asymptotic form
\begin{align}
  N\ = \Omega\left(\epsilon^{-2}\bigl(1+0.99\kappa\bigr)^n\right),
\end{align}
which establishes Thm.~\ref{thm::scaling-single-copy2}.
We reiterate that proving the bound for a specific example state is sufficient since a general learning scheme should provide guarantees in accuracy/success probability for any input state. 

\textit{Access to reflected copies.} Let us now argue that the same bound applies even when the $N$ copies are made up of both $\hat{\rho}$ and $\hat{\rho}^{(U)}$.
Using the identities, Eqs.~$\eqref{eq::reflected}$ and $\eqref{eq::chi-three-peak}$, we notice that if $\hthree$ has peaks at $\pm\gamma$, then $\hthree^{(U)}$ has peaks at $\pm U^\T\gamma^*$.
When using $\hror_{\gamma}$, Eq.~\eqref{app-eq:G1} may change its form to 
\begin{align}
    G_\sigma^{o_{\leq i}}(\gamma)
    &=\frac{\frac{1}{\pi^n}\int d^{2n}\alpha e^{-\frac{|\alpha|^2}{2\Sigma^2}}e^{-\frac{|\gamma|^2}{2\Sigma^2}}e^{-\frac{|U^{\T} \gamma^*-\alpha|^2}{2\sigma^2}}G^{o_{\leq i}}(\alpha)}{\frac{1}{\pi^n}\int d^{2n}\alpha  e^{-\frac{|\alpha|^2}{2\Sigma^2}}e^{-\frac{|\alpha|^2}{2\sigma^2}}G^{o_{\leq i}}(\alpha)},
\end{align}
for some $i$'s.
However, after applying \eqref{app-eq:inequalities}, the only relevant quantity is $\mathbb{E}_\gamma|G_\sigma^{o_{\leq i}}(\gamma)|^2$, which is invariant under the $\gamma \mapsto U\gamma^*$ mapping. 
Hence the same upper bound \eqref{eq::three-peak-bound} applies, and the exponential lower bound remains valid even when reflected copies are available.

\section{Lower bound for entanglement-free schemes with reflection symmetry ($\hro=\hat{\rho}^{(U)}$) }\label{sec::lower-bound_reflection_symmetry}

The BM protocol in the main text requires simultaneous access to reflected states and entangled measurements. 
We have already shown that any learning scheme that prohibits entangled measurements incurs exponential sample complexity to learn the characteristic function.
This lower bound continues to hold even when reflected states are available.

On the other hand, if the available states exhibit reflection symmetry about some axes ($\hro=\hror$), the BM protocol can be applied directly to two copies of $\hro$.
In this setting, the learning task admits an $n$-independent sample complexity and does not require additional reflected states. States with reflection symmetry are therefore of particular interest. 
In this section, we analyze the sample complexity of learning the characteristic function when entangled measurements are prohibited but the states possess intrinsic reflection symmetry.

The three-peak state used in the previous section lacks reflection symmetry in general. To address this, we instead use the reflection-symmetric five-peak state, $\hfive$ (see Sec.~\ref{sec:input_states}).
Because we now consider a different family, we formulate a new hypothesis-testing task and devise a strategy that reduces the distinguishing problem to characteristic-function learning task.

Although our analysis is stated for the five-peak family, any valid learning scheme for reflection-symmetric inputs must accurately learn arbitrary reflection-symmetric states. 
Hence, establishing exponential sample complexity for five-peak states immediately yields an exponential lower bound when accomplishing the learning task for any general reflection-symmetric input state under EF measurements.

\subsection{Revealed hypothesis testing for five-peak states}\label{sec::five-peak-hypothesis}

\begin{theorem}\label{thm::scaling-single-copy3} 
Let $\hat{\rho}$ be an arbitrary $n$-mode quantum state ($n\geq 8$) with reflection symmetry, $\hro = \hro^{(U)}$, for some symmetric unitary $U\in \mathbb{C}^{n\times n}$. 
Assume $\epsilon\in(0,0.1225)$ and $\kappa>0$, and consider any learning scheme that operates on single copies of $\hro$. Suppose $\chi_{\hat{\rho}}(\alpha)$ is learned using $N$ copies of $\hro$, and the scheme then receives a query $\alpha\in\mathbb{C}^n$ with $|\alpha|^2\leq\kappa n$. If the scheme outputs an estimate $\tilde{\chi}_{\hro}(\alpha)$ satisfying $|\tilde{\chi}_{\hro}(\alpha)-\chi_{\hat{\rho}}(\alpha)|\leq\epsilon$ with probability at least $2/3$, then
\begin{equation}
N=\Omega\left(\epsilon^{-2}\left(1+0.66\kappa\right)^{n}\right).
\end{equation}
\end{theorem}

To establish hardness, we introduce a partially revealed hypothesis-testing task. 
Alice samples $s\in\{\pm{1}\}$ uniformly and $\gamma\in\mathbb{C}^n$ from the multivariate Gaussian distribution
\begin{equation}
    q(\gamma) = \left( \frac{1}{2\pi\sigma_{\gamma}^{2}} \right)^n e^{-\frac{\abs{\gamma}^{2}}{2\sigma^{2}_{\gamma}}}.
\end{equation}
She then, with equal probability, performs one of the following:

1) Prepare $N$ copies of the reflection-symmetric state $\hat{\sigma}_0$;

2) Prepare $N$ copies of the reflection-symmetric state $\hat{\sigma}_{s\gamma}$.

Bob receives the symmetric unitary matrix $U$ that defines $\hat{\sigma}_{s\gamma}$, together with the $N$ copies prepared by Alice. 
He measures these copies to acquire information about the state. 
After he completes all quantum measurements and retains only classical data, Alice reveals $\gamma$ while keeping $s$ secret. 
Bob must then decide which hypothesis holds. 
Crucially, all quantum measurements must be completed before $\gamma$ is revealed; only classical post-processing is allowed afterwards.

If Bob has access to a scheme that accurately learns the characteristic function (with the sign or up to an unknown sign), he can distinguish the two hypotheses with high probability.
Success in this task implies that the TVD between the corresponding measurement outcome distributions is sufficiently large.
In the next subsection, we show that under EF measurements achieving such TVD requires exponentially many samples.

We first analyze the separation between the characteristic functions of the two hypotheses:
\begin{align}
    \underset{\tau_1,\tau_2\in{\{ \pm 1 \}}}{\min}
    |\tau_1\chi_{\hat{\sigma}_{s\gamma}}(\alpha) - \tau_2\chi_{\hat{\sigma}_0}(\alpha)| 
    &\geq 
    |\chi_{\hat{\sigma}_{s\gamma}}(\alpha) - \chi_{\hat{\sigma}_0}(\alpha)| \\
    &\geq\epsilon_0 e^{-\frac{|\alpha|^2}{2\Sigma^2}} e^{-\frac{|\gamma|^2}{2\Sigma^2}} \abs{e^{-\frac{|\gamma-\alpha|^2}{2\sigma^2}} - e^ {-\frac{|\gamma+\alpha|^2}{2\sigma^2}}  + e^{-\frac{|U^{\T}\gamma^*-\alpha|^2}{2\sigma^2}} - e^ {-\frac{|U^{\T}\gamma^*+\alpha|^2}{2\sigma^2}}}.
    %\\&= \epsilon_0 e^{-\frac{|\alpha|^2}{2\Sigma^2}} e^{-\frac{|\gamma|^2}{2\Sigma^2}}\abs{e^{-\frac{|\gamma-\alpha|^2}{2\sigma^2}} - e^ {-\frac{|U^{\T}\gamma^*+\alpha|^2}{2\sigma^2}}    + e^{-\frac{|U^{\T}\gamma^*-\alpha|^2}{2\sigma^2}} - e^ {-\frac{|\gamma+\alpha|^2}{2\sigma^2}}}. 
\end{align}
When $\gamma=-U^{\T}\gamma^*$, the bound above evaluates to zero, making the two hypotheses identical and hence indistinguishable.
This happens because, under that condition, the five-peak state $\hfive$ reduces exactly to the thermal state $\hat{\sigma}_{0}$.
To handle this degeneracy, we specify a criterion that allows Bob, for any $\gamma$ and any symmetric unitary $U$, to decide whether the characteristic functions differ sufficiently for reliable discrimination.

Without loss of generality, we first analyze the $U=I$ case and later extend the argument to any symmetric unitary $U$. 
When $I$ defines $\hfive$, we have $\chi_{\hfive}(\alpha)=\chi_{\hfive}(\alpha^*)$, so the characteristic function is invariant under $\Im(\alpha)\rightarrow-\Im(\alpha)$, i.e., under $n$ reflections across the real axes in phase space. 
If $\gamma=-\gamma^*$ ($\Re(\gamma)=0$), four of the five peaks destructively interfere and the state reduces to $\hat{\sigma}_0$, making discrimination impossible. 
Therefore, ensuring that $|\Re(\gamma)|$ is sufficiently large enables reliable discrimination by estimating the characteristic function at $\alpha=\gamma$.

To be specific, let $\gamma_R \equiv \Re(\gamma)$ and $\gamma_I \equiv \Im(\gamma)$, and consider $U=I$. Then
\begin{align}
    \underset{\tau_1,\tau_2\in{\{ \pm 1 \}}}{\min}
    |\tau_1\chi_{\hat{\sigma}_{s\gamma}}(\gamma) - \tau_2\chi_{\hat{\sigma}_0}(\gamma)| &\geq 
    \epsilon_0 e^{-\frac{|\gamma|^2}{\Sigma^2}}  \abs{1 - e^{-\frac{2|\gamma_R|^2}{\sigma^2}} e^{-\frac{2|\gamma_I|^2}{\sigma^2}}  + e^{-\frac{2|\gamma_I|^2}{\sigma^2}} - e^ {-\frac{2|\gamma_R|^2}{\sigma^2}}} \\
    &= \epsilon_0 e^{-\frac{|\gamma_R|^2}{\Sigma^2}}e^{-\frac{|\gamma_I|^2}{\Sigma^2}}  \left(1 +  e^{-\frac{2|\gamma_I|^2}{\sigma^2}} \right) \left(1 -  e^{-\frac{2|\gamma_R|^2}{\sigma^2}} \right)
    \geq 0.98 e^{-\frac{\kappa n}{\Sigma^2}}\epsilon_0,
\end{align}
under the constraints $2\sigma^2<|\gamma_R|^2\leq \frac{2}{3}\kappa n$ and $0<|\gamma_I|^2\leq \frac{1}{3}\kappa n$. Alice samples from a $2n$-dimensional multivariate normal distribution, and for $\bold{x}\in\mathbb{R}^{n}$ we use
\begin{align}
    \int_{|\bold{x}|\leq R} d\bold{x} \frac{1}{(2\pi\sigma_\gamma^2)^{n/2}}  \exp(-\frac{|\bold{x}|^2}{2\sigma_\gamma^2}) 
    =1- \frac{\Gamma(n/2,\frac{R^2}{2\sigma_\gamma^2})}{\Gamma(n/2)},
\end{align}
where $\Gamma(n,x) = \int_{x}^{\infty} dt~ t^{n-1}e^{-t}$ is the upper incomplete gamma function and $\Gamma(n) = \Gamma(n,0)$. Since we have $2n$ i.i.d. variables, we truncate the real and imaginary parts separately:
\begin{align}
   P_{R}\equiv \Pr(2\sigma^2 <|\gamma_R|^2\leq \frac{2}{3}\kappa n) = \frac{\Gamma(n/2,\frac{\sigma^2}{\sigma_\gamma^2})}{\Gamma(n/2)} - \frac{\Gamma(n/2,\frac{\kappa n}{3\sigma_\gamma^2})}{\Gamma(n/2)}, \quad  
   P_I \equiv \Pr(0 < |\gamma_I|^2\leq \frac{1}{3} \kappa n ) = 1-\frac{\Gamma(n/2,\frac{\kappa n}{6\sigma_\gamma^2})}{\Gamma(n/2)}.
\end{align}
Set $\sigma_\gamma^2 = \frac{0.99}{3}\kappa $ with $\kappa \geq \kappa_{\text{min}}\equiv \frac{3}{0.99} \sigma^2$ so that $\kappa n\geq 3n\sigma_{\gamma}^2/0.99$ and $\sigma^2\leq \sigma_\gamma^2$. Then
\begin{align}
    &\frac{\Gamma(n/2,\frac{\sigma^2}{\sigma_\gamma^2})}{\Gamma(n/2)} \geq \frac{\Gamma(n/2,1)}{\Gamma(n/2)} \geq 0.9810, \quad
    \frac{\Gamma(n/2,\frac{\kappa n }{3\sigma_\gamma^2})}{\Gamma(n/2)} \leq \frac{\Gamma(n/2,\frac{n}{0.99})}{\Gamma(n/2)} \leq 0.0402, 
    &\frac{\Gamma(n/2,\frac{\kappa n }{6\sigma_\gamma^2})}{\Gamma(n/2)} \leq \frac{\Gamma(n/2,\frac{n}{1.98})}{\Gamma(n/2)} \leq 0.4677,
\end{align}
where we considered $n\geq 8$. 
For the second term, we directly check that $\frac{\Gamma(n/2,\frac{n}{0.99})}{\Gamma(n/2)}\leq 0.0402$ on the $8\leq n \leq 30000$ interval (see Fig.~\ref{fig::two_bounds}). When $n \geq 30000$, we use the following inequality~\cite{ghosh2021exponential}:
\begin{equation}
    \frac{\Gamma(n,kn)}{\Gamma(n)} \leq (k e^{1-k})^{n}, \quad \forall k>1. \label{eq::gamma-ineq}
\end{equation}
Setting  $k=2/0.99$ and $n = 30000$, $(k e^{1-k})^{n/2}  <0.0402$. By monotonicity, the bound holds for any $n\geq30000$. Combining the two intervals completes the proof for the bound of the second term. For the third term, $\frac{\Gamma(n/2,\frac{n}{1.98})}{\Gamma(n/2)}\leq 0.4539$ on the $8\leq n\leq 30000$ interval (see Fig.~\ref{fig::two_bounds}) and $\frac{\Gamma(n/2,\frac{n}{1.98})}{\Gamma(n/2)}\leq 0.4677$ when $n\geq 30000$ using the inequality~\eqref{eq::gamma-ineq}.
Hence
\begin{align}
    P_R \geq 0.9408,~ P_I \geq 0.5323,~ P \equiv P_R\times P_I \geq 0.5007 \geq 1/2.
\end{align}
Therefore, with probability at least $1/2$, $\gamma$ lies in the range $2\sigma^2 < |\gamma_R|^2\leq \frac{2}{3}\kappa n$ and $0 <|\gamma_I|^2\leq \frac{1}{3}\kappa n$, allowing Bob to estimate the characteristic function (within admissible error and with success probability at least $2/3$) and determine Alice's hypothesis.

The argument extends to any symmetric unitary $U$. By the Takagi decomposition $U=VV^{\T}$ (see Sec.~\ref{sec::reflected-states}), the unitary $V$ defines reflections about general axes. 
For successful discrimination, it suffices that $|\Re(V^\dagger\gamma)|^2>2\sigma^2$. We then have
\begin{align}
    |\chi_{\hat{\sigma}_{s\gamma}}(\alpha) - \chi_{\hat{\sigma}_0}(\alpha)| 
    &\geq \epsilon_0  e^{-\frac{|\alpha|^2}{2\Sigma^2}} e^{-\frac{|\gamma|^2}{2\Sigma^2}}   \abs{   e^{-\frac{|\gamma-\alpha|^2}{2\sigma^2}} - e^ {-\frac{|U^{\T}\gamma^*+\alpha|^2}{2\sigma^2}}    + e^{-\frac{|U^{\T}\gamma^*-\alpha|^2}{2\sigma^2}} - e^ {-\frac{|\gamma+\alpha|^2}{2\sigma^2}}    }, \\
    |\chi_{\hat{\sigma}_{s\gamma}}(\gamma) - \chi_{\hat{\sigma}_0}(\gamma)| 
    &\geq \epsilon_0  e^{-\frac{|\gamma|^2}{\Sigma^2}}  \abs{  1 - e^ {-\frac{|VV^{\T}\gamma^*+\gamma|^2}{2\sigma^2}}    + e^{-\frac{|VV^{\T}\gamma^*-\gamma|^2}{2\sigma^2}} - e^ {-\frac{2|\Re(\gamma)|^2}{\sigma^2}} e^ {-\frac{2|\Im(\gamma)|^2}{\sigma^2}}    } \\
    & = \epsilon_0 e^ {-\frac{|\Re(V^\dagger\gamma)|^2}{\Sigma^2}} e^ {-\frac{|\Im(V^\dagger\gamma)|^2}{\Sigma^2}}  \abs{ 1 - e^ {-\frac{2|\Re(V^\dagger\gamma)|^2}{\sigma^2}}   + e^ {-\frac{2|\Im(V^\dagger\gamma)|^2}{\sigma^2}}  - e^ {-\frac{2|\Re(V^\dagger\gamma)|^2}{\sigma^2}} e^{-\frac{2|\Im(V^\dagger\gamma)|^2}{\sigma^2}}    } \\
    & = \epsilon_0  e^ {-\frac{|\gamma_R'|^2}{\Sigma^2}} e^ {-\frac{|\gamma_I'|^2}{\Sigma^2}}    \bigl( 1 + e^ {-\frac{2|\gamma_I'|^2}{\sigma^2}} \bigr)  \bigl( 1 - e^ {-\frac{2|\gamma_R'|^2}{\sigma^2}} \bigr)  \geq 0.98  e^{-\frac{\kappa n }{\Sigma^2}} \epsilon_0,
\end{align}
with $\gamma_R'\equiv \Re(V^\dagger \gamma)$ and $\gamma_I'\equiv \Im(V^\dagger \gamma)$, and regions $2\sigma^2 <|\gamma_R'|^2\leq \frac{2}{3} \kappa n$ and $0< |\gamma_I'|^2 \leq \frac{1}{3} \kappa n$. 
Since $V^\dagger\gamma$ has the same distribution as $\gamma$, these conditions hold with probability at least $1/2$.

The success probability of Bob's strategy satisfies
\begin{align}
    \Pr(\text{success}) &\geq \frac{1}{2} (1-\Pr_{\gamma'}( 2\sigma^2<|\gamma_R'|^2\leq \tfrac{2}{3} \kappa n
    \cap 
    0 < |\gamma_I'|^2 \leq \tfrac{1}{3}\kappa n )) + 
    \frac{2}{3} \Pr_{\gamma'}( 2\sigma^2 < |\gamma_R'|^2 \leq \tfrac{2}{3} \kappa n 
    \cap 
    0 < |\gamma_I'|^2 \leq \tfrac{1}{3} \kappa n) \\ 
    &=2/3 - \frac{1}{6} (1-\Pr_{\gamma'}( 2\sigma^2 < |\gamma_R'|^2\leq \tfrac{2}{3}\kappa n 
    \cap
    0 < |\gamma_I'|^2 \leq \tfrac{1}{3}\kappa n )) \\
    &\geq 1/2 (1+1/6).
\end{align}
Thus,
\begin{align}\label{eq::tvd-sec7}
    \mathbb{E}_{\gamma} \text{TVD}(p_0,p_{\pm\gamma}) \geq 1/6,
\end{align}
provided a scheme that learns the characteristic function and outputs estimators at any phase space point $\alpha$ with $ |\alpha|^2 \leq \kappa n$ achieves success probability $2/3$ and
\begin{align}\label{eq::distinguishing-5-peak}
    \epsilon \leq 0.49 e^{-\frac{\kappa n }{\Sigma^2}}\epsilon_0
\end{align}
accuracy.

Finally, Bob knows the symmetric unitary matrix $U$ but not the unitary matrix $V$. 
The Takagi decomposition $U = VV^{\T}$ can be obtained by a finite sequence of arithmetic operations and quadratic radicals~\cite{ikramov2012takagi} and is unique up to an orthogonal matrix $O$ such that $U = VO(VO)^{\T} = VV^{\T}$. 
Since for any orthogonal $O$ we have $|\Re((VO)^\dagger \gamma)| = |\Re(V^\dagger \gamma)|$ and $|\Im((VO)^\dagger \gamma)| = |\Im(V^\dagger \gamma)|$, any such decomposition suffices.  
Thus, once $\gamma$ is revealed, Bob can determine $V$ and apply the strategy whenever the components of $V^\dagger\gamma$ fall in the specified ranges.

\begin{figure}[htbp]
  \centering
  % First image on the left
  \includegraphics[width=0.4\textwidth]{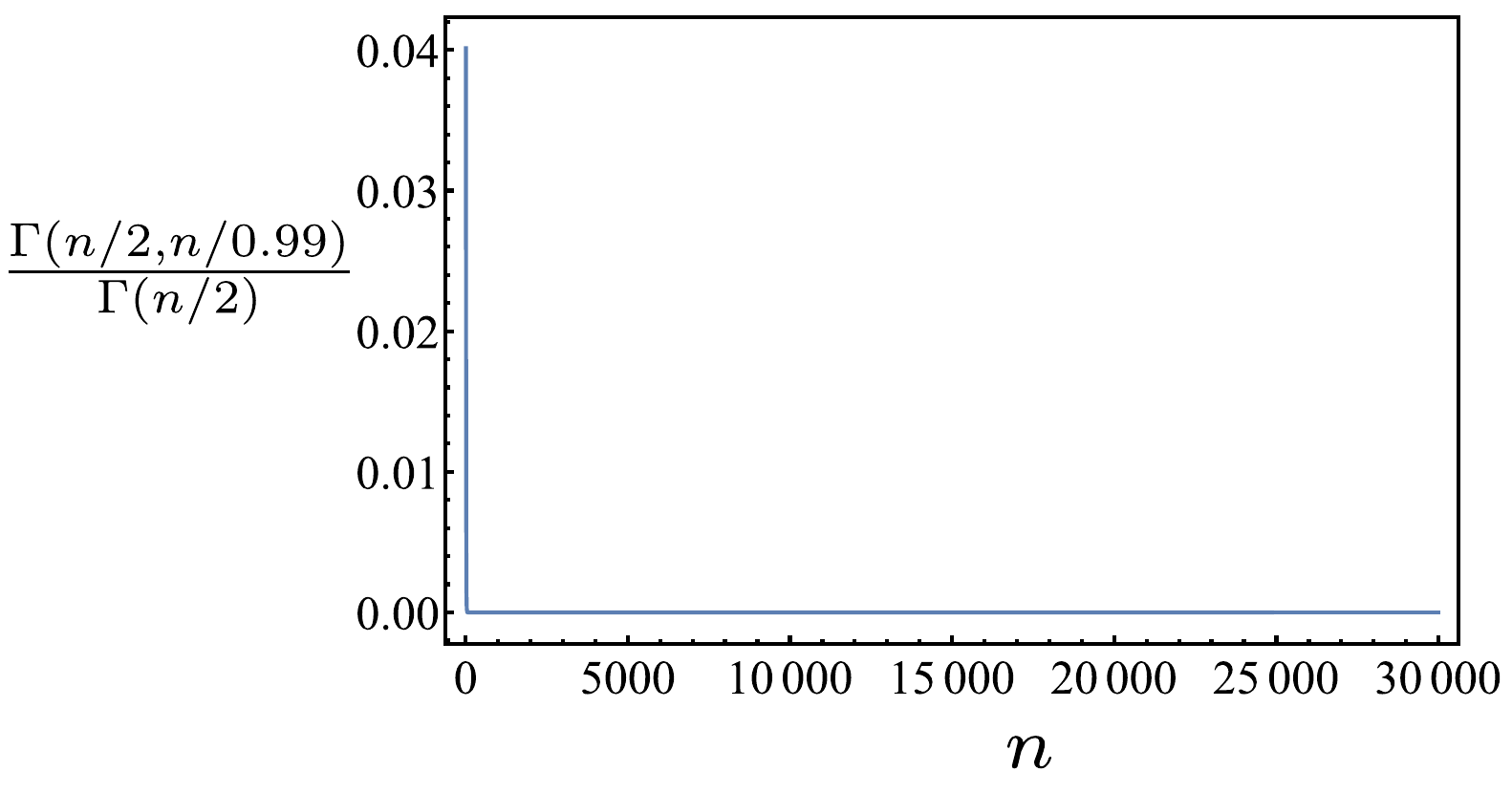}%
  \hspace{0.02\textwidth} % Small gap between the figures
  % Second image on the right
  \includegraphics[width=0.4\textwidth]{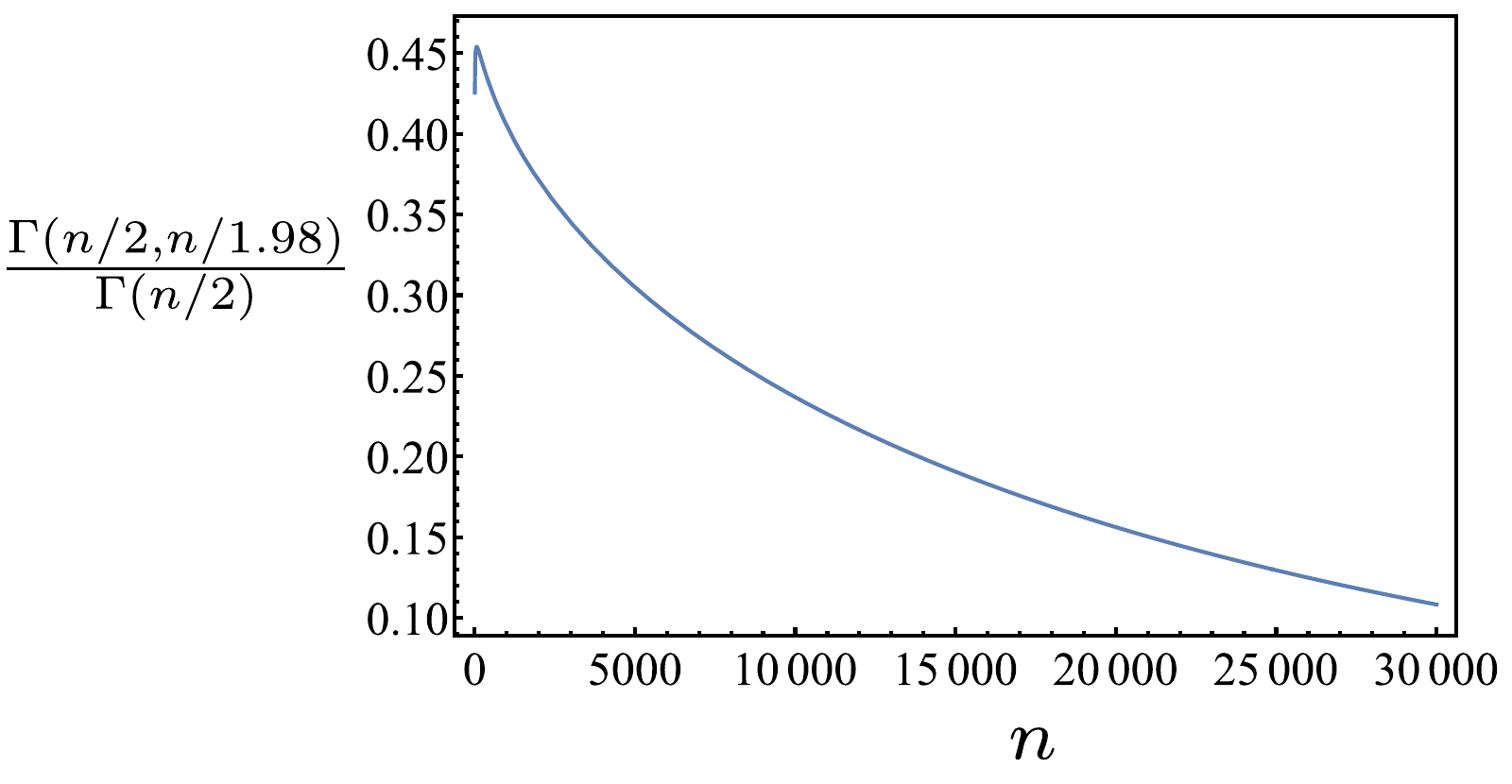}
  \caption{Manual verification that $\frac{\Gamma(n/2,n/0.99)}{\Gamma(n/2)}$ is bounded by $0.0402$ and $\frac{\Gamma(n/2,n/1.98)}{\Gamma(n/2)}$ by $0.4539$ on the $8\leq n\leq30000$ interval.}
  \label{fig::two_bounds}
\end{figure}

\subsection{TVD upper bound}
In the previous subsection, we introduced a partially revealed hypothesis-testing task and reduced it to characteristic-function learning.
We now show that distinguishing the two hypotheses remains exponentially difficult even for reflection-symmetric states by computing the TVD between the corresponding outcome distributions and proving that the single-copy TVD is exponentially small. 

The difference between the two probabilities is written as:
\begin{align}
    &p_0(o_{1:N})-p_{\pm\gamma}(o_{1:N}) \\
    &=p_0(o_{1:N})\left(1-\frac{p_{\pm\gamma}(o_{1:N})}{p_0(o_{1:N})}\right) \\
    &=p_0(o_{1:N})\left(1-\mathbb{E}_{s=\pm1}\prod_{i=1}^N \left(\frac{\frac{1}{\pi^n}\int d^{2n}\alpha_i \chi_{\hat{\sigma}_{s\gamma}}(\alpha_i)\Tr[\hat{D}^\dagger(\alpha)\hat{\Pi}_{o_{1:i-1}}(o_i)]}{\frac{1}{\pi^n}\int d^{2n}\alpha_i \chi_{\hat{\sigma}_0}(\alpha_i)\Tr[\hat{D}^\dagger(\alpha)\hat{\Pi}_{o_{1:i-1}}(o_i)]}\right)\right) \\ 
    &=p_0(o_{1:N})\left(1-\mathbb{E}_{s=\pm1}\prod_{i=1}^N \left(1+\frac{\frac{1}{\pi^n}\int d^{2n}\alpha_i \chi_{\hat{\sigma}_{s\gamma}^\text{(add)}}(\alpha_i)\Tr[\hat{D}^\dagger(\alpha)\hat{\Pi}_{o_{1:i-1}}(o_i)]}{\frac{1}{\pi^n}\int d^2\alpha_i \chi_{\hat{\sigma}_0}(\alpha_i)\Tr[\hat{D}^\dagger(\alpha)\hat{\Pi}_{o_{1:i-1}}(o_i)]}\right)\right) \\ 
    &=p_0(o_{1:N})\left(1-\mathbb{E}_{s=\pm1}\prod_{i=1}^N \left(1-2\epsilon_0 \Im G_\sigma^{o_{\leq i}}(s\gamma)-2\epsilon_0 \Im H_\sigma^{o_{\leq i}}(s\gamma)\right)\right),
\end{align}
where we defined
\begin{align}
    G_\sigma^{o_{\leq i}}(\gamma)
    \equiv \frac{\frac{1}{\pi^n}\int d^{2n}\alpha e^{-\frac{|\alpha|^2}{2\Sigma^2}}e^{-\frac{|\gamma|^2}{2\Sigma^2}}e^{-\frac{|\gamma-\alpha|^2}{2\sigma^2}}G^{o_{\leq i}}(\alpha)}{\frac{1}{\pi^n}\int d^{2n}\alpha  e^{-\frac{|\alpha|^2}{2\Sigma^2}}e^{-\frac{|\alpha|^2}{2\sigma^2}}G^{o_{\leq i}}(\alpha)},~~
    H_\sigma^{o_{\leq i}}(\gamma)
    \equiv \frac{\frac{1}{\pi^n}\int d^{2n}\alpha e^{-\frac{|\alpha|^2}{2\Sigma^2}}e^{-\frac{|\gamma|^2}{2\Sigma^2}}e^{-\frac{|U^{\T}\gamma^*-\alpha|^2}{2\sigma^2}}G^{o_{\leq i}}(\alpha)}{\frac{1}{\pi^n}\int d^{2n}\alpha e^{-\frac{|\alpha|^2}{2\Sigma^2}}e^{-\frac{|\alpha|^2}{2\sigma^2}}G^{o_{\leq i}}(\alpha)},
\end{align}
and
\begin{align}
    G^{o_{\leq i}}(\alpha_i)\equiv \Tr[\hat{D}^\dagger(\alpha)\hat{\Pi}_{o_{1:i-1}}(o_i)].
\end{align}
Note that $G^{o_{\leq i}}_{\sigma}(\alpha)=G^{o_{\leq i}}_{\sigma}(-\alpha)^*$ and $H^{o_{\leq i}}_{\sigma}(\alpha)=H^{o_{\leq i}}_{\sigma}(-\alpha)^*$.
We thus have the following expression:
\begin{align}
    \mathbb{E}_\gamma\text{TVD}(p_0,p_{\pm\gamma})
    =\mathbb{E}_\gamma\sum_{o_{1:N}}p_0(o_{1:N})\max\left\{0,1-\mathbb{E}_{s=\pm1}\prod_{i=1}^N \left(1-2\epsilon_0 \Im G_\sigma^{o_{\leq i}}(s\gamma)-2\epsilon_0 \Im H_\sigma^{o_{\leq i}}(s\gamma)\right)\right\}
\end{align}
Similar to the chain of inequalities in~\eqref{app-eq:inequalities}-\eqref{eq::tvd-ineq-end}, we derive the following inequalities:
\begin{align}
    \mathbb{E}_{s=\pm1}\prod_{i=1}^N \left(1-2\epsilon_0 \Im G_\sigma^{o_{\leq i}}(s\gamma)-2\epsilon_0 \Im H_\sigma^{o_{\leq i}}(s\gamma)\right)
    \geq 1-4\epsilon_0^2\sum_{i=1}^N |G_\sigma^{o_{\leq i}}(\gamma)+H_\sigma^{o_{\leq i}}(\gamma)|^2.
\end{align}
Thus we obtain
\begin{align}\label{eq::star-tvd-upperbound2}
    \mathbb{E}_\gamma\text{TVD}(p_0,p_{\pm\gamma})
    \leq \sum_{o_{1:N}}p_0(o_{1:N})\sum_{i=1}^N 4 \epsilon_0^2\mathbb{E}_\gamma |G_\sigma^{o_{\leq i}}(\gamma)+H_\sigma^{o_{\leq i}}(\gamma)|^2. 
\end{align}
Here, by using the Cauchy-Schwarz inequality,
\begin{align}
    \mathbb{E}_\gamma |G_\sigma^{o_{\leq i}}(\gamma)+H_\sigma^{o_{\leq i}}(\gamma)|^2
    %&=\mathbb{E}_\gamma [|G_\sigma^{o_{\leq i}}(\gamma)|^2+|H_\sigma^{o_{\leq i}}(\gamma)|^2+G_\sigma^{o_{\leq i}*}(\gamma)H_\sigma^{o_{\leq i}}(\gamma)+H_\sigma^{o_{\leq i}*}(\gamma)G_\sigma^{o_{\leq i}}(\gamma)] \\ 
    &\leq 2\mathbb{E}_\gamma [|G_\sigma^{o_{\leq i}}(\gamma)|^2]+2\mathbb{E}_\gamma[|H_\sigma^{o_{\leq i}}(\gamma)|^2].
\end{align}
Since we already have an upper bound for the expectation value of the first term~\eqref{eq::three-peak-bound}, the remaining task is to find an upper bound for the expectation value of the second term:
\begin{align}
    \mathbb{E}_\gamma |H_\sigma^{o_{\leq i}}(\gamma)|^2
    =\mathbb{E}_\gamma\left[\frac{\left|\frac{1}{\pi^n}\int d^{2n}\alpha e^{-\frac{|\alpha|^2}{2\Sigma^2}}e^{-\frac{|\gamma|^2}{2\Sigma^2}}e^{-\frac{|U^{\T}\gamma^*-\alpha|^2}{2\sigma^2}}\Tr[\hat{D}^\dagger(\alpha)\hat{\Pi}_{o_{1:i-1}}(o_i)]\right|^2}{\left|\frac{1}{\pi^n}\int d^{2n}\alpha e^{-\frac{|\alpha|^2}{2\Sigma^2}}e^{-\frac{|\alpha|^2}{2\sigma^2}} \Tr[\hat{D}^\dagger(\alpha)\hat{\Pi}_{o_{1:i-1}}(o_i)]\right|^2}\right].
\end{align}
As can be easily seen, by performing the $\gamma \mapsto U \gamma^*$ mapping, the expectation value of the second term becomes equal to that of the first term.
Thus, we obtain the following:
\begin{align}
    \mathbb{E}_\gamma |G_\sigma^{o_{\leq i}}(\gamma)+H_\sigma^{o_{\leq i}}(\gamma)|^2
    \leq \frac{4}{(1+2\sigma_\gamma^2)^n},
\end{align}
and the averaged TVD is upper bounded by
\begin{align}\label{eq::TVD-sec7}
    \mathbb{E}_\gamma\text{TVD}(p_0,p_{\pm\gamma})
    \leq \frac{16N\epsilon_0^2}{(1+2\sigma_\gamma^2)^n}.
\end{align}

To express the bound in terms of the learning scheme parameters $(\kappa,\epsilon)$, parameterize the five-peak family by
\begin{equation*}
  \Sigma^2  = \frac{\kappa n}{\log(1+\eta_3)},
  \qquad \epsilon_0 = \frac{e^{\frac{\kappa n }{\Sigma^2}}}{0.49} \epsilon = \frac{(1+\eta_3)}{0.49}\epsilon,
  \qquad \eta_3 \in \left( 0, \frac{0.1225}{\epsilon} - 1 \right].
\end{equation*}
Then the distinguishing criterion~\eqref{eq::distinguishing-5-peak} holds for $\epsilon \in (0, 0.1225)$. Moreover, to ensure $\kappa \geq \kappa_{\text{min}} = \frac{3}{0.99} \sigma^2$, it suffices that
\begin{align}
    0.99 \kappa \geq \frac{6\kappa n L}{(\kappa n)^2 - L^2},
\end{align}
with $L\equiv \log(1+\eta_3)$ and using Eq.~\eqref{eq::sigma-Sigma-auxiliary}. As such, any $\eta_3$ within our admissible domain that also satisfies
\begin{align}
    0 < L \leq n \left( -\frac{3}{0.99} + \sqrt{ \left(\frac{3}{0.99}\right)^2 + \kappa^2 }  \right)
\end{align}
suffices.
Lastly, as previously assigned, $\sigma_\gamma^2 = \frac{0.99}{3} \kappa$.

Combining~\eqref{eq::TVD-sec7} with~\eqref{eq::tvd-sec7} gives
\begin{align}
    1/6 \leq  \mathbb{E}_\gamma\text{TVD}(p_0,p_{\pm\gamma})
    \leq 
    \frac{16N\epsilon_0^2}{(1+2\sigma_\gamma^2)^n} 
    = \frac{16 (1+\eta_3)^2}{0.49^2} \frac{N \epsilon^{2}}{(1+0.66\kappa)^n},
\end{align}
so the number of copies obeys
\begin{align}
    N \geq \frac{0.49^2}{96 (1+\eta_3)^2} \epsilon^{-2} (1+0.66\kappa)^n.
\end{align}
Asymptotically,
\begin{align}
    N = \Omega \left( \epsilon^{-2} (1+0.66\kappa)^n \right),
\end{align}
which completes the proof of Thm.~\ref{thm::scaling-single-copy3}.

\section{Lower bound for entanglement-assisted schemes without reflected states}\label{sec::no-reflected-state}

In this section, we investigate state-learning schemes where the learner has access to entangled measurements but lacks access to reflected states. 
Our goal is to clarify entanglement’s role in learning tasks in CV systems and to pinpoint the source of the BM protocol’s advantage.

\begin{theorem}\label{thm::scaling-entangled}
Let $\hat{\rho}$ be an arbitrary $n$-mode quantum state ($n\geq 8$) lacking reflection symmetry, i.e., $\hro \neq \hro^{(U)}$. 
Assume $\epsilon\in(0,0.245)$,  $K \leq 0.22/\epsilon$, and $\kappa \geq 1/0.99$. 
Consider any learning scheme that performs joint measurements on $K$ copies of the input state, i.e., on $\hro^{\otimes K}$.
Suppose the characteristic function $\chi_{\hat{\rho}}(\alpha)$ is learned using $N$ copies of $\hro$, after which the scheme receives a query $\alpha\in \mathbb{C}^n$ with $|\alpha|^2 \leq \kappa n$. 
If the scheme outputs an estimate  $\tilde{\chi}_{\hro}(\alpha)$ of $\chi_{\hat{\rho}}(\alpha)$ such that $|\tilde{\chi}_{\hro}(\alpha)-\chi_{\hat{\rho}}(\alpha)|\leq \epsilon$ with probability at least 2/3, then it requires 
\begin{equation}
N=\Omega \bigl( {K^{-1}\epsilon^{-2}} \big(1+0.99\kappa\big)^{n/2} \bigr) 
\end{equation}
samples to complete the task.
\end{theorem}
We prove the theorem by introducing a partially revealed hypothesis-testing task, similar to the one in Sec.~\ref{sec::revealed_hyp_test}.
This time, Alice prepares only $\hro_0$ or $\hro_{s\gamma}$, depending on the hypothesis, without additionally preparing reflected states associated with any symmetric unitary $U\in\mathbb{C}^{n\times n}$.
Bob measures the states, possibly using entangled measurements, and stores the classical data. 
After $\gamma$ is revealed, Bob's job is to guess Alice's hypothesis.

As in Sec.~\ref{sec::no-entanglement}, we establish two key claims:
(1) Any protocol that learns the characteristic function can be used to win the hypothesis‑testing task; and
(2) Success in that task implies a non‑vanishing TVD between the associated distributions, which in turn requires exponential sample complexity.

Given the lack of access to reflected states (or reflection‑symmetric states), our analysis relies on the previously introduced three‑peak quantum state.
Specifically, we focus on a three-peak state parametrized by $\epsilon_0 = \frac{\epsilon}{0.98}e^{\frac{\kappa n }{\Sigma^2}}$ and $\Sigma^2 = \frac{\kappa n}{\log(1+\eta_3)}$ for an arbitrarily small constant $\eta_3>0$.
Furthermore, Alice samples $\gamma\in\mathbb{C}^{n}$ from the multivariate Gaussian in Eq.~\eqref{eq::multivariate-gaussian} with $2\sigma_\gamma^2=0.99\kappa$ (see Sec.~\ref{sec::revealed_hyp_test} for details and constraints).
Hence, the proof of claim (1) is already given in Sec.~\ref{sec::revealed_hyp_test}; here we establish claim (2), i.e., we show that achieving a non‑vanishing TVD between the two hypotheses, even when the POVM allows joint measurements, requires exponential sample complexity.

We reiterate that although our proof explicitly uses a specific family (the three‑peak‑state family), any valid learning scheme must accurately estimate the characteristic function for any $n$-mode input state.
Consequently, the exponential sample complexity established here provides a general lower bound applicable to all such learning schemes.
\subsection{TVD upper bound}
For the second claim, we now show that any scheme performing entangled measurements on $K$ copies at a time (i.e., on $\hro^{\otimes K}$) requires exponentially many samples to distinguish between the two hypotheses.
The measurement probability is given by:
\begin{align}
    p(o_{1:N_b}) = \prod_{i=1}^{N_b} \Tr[\hat{\rho}^{\otimes K}\hat{\Pi}_{o_{1:i-1}}(o_i)],
\end{align}
where $N_b$ represents the number of $K$-copy batches, and $N = K N_b$ is the total number of copies. 
First, we lower-bound the ratio of the output probabilities as
\begin{align}
    \frac{p_{\pm\gamma}(o_{1:N_b})}{p_0(o_{1:N_{b}})}
    =\mathbb{E}_{s=\pm1}\prod_{i=1}^{N_{b}} F^i_{\gamma,s}
    \geq \prod_{i=1}^{N_{b}}\sqrt{F^i_{\gamma,+}F^i_{\gamma,-}}
    =\prod_{i=1}^{N_{b}}\sqrt{1 - G^i_{\gamma}}
    \geq \prod_{i=1}^{N_{b}}\sqrt{1-\max(0,G^i_{\gamma})}
    &\geq \prod_{i=1}^{N_{b}} \left( 1 - \max(0,G^i_{\gamma}) \right) \\
    &\geq \left( 1 - \sum_{i=1}^{N_{b}} \max(0,G^i_{\gamma}) \right),
\end{align}
where we defined
\begin{align}
    F^i_{\gamma,s}\equiv \frac{\Tr[\hat{\rho}_{s\gamma}^{\otimes K}\hat{\Pi}_{o_{1:i-1}}(o_i)]}{\Tr[\hat{\rho}_0^{\otimes K}\hat{\Pi}_{o_{1:i-1}}(o_i)]},~~~
    %1-G_\gamma^i\equiv F^i_{\gamma,+}F^i_{\gamma,-}, \\ 
    G_\gamma^i\equiv 1-F^i_{\gamma,+}F^i_{\gamma,-}=1-\frac{\Tr[\hat{\rho}_{+\gamma}^{\otimes K}\hat{\Pi}_{o_{1:i-1}}(o_i)]}{\Tr[\hat{\rho}_0^{\otimes K}\hat{\Pi}_{o_{1:i-1}}(o_i)]}\frac{\Tr[\hat{\rho}_{-\gamma}^{\otimes K}\hat{\Pi}_{o_{1:i-1}}(o_i)]}{\Tr[\hat{\rho}_0^{\otimes K}\hat{\Pi}_{o_{1:i-1}}(o_i)]},
\end{align}
and used the AM-GM inequality along with the non-negativity of the real probability ratios $F^i_{\gamma,s}$ to derive the first inequality.
The third inequality uses the fact that $\sqrt{1-x} \geq (1-x)$, $\forall 0\leq x\leq1$. The fourth inequality uses the fact that $\prod_{i}(1-x_{i}) \geq 1 - \sum_{i} x_{i}$, $\forall~0\leq x\leq1$.
Thus, under the two hypotheses, the difference between the outcome probabilities is upper-bounded as
\begin{align}
    p_0(o_{1:N_{b}})-p_{\pm\gamma}(o_{1:N_{b}}) 
    =p_0(o_{1:N_{b}})\left(1-\frac{p_{\pm\gamma}(o_{1:N_{b}})}{p_0(o_{1:N_{b}})}\right)
    \leq p_0(o_{1:N_{b}})  \sum_{i=1}^{N_{b}} \max(0,G^i_{\gamma}).
\end{align}
\iffalse
Under the two hypotheses, we compute the TVD:
\begin{align}
    p_0(o_{1:N})-p_{\pm\gamma}(o_{1:N}) 
    &=p_0(o_{1:N})\left(1-\frac{p_{\pm\gamma}(o_{1:N})}{p_0(o_{1:N})}\right) \\
    %=p_0(o_{1:N})\left(1-\mathbb{E}_{s=\pm1}\prod_{i=1}^N \frac{\Tr[\hat{\rho}_{s\gamma}^{\otimes K}\hat{\Pi}_{o_i}^{o_{1:i-1}}]}{\Tr[\hat{\rho}_0^{\otimes K}\hat{\Pi}_{o_i}^{o_{1:i-1}}]}\right) \\
    &=p_0(o_{1:N})\left(1-\mathbb{E}_{s=\pm1}\prod_{i=1}^N F^i_{\gamma,s}\right) \\ 
    &\leq p_0(o_{1:N})\left(1-\prod_{i=1}^{N}\sqrt{F^i_{\gamma,+}  F^i_{\gamma,-}}\right) \\ 
    & = p_0(o_{1:N})\left(1- \prod_{i=1}^{N}\sqrt{ 1 - G^i_{\gamma}  } \right) \\ 
    & \leq p_0(o_{1:N})\left(1- \prod_{i=1}^{N}\sqrt{ 1 - \max(0,G^i_{\gamma})  }\right) \\ 
    & \leq p_0(o_{1:N})\left(1- \prod_{i=1}^{N} \left( 1 - \max(0,G^i_{\gamma}) \right)\right) \\ 
    & \leq p_0(o_{1:N})\left(1- \left( 1 - \sum_{i=1}^{N} \max(0,G^i_{\gamma}) \right)\right) \\ 
    & = p_0(o_{1:N})  \sum_{i=1}^{N} \max(0,G^i_{\gamma})
\end{align}
\fi
Averaging over $\gamma$, the expected TVD is upper‑bounded by
\begin{align}
    \mathbb{E}_\gamma \text{TVD}[p_0,p_{\pm\gamma}]
    \leq  \sum_{o_{1:N_{b}}}p_0(o_{1:N_{b}}) \mathbb{E}_\gamma\left[\max\left\{0, \sum_{i=1}^{N_{b}} \max(0, G^i_{\gamma}) \right\}\right] = \sum_{o_{1:N_{b}}}p_0(o_{1:N_{b}}) \mathbb{E}_\gamma\left[  \sum_{i=1}^{N_{b}} \max(0, G^i_{\gamma}) \right].
\end{align}
Thus, as in the previous sections, finding an upper bound for $\mathbb{E}_\gamma[\max\{0, G^i_{\gamma}\}]$ allows us to determine a lower bound on sample complexity.
\iffalse
Recall that our three-peak state is defined as
\begin{align}
    \hat{\rho}_\gamma=(1-\nu^2)\left[ \nu^{\hat{n}}(\hat{D}(0)+2i\epsilon_0\hat{D}(\gamma)-2i\epsilon_0\hat{D}^\dagger(\gamma))\nu^{\hat{n}} \right].
\end{align}
Here $\mathcal{N}$ is the normalization factor and for any vector $|\psi\rangle$, \co{Remove?}
\begin{align}
    \langle \psi|\hat{\rho}_\gamma^{\otimes K}|\psi\rangle
    &=(1-\nu^2)^{-K}\langle \psi|(\nu^{\hat{n}}(\hat{D}(0)+2i\epsilon_0\hat{D}(\gamma)-2i\epsilon_0\hat{D}^\dagger(\gamma))\nu^{\hat{n}})^{\otimes K}|\psi\rangle \\
    &\geq(1-\nu^2)^{-K}\langle \psi|(\nu^{\hat{n}}(\hat{D}(0)-4\epsilon_0\hat{D}(0))\nu^{\hat{n}})^{\otimes K}|\psi\rangle \\
    &=(1-4\epsilon_0)^K(1-\nu^2)^{-K}\langle \psi|\nu^{\hat{n}}\hat{D}(0)\nu^{\hat{n}}|\psi\rangle \\ 
    &=(1-4\epsilon_0)^K\langle \psi|\hat{\rho}_0^{\otimes K}|\psi\rangle.
\end{align}
For $F^i_{\gamma,s} \geq 0$ to hold, assuming $\epsilon_0\leq 1/4$, any non-negative integer suffices for $K$:
\begin{align}
    \frac{\langle \psi|\hat{\rho}_\gamma^{\otimes K}|\psi\rangle}{\langle \psi|\hat{\rho}_0^{\otimes K}|\psi\rangle}
    \geq (1-4\epsilon_0)^K
    \geq 0.
\end{align}
\fi
Let us proceed by upper-bounding $\mathbb{E}_\gamma[\max\{0,G_\gamma^i\}]$.
Recall that
\begin{align}
    G_\gamma^i
    &=1-\frac{\Tr[\hat{\rho}_{+\gamma}^{\otimes K}\hat{\Pi}_{o_{1:i-1}}(o_i)]}{\Tr[\hat{\rho}_0^{\otimes K}\hat{\Pi}_{o_{1:i-1}}(o_i)]}\frac{\Tr[\hat{\rho}_{-\gamma}^{\otimes K}\hat{\Pi}_{o_{1:i-1}}(o_i)]}{\Tr[\hat{\rho}_0^{\otimes K}\hat{\Pi}_{o_{1:i-1}}(o_i)]} \\
    &=\frac{\Tr[(\hat{\rho}_0^{\otimes K}\otimes \hat{\rho}_0^{\otimes K})(\hat{\Pi}_{o_{1:i-1}}(o_i)\otimes \hat{\Pi}_{o_{1:i-1}}(o_i))]-\Tr[(\hat{\rho}_{+\gamma}^{\otimes K}\otimes \hat{\rho}_{-\gamma}^{\otimes K})(\hat{\Pi}_{o_{1:i-1}}(o_i)\otimes \hat{\Pi}_{o_{1:i-1}}(o_i))]}{\Tr[(\hat{\rho}_0^{\otimes K}\otimes \hat{\rho}_0^{\otimes K})(\hat{\Pi}_{o_{1:i-1}}(o_i)\otimes \hat{\Pi}_{o_{1:i-1}}(o_i))]}.
\end{align}
From now on, we omit the indices of $\hat{\Pi}_{o_{1:i-1}}(o_i)$ for simplicity.
Defining $\hat{\rho}_{\pm\gamma}\equiv \hat{\rho}_0\pm \hat{\rho}_\gamma^{\text{(add)}}$ and using the fact that $\hat{\rho}^{\text{(add)}}_\gamma=-\hat{\rho}^{\text{(add)}}_{-\gamma}$, the numerator can be simplified as
\begin{align}
    &\Tr[(\hat{\rho}_{0}^{\otimes K}\otimes \hat{\rho}_{0}^{\otimes K})(\hat{\Pi}\otimes \hat{\Pi})]-\Tr[((\hat{\rho}_0+\hat{\rho}^{\text{(add)}}_{\gamma})^{\otimes K}\otimes (\hat{\rho}_0-\hat{\rho}^{\text{(add)}}_{\gamma})^{\otimes K})(\hat{\Pi}\otimes\hat{\Pi})] \\
    =&\Tr[(\hat{\rho}_{0}^{\otimes K}\otimes \hat{\rho}_{0}^{\otimes K})(\hat{\Pi}\otimes \hat{\Pi})]-\Tr[\left(\sum_{x,y\in\{0,1\}^K} (-1)^{|y|}\hat{\omega}_{x}\otimes \hat{\omega}_y \right)(\hat{\Pi}\otimes \hat{\Pi})] \\ 
    =&\Tr[(\hat{\rho}_{0}^{\otimes K}\otimes \hat{\rho}_{0}^{\otimes K})(\hat{\Pi}\otimes \hat{\Pi})]-\Tr[\left(\hat{\rho}_0^{\otimes K}+\sum_{x\in\{0,1\}^K:|x|>0:\text{even}}\hat{\omega}_x\right)\hat{\Pi}]^2+\Tr[\left(\sum_{x\in\{0,1\}^K:|x|:\text{odd}}\hat{\omega}_x\right)\hat{\Pi}]^2 \\ 
    =&-2\Tr[\left(\hat{\rho}_0^{\otimes K}\otimes \sum_{x\in\{0,1\}^K:|x|>0:\text{even}}\hat{\omega}_x\right)(\hat{\Pi}\otimes \hat{\Pi})]-\Tr[\left(\sum_{x\in\{0,1\}^K:|x|>0:\text{even}}\hat{\omega}_x\right)\hat{\Pi}]^2+\Tr[\left(\sum_{x\in\{0,1\}^K:|x|:\text{odd}}\hat{\omega}_x\right)\hat{\Pi}]^2 \\ 
    \leq& 2\left|\Tr[\left(\hat{\rho}_0^{\otimes K}\otimes \sum_{x\in\{0,1\}^K:|x|>0:\text{even}}\hat{\omega}_x\right)(\hat{\Pi}\otimes \hat{\Pi})]\right|+\Tr[\left(\sum_{x\in\{0,1\}^K:|x|:\text{odd}}\hat{\omega}_x\right)\hat{\Pi}]^2,
\end{align}
where $\hat{\omega}_x$ represents a Hermitian operator with the $x_i=0$ parts being $\hat{\rho}_0$ and the $x_i=1$ parts being $\hat{\rho}^{\text{(add)}}_\gamma$. For the second equality, note that if $x$ and $y$ have different parity, the exchange of $\hat{\omega}_x\otimes \hat{\omega}_y$ and $\hat{\Pi}\otimes \hat{\Pi}$ cancels out due to the $(-1)^{|y|}$ term.

Therefore, using the Cauchy-Schwarz inequality, we obtain
\begin{align}
    %G^i_{\gamma} &\leq 2 \sqrt{\frac{\Tr[\left( \sum_{x\in\{0,1\}^{K}:|x|>0:\text{even}} \hat{\omega}_{x}  \right)\hat{\Pi}]^{2}}{ \Tr[\hat{\rho}_{0}^{\otimes K}\hat{\Pi}]^{2}}} + \frac{\Tr[\left(\sum_{x\in\{0,1\}^K:|x|:\text{odd}}\hat{\omega}_x\right)\hat{\Pi}]^2}{\Tr[\hat{\rho}_{0}^{\otimes K}\hat{\Pi}]^{2} }, \\
    \mathbb{E}_{\gamma}[\max(0,G^i_{\gamma})] &\leq 2 \sqrt{ \mathbb{E}_{\gamma} \left[ \frac{\Tr[\left( \sum_{x\in\{0,1\}^{K}:|x|>0:\text{even}} \hat{\omega}_{x}  \right) \hat{\Pi}]^{2}}{\Tr[\hat{\rho}_{0}^{\otimes K}\hat{\Pi}]^{2}} \right]} + \mathbb{E}_{\gamma}\left[ \frac{\Tr[\left(\sum_{x\in\{0,1\}^K:|x|:\text{odd}}\hat{\omega}_x\right)\hat{\Pi}]^2}{\Tr[\hat{\rho}_{0}^{\otimes K}\hat{\Pi}]^{2} } \right]. \label{eq::both-terms-bound}
\end{align}
We proceed by upper-bounding the first term on the right-hand side: 
\begin{align}\label{eq::former-term-expectation}
    &\frac{\mathbb{E}_\gamma \left[\Tr[\left(\sum_{x,y\in\{0,1\}^K:|x|,|y|>0:\text{even}}\hat{\omega}_x \otimes \hat{\omega}_y\right)(\hat{\Pi}\otimes\hat{\Pi})]\right]}{\Tr[(\hat{\rho}_{0}^{\otimes K} \otimes \hat{\rho}_{0}^{\otimes K})(\hat{\Pi}\otimes\hat{\Pi})]}.
\end{align}
Let us fix $x$ and $y$ so that $x=(1,\dots,1,0,\dots,0)$ where the first $k$ entries are equal to one and $y=(1,\dots,1,0,\dots,0)$ where the first $l$ entries are equal to one.
It will be clear that the same upper bound applies for different orderings, i.e., $k$ and $l$ determine the upper bound.
For these fixed $x$ and $y$, using the expansion of $\hat{\omega}_x$ and $\hat{\omega}_y$ in terms of displacement operators, we have
\begin{align}
    &\mathbb{E}_\gamma\left[\Tr[(\hat{\omega}_x\otimes \hat{\omega}_y)(\hat{\Pi}\otimes\hat{\Pi})]\right] \\
    &=(2i\epsilon_0)^{k+l}\int \frac{d^{2nK}\alpha d^{2nK}\beta}{\pi^{2nK}} \prod_{i=1}^K \left(e^{-\frac{|\alpha_i|^2+|\beta_i|^2}{2\Sigma^2}}\right)\prod_{i=k+1}^K e^{-\frac{|\alpha_i|^2}{2\sigma^2}} \prod_{i=l+1}^K e^{-\frac{|\beta_i|^2}{2\sigma^2}} \Tr[(\hat{D}^\dagger(\alpha)\otimes \hat{D}^\dagger(\beta))(\hat{\Pi}\otimes\hat{\Pi})] \\
    &\times \mathbb{E}_\gamma\left[\prod_{i=1}^ke^{-\frac{\abs{\gamma}^{2}}{2\Sigma^{2}}} \left(e^{-\frac{|\alpha_i-\gamma|^2}{2\sigma^2}}-e^{-\frac{|\alpha_i+\gamma|^2}{2\sigma^2}}\right)\prod_{i=1}^l\ e^{-\frac{\abs{\gamma}^{2}}{2\Sigma^{2}}} \left(e^{-\frac{|\beta_i-\gamma|^2}{2\sigma^2}}-e^{-\frac{|\beta_i+\gamma|^2}{2\sigma^2}}\right)\right] \\
    &=(2i\epsilon_0)^{k+l}\int \frac{d^{2nK}\alpha d^{2nK}\beta}{\pi^{2nK}} \prod_{i=1}^K \left(e^{-\frac{|\alpha_i|^2+|\beta_i|^2}{2\Sigma^2}}\right)\prod_{i=k+1}^K e^{-\frac{|\alpha_i|^2}{2\sigma^2}} \prod_{i=l+1}^K e^{-\frac{|\beta_i|^2}{2\sigma^2}} \Tr[(\hat{D}^\dagger(\alpha)\otimes \hat{D}^\dagger(\beta))(\hat{\Pi}\otimes\hat{\Pi})] \\
    &\times \left[ \sum_{t\in \{0,1\}^{(k+l)}} (-1)^{|t|}  \mathbb{E}_\gamma\left[e^{-\frac{\abs{\gamma}^{2}}{2\Sigma^{2}}(k+l)} \prod_{i=1}^k  e^{-\frac{|(-1)^{t_{i}}\alpha_i-\gamma|^2}{2\sigma^2}}  \prod_{i=1}^l \left(e^{-\frac{|(-1)^{t_{k+i}} \beta_i-\gamma|^2}{2\sigma^2}}\right)\right]\right].
\end{align}
The above expectation value is decomposed into a sum of $2^{k+l}$ terms. 
Note that
\begin{align}
    \mathbb{E}_\gamma\left[e^{-\frac{\abs{\gamma}^{2}}{2\Sigma^{2}} (k+l)}\prod_{i=1}^ke^{-\frac{|\alpha_i-\gamma|^2}{2\sigma^2}} \prod_{i=1}^l e^{-\frac{|\beta_i-\gamma|^2}{2\sigma^2}}\right]
    =\left(\frac{\sigma^2}{\sigma^2+(k+l)\sigma_\gamma'^2}\right)^ne^{\frac{|\omega|^2}{2\Gamma^2}}\prod_{i=1}^k e^{-\frac{|\alpha_i|^2}{2\sigma^2}} \prod_{i=1}^l e^{-\frac{|\beta_i|^2}{2\sigma^2}},
\end{align}
where $\omega\equiv \sum_{i=1}^k \alpha_i+\sum_{i=1}^l \beta_i$,  $1/\Gamma^2\equiv \sigma_\gamma'^2/(\sigma^4+(k+l)\sigma^2\sigma_\gamma'^2)$ with $\sigma_{\gamma}'^2 \equiv \sigma_\gamma^2/[1+\sigma_\gamma^2(k+l)/\Sigma^2]$ (the $k$ and $l$ dependence of $\sigma_\gamma'^2$ is implicit).

We evaluate one of the $2^{k+l}$ terms (again each of the $2^{k+l}$ has the same upper bound, which will be clear in the derivation), while omitting the $(2i\epsilon_0)^{k+l}$ prefactor,
\begin{align}
    &\int \frac{d^{2nK}\alpha d^{2nK}\beta }{\pi^{2nK}}\prod_{i=1}^K e^{-\frac{|\alpha_i|^2+|\beta_i|^2}{2\Sigma^2}}\prod_{i=k+1}^K e^{-\frac{|\alpha_i|^2}{2\sigma^2}} \prod_{i=l+1}^K e^{-\frac{|\beta_i|^2}{2\sigma^2}} \mathbb{E}_\gamma\left[e^{-\frac{\abs{\gamma}^{2}}{2\Sigma^{2}}(k+l)}\prod_{i=1}^ke^{-\frac{|\alpha_i-\gamma|^2}{2\sigma^2}} \prod_{i=1}^l e^{-\frac{|\beta_i-\gamma|^2}{2\sigma^2}}\right]  \\ 
    &~~~\times \Tr[(\hat{D}^\dagger(\alpha)\otimes \hat{D}^\dagger(\beta))(\hat{\Pi}\otimes\hat{\Pi})] \\
    &=\left(\frac{\sigma^2}{\sigma^2+(k+l)\sigma_\gamma'^2}\right)^n\int \frac{d^{2nK}\alpha d^{2nK}\beta }{\pi^{2nK}} e^{\frac{|\omega|^2}{2\Gamma^2}}\prod_{i=1}^K e^{-\frac{|\alpha_i|^2+|\beta_i|^2}{2\Sigma^2}}\prod_{i=k+1}^K e^{-\frac{|\alpha_i|^2}{2\sigma^2}} \prod_{i=l+1}^K e^{-\frac{|\beta_i|^2}{2\sigma^2}}\prod_{i=1}^k e^{-\frac{|\alpha_i|^2}{2\sigma^2}}\prod_{i=1}^l e^{-\frac{|\beta_i|^2}{2\sigma^2}} \\
    &~~~\times \Tr[(\hat{D}^\dagger(\alpha)\otimes \hat{D}^\dagger(\beta))(\hat{\Pi}\otimes\hat{\Pi})] \\  \notag
    &=\left(\frac{\sigma^2}{\sigma^2+(k+l)\sigma_\gamma'^2}\right)^n\int \frac{d^{2nK}\alpha d^{2nK}\beta }{\pi^{2nK}} e^{\frac{(k+l)|\alpha_1|^2}{2\Gamma^2}}\prod_{i=1}^K e^{-\frac{|\alpha_i|^2+|\beta_i|^2}{2\Sigma^2}}\prod_{i=k+1}^K e^{-\frac{|\alpha_i|^2}{2\sigma^2}} \prod_{i=l+1}^K e^{-\frac{|\beta_i|^2}{2\sigma^2}} \prod_{i=1}^k e^{-\frac{|\alpha_i|^2}{2\sigma^2}}\prod_{i=1}^l e^{-\frac{|\beta_i|^2}{2\sigma^2}} \\
    &~~~\times \Tr[\hat{U}_\text{LON}(\hat{D}^\dagger(\alpha)\otimes \hat{D}^\dagger(\beta))\hat{U}_\text{LON}^\dagger(\hat{\Pi}\otimes \hat{\Pi})] \\ \notag
    &=\left(\frac{\sigma^2}{\sigma^2+(k+l)\sigma_\gamma'^2}\right)^n\int \frac{d^{2nK}\alpha d^{2nK}\beta }{\pi^{2nK}} e^{-\frac{|\alpha_1|^2}{2(\sigma^2+(k+l)\sigma_\gamma'^2)}}\prod_{i=1}^K e^{-\frac{|\alpha_i|^2+|\beta_i|^2}{2\Sigma^2}} \prod_{i=k+1}^K e^{-\frac{|\alpha_i|^2}{2\sigma^2}} \prod_{i=l+1}^K e^{-\frac{|\beta_i|^2}{2\sigma^2}}\prod_{i=2}^k e^{-\frac{|\alpha_i|^2}{2\sigma^2}}\prod_{i=1}^l e^{-\frac{|\beta_i|^2}{2\sigma^2}}  \\
    &~~~\times \Tr[(\hat{D}^\dagger(\alpha)\otimes \hat{D}^\dagger(\beta))\hat{\Pi}'] \\  \notag
    &=\left(\frac{\sigma^2}{\sigma^2+(k+l)\sigma_\gamma'^2}\right)^n \left(\frac{2}{2b+1}\right)^{n}\left(\frac{2}{2a+1}\right)^{n(2K-1)}\Tr[\left(\left(\frac{2b-1}{2b+1}\right)^{\hat{N}_{1:1}}\otimes \left(\frac{2a-1}{2a+1}\right)^{\hat{N}_{2:2K}}\right)\hat{\Pi}'].
\end{align}
For the second equality, for $(\alpha_1,\dots,\alpha_k,\beta_1,\dots,\beta_l)$ we use a unitary whose first‑row entries are all $1/\sqrt{k+l}$ to perform the change of variables.
To account for the change of variables on the displacement operators, we introduce the corresponding linear-optical network unitary $\hat{U}_{\text{LON}}$ and defined ${\hat{\Pi}'\equiv \hat{U}_\text{LON}^\dagger (\hat{\Pi}\otimes \hat{\Pi})\hat{U}_\text{LON}}$.
For the last equality, we defined
\begin{align}
    a\equiv \frac{1}{2\Sigma^2}+\frac{1}{2\sigma^2}>\frac{1}{2}, ~~b \equiv \frac{1}{2\Sigma^2}+\frac{1}{2(\sigma^2+(k+l)\sigma_\gamma'^2)}>0,
\end{align}
and again used the expansion of the operator $z^{\hat{N}}$ in terms of displacement operators.
The denominator in Eq.~\eqref{eq::former-term-expectation} corresponds to the $k=l=0$ case, and thus $a=b$.
Therefore, the fraction becomes
\begin{align}
    \left(\frac{\sigma^2}{\sigma^2+(k+l)\sigma_\gamma'^2}\right)^n \left(\frac{2a+1}{2b+1}\right)^n \frac{\Tr[\left(\left(\frac{2b-1}{2b+1}\right)^{\hat{N}_{1:1}}\otimes \left(\frac{2a-1}{2a+1}\right)^{\hat{N}_{2:2K}}\right)\hat{\Pi}']}{\Tr[\left(\left(\frac{2a-1}{2a+1}\right)^{\hat{N}_{1:1}}\otimes \left(\frac{2a-1}{2a+1}\right)^{\hat{N}_{2:2K}}\right)\hat{\Pi}']},
\end{align}
where the prefactor simplifies to
\begin{align}
    \left(\frac{\sigma^2}{\sigma^2+(k+l)\sigma_\gamma'^2}\right)^n \left(\frac{2a+1}{2b+1}\right)^n= \left(\frac{\sigma^2+\Sigma^2+\sigma^2\Sigma^2}{\sigma^2+\Sigma^2+\sigma^2\Sigma^2+(k+l)(1+\Sigma^2)\sigma_\gamma'^2}\right)^n =\left(\frac{1+\nu}{1+\nu+2\nu(k+l) \sigma_\gamma'^2}\right)^n.
\end{align}
To bound the expression, note that $b=1/2$ when $(k+l)\sigma_\gamma^2=1$.
We therefore analyze the prefactor across two regions.
Region 1 is defined by $0<(k+l)\sigma_\gamma^2<1 \iff b>1/2$. 
In this case, since the associated operators are positive semi-definite and $b< a$ for $0<\nu<1$, we have
\begin{align}
    \frac{\Tr[\left(\left(\frac{2b-1}{2b+1}\right)^{\hat{N}_{1:1}}\otimes \left(\frac{2a-1}{2a+1}\right)^{\hat{N}_{2:2K}}\right)\hat{\Pi}']}{\Tr[\left(\left(\frac{2a-1}{2a+1}\right)^{\hat{N}_{1:1}}\otimes \left(\frac{2a-1}{2a+1}\right)^{\hat{N}_{2:2K}}\right)\hat{\Pi}']} \leq 1.
\end{align}
Region 2 is defined by $(k+l)\sigma_\gamma^2\geq1 \iff 0<b\leq1/2$. We bound the term as follows:
\begin{align}
    \frac{\Tr[\left(\left(\frac{2b-1}{2b+1}\right)^{\hat{N}_{1:1}}\otimes \left(\frac{2a-1}{2a+1}\right)^{\hat{N}_{2:2K}}\right)\hat{\Pi}']}{\Tr[\left(\left(\frac{2a-1}{2a+1}\right)^{\hat{N}_{1:1}}\otimes \left(\frac{2a-1}{2a+1}\right)^{\hat{N}_{2:2K}}\right)\hat{\Pi}']}&\leq \frac{\Tr[\left(\abs{\frac{2b-1}{2b+1}}^{\hat{N}_{1:1}}\otimes \left(\frac{2a-1}{2a+1}\right)^{\hat{N}_{2:2K}}\right)\hat{\Pi}']}{\Tr[\left(\left(\frac{2a-1}{2a+1}\right)^{\hat{N}_{1:1}}\otimes \left(\frac{2a-1}{2a+1}\right)^{\hat{N}_{2:2K}}\right)\hat{\Pi}']} = \frac{\Tr[\left(\left(\frac{1-2b}{1+2b}\right)^{\hat{N}_{1:1}}\otimes \left(\frac{2a-1}{2a+1}\right)^{\hat{N}_{2:2K}}\right)\hat{\Pi}']}{\Tr[\left(\left(\frac{2a-1}{2a+1}\right)^{\hat{N}_{1:1}}\otimes \left(\frac{2a-1}{2a+1}\right)^{\hat{N}_{2:2K}}\right)\hat{\Pi}']} \leq 1,
\end{align}
where the last inequality requires $\frac{2a-1}{2a+1} - \frac{1-2b}{1+2b} = \frac{2\nu^2 \left( (1+\nu) - (1-\nu)\sigma'^{2}_{\gamma}(k+l) \right)}{(1+\nu) + 2\nu\sigma'^{2}_{\gamma}(k+l)} \geq 0$. 
Equivalently, $(1+\nu) - (1-\nu)\sigma'^{2}_{\gamma}(k+l) = (1+\nu)\left(1 - \frac{1}{1+ \frac{1+\nu}{(1-\nu)(k+l)\sigma_{\gamma}^{2}}}\right) \geq 0 $, which holds for all $0<\nu<1$.
Noting that $2\leq k+l\leq 2K$, we define two regions in which the bound holds for any term of the expectation value.
\begin{align}
&\text{Region 1: } b>1/2 \iff \sigma_{\gamma}^2 < \frac{1}{2K},   \\
&\text{Region 2: } 0<b\leq1/2 \iff \sigma_{\gamma}^2 \geq \frac{1}{2}.
\end{align}

Note that there are $2^{k+l}$ similar terms, each with the same upper bound. Before proceeding to bound both terms in Eq.~$\eqref{eq::both-terms-bound}$, we define $\sigma^2_{\gamma,(k,l)} \equiv (k+l) \sigma'^2_{\gamma} = \frac{1}{\frac{1}{\Sigma^2}+ \frac{1}{(k+l)\sigma_\gamma^2}}$.
\begin{align}
    \mathbb{E}_\gamma\left[\frac{ \Tr[\left(\sum_{x\in\{0,1\}^K:|x|>0:\text{even}}\hat{\omega}_x\right)\hat{\Pi}] ^2}{\Tr[\hat{\rho}_{0}^{\otimes K}\hat{\Pi}]^2}\right] 
    &= \sum_{x\in\{0,1\}^K:|x|>0:\text{even}}\sum_{y\in\{0,1\}^K:|y|>0:\text{even}} \frac{\mathbb{E}_\gamma\left[\Tr[(\hat{\omega}_x\otimes \hat{\omega}_y)(\hat{\Pi}\otimes\hat{\Pi})]\right]}{\Tr[\hat{\rho}_{0}^{\otimes K}\hat{\Pi}]^2} \\ 
    &\leq \sum_{k=2,\text{even}}^{K}\sum_{l=2,\text{even}}^K\binom{K}{k}\binom{K}{l}(4\epsilon_0)^{k+l}\left(\frac{1+\nu}{1+\nu+2\nu(k+l) {\sigma'}_\gamma^2}\right)^n \\ 
    &\leq \left(\frac{1+\nu}{1+\nu+2\nu  \sigma^2_{\gamma,(2,2)} } \right)^n\sum_{k=2,\text{even}}^{K}\sum_{l=2,\text{even}}^K\binom{K}{k}\binom{K}{l}(4\epsilon_0)^{k+l} \\ 
    &\leq \left(\frac{1+\nu}{1+\nu+2\nu  \sigma^2_{\gamma,(2,2)}}\right)^n\left(\sum_{k=2,\text{even}}^{K}(4K\epsilon_0)^{k}\right)^2 \\ 
    &\leq \left(\frac{1+\nu}{1+\nu+2\nu  \sigma^2_{\gamma,(2,2)}}\right)^n\frac{(4K\epsilon_0)^4}{(1-(4K\epsilon_0)^2 )^2} \leq \left(\frac{1+\nu}{1+\nu+2\nu  \sigma^2_{\gamma,(2,2)}}\right)^n\frac{(4K\epsilon_0)^4}{(1-4K\epsilon_0 )^2},
\end{align}
where we used $\binom{K}{k}\leq K^k$, and assumed $4K\epsilon_0 < 1$.

Let us now turn to the second term of Eq.~\eqref{eq::both-terms-bound}, which is obtained in a similar fashion:
\begin{align}
    \mathbb{E}_\gamma\left[\frac{\Tr[\left(\sum_{x\in\{0,1\}^K:|x|:\text{odd}}\hat{\omega}_x\right)\hat{\Pi}]^2}{\Tr[\hat{\rho}_0^{\otimes K}\hat{\Pi}]^2}\right]
    &= \sum_{x\in\{0,1\}^K:|x|:\text{odd} }\sum_{y\in\{0,1\}^K:|y|:\text{odd}}\frac{\mathbb{E}_\gamma\left[\Tr[(\hat{\omega}_x\otimes \hat{\omega}_y)(\hat{\Pi}\otimes\hat{\Pi})]\right]}{\Tr[\hat{\rho}_{0}^{\otimes K}\hat{\Pi}]^2} \\ 
    &\leq \sum_{k=1,\text{odd}}^{K}\sum_{l=1,\text{odd}}^K\binom{K}{k}\binom{K}{l}(4\epsilon_0)^{k+l}\left(\frac{1+\nu}{1+\nu+2\nu(k+l)\sigma'^2_\gamma}\right)^n \\ 
    &\leq \left(\frac{1+\nu}{1+\nu+2\nu\sigma^2_{\gamma,(1,1)}}\right)^n\sum_{k=1,\text{odd}}^{K}\sum_{l=1,\text{odd}}^K\binom{K}{k}\binom{K}{l}(4\epsilon_0)^{k+l} \\ 
    &\leq \left(\frac{1+\nu}{1+\nu+2\nu\sigma^2_{\gamma,(1,1)}}\right)^n\left(\sum_{k=1,\text{odd}}^{K}(4K\epsilon_0)^{k}\right)^2 \\ 
    &\leq \left(\frac{1+\nu}{1+\nu+2\nu\sigma^2_{\gamma,(1,1)}}\right)^n\frac{(4K\epsilon_0)^2}{(1-4K\epsilon_0)^2}.
\end{align}

Therefore, the averaged TVD is bounded by:
\begin{align}
    \mathbb{E}_\gamma[\text{TVD}(p_0,p_{\pm\gamma})]
    &\leq N_b\left[2\left(\left(\frac{1+\nu}{1+\nu+2\nu\sigma^2_{\gamma,(2,2)}}\right)^n\frac{(4K\epsilon_0)^4}{(1-4K\epsilon_0)^2}   \right)^{1/2}+\left(\frac{1+\nu}{1+\nu+2\nu\sigma^2_{\gamma,(1,1)}}\right)^n\frac{(4K\epsilon_0)^2}{(1-4K\epsilon_0)^2}\right] \\ 
    &\leq 320N_b K^2\epsilon_0^2\left[ \left(\frac{1+\nu}{1+\nu+2\nu\sigma^2_{\gamma,(2,2)}}\right)^{n/2}   + 5 \left(\frac{1+\nu}{1+\nu+2\nu\sigma^2_{\gamma,(1,1)}}\right)^n \right] \\     
    &\leq 320N_b K^2\epsilon_0^2\left[\left(\frac{1+\nu}{1+\nu+2\nu\sigma^2_{\gamma,(2,2)}}\right)^{n/2} +5\left(\frac{1+\nu}{1+\nu+2\nu\sigma^2_{\gamma,(1,1)}}\right)^{n/2}\right] \\ 
    &\leq 1920N_b K^2\epsilon_0^2\left(\frac{1+\nu}{1+\nu+2\nu\sigma^2_{\gamma,(1,1)}}\right)^{n/2},
\end{align}
where we used $\sigma_{\gamma,(2,2)}^2 \geq \sigma_{\gamma,(1,1)}^2$ and set $4K\epsilon_0\leq 0.9$.
Finally, the number of $K$-copy batches, $N_b$, required to satisfy the $\text{TVD} \geq 1/6$ lower bound from the revealed hypothesis testing task is:
\begin{equation}
   N_b\geq 8.7\cdot10^{-5} \frac{\left(1 + 2\tilde{\sigma}^2_{\gamma}\right)^{n/2}}{K^2\epsilon_0^2},
\end{equation}
where we defined $\tilde{\sigma}^{2}_{\gamma} \equiv \sigma^{2}_{\gamma}  \frac{1}{ \frac{1}{2}(\frac{1}{\nu}+1) + \sigma^2_{\gamma} (\frac{1}{\nu}-1) }$. 
As long as $\sigma_\gamma^2$ lies in either region, $\sigma_\gamma^2< 1/(2K)$ or $\sigma_\gamma^2 \geq 1/2$,  the proof holds for any $0 < \nu<1$. 
This makes our proof reasonably general. 
One can also verify that, for the $K=1$ case, the two regions cover the entire domain of $\sigma_\gamma^2$, just as we saw happen in the proof for single-copy measurement protocols. 

To relate the sample complexity to the learning parameters, recall that $\epsilon_0=\frac{1+\eta_3}{0.98}\epsilon$, $2\sigma_\gamma^2=0.99\kappa$, $\nu=\frac{\kappa n-\log(1+\eta_3)}{\kappa n+\log(1+\eta_3)}$, and $\eta_3\in\bigl(0,\tfrac{0.245}{\epsilon}-1\bigr]$, with $\eta_3$ also satisfying Eq.~\eqref{eq::eta-3-constraint} (see Sec.~\ref{sec::no-entanglement} for further details).
Then the number of batches must satisfy
\begin{align}
   N_b & \geq \frac{8.4\cdot10^{-5}}{K^2 \epsilon^2 (1+\eta_3)^2}  
   \left( \frac{1+0.99\kappa}{1+\tfrac{0.99}{n}\log(1+\eta_3)} \right)^{n/2}.\\
\end{align}
It is sufficient to further require
\begin{align}
   N_b & \geq \frac{8.4\cdot10^{-5}}{K^2 \epsilon^2 (1+\eta_3)^2}  
   \left( 1+0.99\kappa \right)^{n/2},\\
    N & \geq \frac{8.4\cdot10^{-5}}{K \epsilon^2 (1+\eta_3)^2}  
   \left( 1+0.99\kappa \right)^{n/2},
\end{align}
which asymptotically yields
\begin{align}
    N = \Omega \left( K^{-1}\epsilon^{-2}  \left( 1 +0.99\kappa \right)^{n/2}  \right),
\end{align}
where $ n \geq 8$, $\epsilon \in(0,0.245)$, $\kappa \geq 1/0.99$, and $K \leq 0.22/\epsilon$.
The condition on $K$ holds because Eq.~\eqref{eq::eta-3-constraint} does not impose tighter constraints on $\eta_3$ (setting $n=8$ and $\kappa = 1/0.99$ results in $\eta_3 \lessapprox 5.7$). 
Thus we choose any arbitrarily small $\eta_3$ within the allowed $\eta_3 \in (0, \frac{0.245}{\epsilon}-1]$ interval, e.g., $\eta_3 := 10^{-6}$, and substitute it into the $4K\epsilon_0 \leq 0.9$ constraint.

The sample complexity required to complete the characteristic function learning task for any general adaptive scheme with access to entangled measurements but without access to reflected states is still exponential in the number of modes, $n$. 
Moreover, as explained in Sec.~\ref{sec::revealed_hyp_test} and mentioned in Thm.~\ref{thm::scaling-entangled} the sample complexity holds whether we adopt the learning task in Def.~\ref{def::learning-task} and only learn up to a global sign, or we consider learning the characteristic function together with the sign.

\section{Upper bound for entanglement-free schemes}\label{sec::heterodyne-general}

In this section, we establish a rigorous sample complexity upper bound for EF learning schemes. We analyze a protocol based on heterodyne detection, which serves as a constructive proof that the characteristic function can be learned without entanglement, albeit with a sample complexity that scales exponentially with the phase space radius.

The protocol consists of performing standard heterodyne detection independently on each copy of the quantum state $\hat{\rho}$. The corresponding POVM elements are projections onto coherent states, $\{\hat{\Pi}_{\zeta}\}_{\zeta\in\mathbb{C}^{n}}$, defined as $\hat{\Pi}_{\zeta}=|\zeta\rangle\langle\zeta|/\pi^{n}$ where $|\zeta\rangle=\hat{D}(\zeta)|0\rangle$~\cite{serafini2017quantum}. The probability density of obtaining outcome $\zeta\in\mathbb{C}^{n}$ is exactly the Husimi $Q$-function:
\begin{align}
    p(\zeta)
    =\frac{\langle \zeta|\hat{\rho}|\zeta\rangle}{\pi^n}
    =\frac{1}{\pi^{2n}}\int d^{2n}\alpha\;\chi_{\hat{\rho}}(\alpha)e^{-|\alpha|^{2}/2}e^{\alpha^{\dagger}\zeta-\zeta^{\dagger}\alpha}.
\end{align}
This relation expresses the measurement statistics as the Fourier transform of the anti-normally ordered characteristic function (or equivalently, the Wigner characteristic function damped by a Gaussian envelope $e^{-|\alpha|^2/2}$). To recover the characteristic function $\chi_{\hat{\rho}}(\alpha)$, we apply the inverse Fourier transform:
\begin{align}
    \chi_{\hat{\rho}}(\alpha)
    = e^{|\alpha|^{2}/2}\int d^{2n}\zeta\; p(\zeta)e^{\zeta^{\dagger}\alpha-\alpha^{\dagger}\zeta}.
\end{align}
Given a dataset of $N$ i.i.d.\ measurement outcomes $\{\zeta^{(j)}\}_{j=1}^{N}$, we construct the following unbiased estimator for the characteristic function at any query point $\alpha_i$:
\begin{align}
    \hat{u}_i \equiv e^{|\alpha_i|^{2}/2}\frac{1}{N}\sum_{j=1}^{N}
    e^{(\zeta^{(j)})^{\dagger}\alpha_i-\alpha_i^{\dagger}\zeta^{(j)}}.
\end{align}
Since the term inside the summation is a bounded random variable (a phase factor), the variance of the estimator is dominated by the prefactor $e^{|\alpha_i|^2}$. By applying Hoeffding's inequality and the union bound over a set of $M$ query points $\mathcal{S}=\{\alpha_i\}_{i=1}^M$, we determine the sample complexity required to achieve an additive error $\epsilon$ with success probability $1-\delta$:
\begin{align}
    \Pr\left[ |\hat{u}_i - \chi_{\hat{\rho}}(\alpha_i)| \leq \epsilon \right] \geq 1 - \delta/M.
\end{align}
If the learning task is restricted to the phase space ball $|\alpha|^{2}\leq\kappa n$, the required sample complexity scales as:
\begin{align}
    N_{\text{HD}} = O\left(\epsilon^{-2} e^{\kappa n} \log(M/\delta)\right).
\end{align}
This heterodyne protocol serves as a concrete demonstration that the learning task, including sign resolution, is solvable without entanglement. 
The resulting sample complexity is exponential, validating the fundamental limitations established in Sec.~\ref{sec::no-entanglement}. 
However, the disparity between the achievable scaling ($e^{\kappa n}$) and the lower bound base ($(1+0.99\kappa)^n$) invites further theoretical analysis to determine whether the lower bound can be tightened or if superior EF protocols exist. 
Crucially, this exponential cost reinforces the operational separation identified in this work: entanglement is an indispensable resource for the characteristic function learning task, the absence of which prohibits the $n$-independent sample complexity achieved by the BM scheme (Sec.~\ref{sec::Appendix_BM}). 
Finally, while this scheme provides a universal upper bound valid for any $n$-mode input state, we show in Sec.~\ref{sec::classicality-section} that this cost can be refined when the input states possess only a limited degree of nonclassicality.

\section{Classicality-aware bounds for entanglement-free schemes}
\label{sec::classicality-section}

In Sec.~\ref{sec::heterodyne-general}, we introduced an EF heterodyne~(HD) (EF–HD) protocol that learns the characteristic function at any post-selected set of $M$ phase space points within the ball $\{\alpha\in\mathbb{C}^n:|\alpha|^2\le\kappa n\}$, with sample complexity
\begin{equation}
N_{\text{HD}}=O\!\big(\epsilon^{-2} e^{\kappa n}\log(M/\delta)\big).
\end{equation}
We also proved exponential lower bounds for EF schemes in Secs.~\ref{sec::no-entanglement} and \ref{sec::lower-bound_reflection_symmetry}. Motivated by the intuition that more classical states should be easier to learn than highly nonclassical ones, we now relate the learning cost directly to the state’s classicality.

Recall from Lemma~\ref{lem:tail} (Sec.~\ref{sec:pre}) that for a state with classicality parameter $s_{\text{max}}\in(0,1]$, the characteristic function exhibits Gaussian decay:
\begin{equation}
|\chi_{\hat\rho}(\alpha)|\le e^{-\frac{s_{\text{max}}}{2}|\alpha|^2}\qquad\forall\,\alpha\in\mathbb{C}^n.
\end{equation}
This exponential tail supports an accuracy‑driven truncation of phase space. Since the learning scheme is limited by an additive error $\epsilon$, we define an effective phase space radius
\begin{equation}\label{eq::phase space-truncation}
L_{\epsilon}(s_{\text{max}}) \equiv \frac{2}{s_{\text{max}}}\log\!\bigl(\epsilon^{-1}\bigr),
\end{equation}
such that for any $|\alpha|^2 \ge L_{\epsilon}(s_{\text{max}})$, the magnitude of the characteristic function is guaranteed to be negligible ($|\chi_{\hat\rho}(\alpha)| \le \epsilon$). Consequently, outside this radius, the zero estimator is sufficient to meet the accuracy requirement. This observation, illustrated in Fig.~\ref{fig::classicality-cutoff}, underpins the classicality‑aware upper and lower bounds derived below.

\begin{figure}[t]
  \centering
  \includegraphics[width=0.44\textwidth]{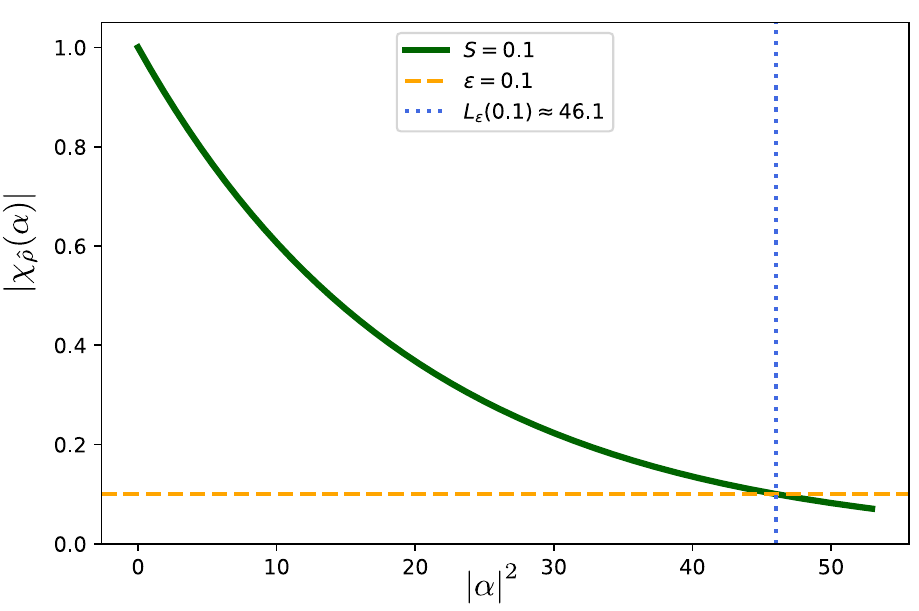}
  \caption{
    Visualization of the effective phase space truncation enabled by classicality. 
    The upper bound on the characteristic function magnitude 
    $|\chi_{\hro}(\alpha)| \leq \exp(-\frac{s_{\text{max}}}{2}|\alpha|^2)$ is plotted against the squared phase space radius $|\alpha|^2$ for $s_{\text{max}}=S=0.1$.
    The horizontal dashed line represents the target precision $\epsilon = 10^{-1}$.
    The intersection points define the effective cutoff radii $L_\epsilon(S)$.
  }
  \label{fig::classicality-cutoff}
\end{figure}

\subsection{Upper bound}
We refine the EF–HD upper bound by restricting inputs to states with classicality at least $s_{\text{max}}$.

\begin{theorem}[Classicality-aware upper bound]\label{thm::heterodyne-classicality}
Fix $\epsilon,\delta\in(0,1)$, $\kappa>0$, and $s_{\text{max}}\in(0,1]$. 
Let $\hro$ be an $n$-mode state with classicality at least $s_{\text{max}}$. 
The heterodyne-based scheme, given $N$ copies of $\hro$, outputs estimates $\{\hat u_i\}_{i=1}^{M}$ satisfying
$|\hat u_i-\chi_{\hro}(\alpha_i)|\leq\epsilon$ for all $|\alpha_i|^2\leq\kappa n$
with probability at least $1-\delta$, using
\begin{align}    
N=
\begin{cases}
O\!\left(\epsilon^{-2}e^{\kappa n}\log(M/\delta)\right), & \text{if }\kappa n\leq L_{\epsilon}(s_{\text{max}}),\\[2mm]
O\!\left(\epsilon^{-2\left(1+\frac{1}{s_{\text{max}}}\right)}\log(M/\delta)\right), & \text{if }\kappa n\geq L_{\epsilon}(s_{\text{max}}),
\end{cases}
\end{align}

where $L_{\epsilon}(s_{\text{max}})=\tfrac{2}{s_{\text{max}}}\log(\epsilon^{-1})$.
\end{theorem}

\begin{proof}
Perform heterodyne measurements on each copy of $\hro$ to obtain i.i.d. outcomes $\{\zeta_j\}_{j=1}^{N}$ drawn from its $Q$–function. 
Set the truncation threshold $L_{\epsilon}(s_{\text{max}})=\frac{2}{s_{\text{max}}}\log(\epsilon^{-1})$ and partition phase space into two regions.

\noindent\textbf{Region I:} $|\alpha_i|^2\leq L_{\epsilon}(s_{\text{max}})$.  
For $\alpha_i$ in this region, use the unbiased estimator (see Sec.~\ref{sec::heterodyne-general})
\begin{equation}
\hat u_i=e^{|\alpha_i|^2/2}\frac{1}{N}\sum_{j=1}^{N}e^{\zeta_j^{\dagger}\alpha_i-\alpha_i^{\dagger}\zeta_j}.
\end{equation}

Applying Hoeffding’s inequality (to both real and imaginary parts) gives
\begin{equation}
N=O(e^{|\alpha_i|^2}\epsilon^{-2}\log(M/\delta))
=O(e^{\kappa n}\epsilon^{-2}\log(M/\delta)),
\end{equation}
since $|\alpha_i|^2\leq\kappa n\leq L_{\epsilon}(s_{\text{max}})$.

\noindent\textbf{Region II:} $|\alpha_i|^2\geq L_{\epsilon}(s_{\text{max}})$.  
By Lemma~\ref{lem:tail}, $|\chi_{\hro}(\alpha_i)|\leq\epsilon$ in this region, and since we are limited by the scheme's accuracy, setting $\hat u_i=0$ ensures $|\hat u_i-\chi_{\hro}(\alpha_i)|\leq\epsilon$.

\smallskip
When $\kappa n\leq L_{\epsilon}(s_{\text{max}})$, all queries lie in Region~I. 
When $\kappa n\geq L_{\epsilon}(s_{\text{max}})$, use the heterodyne-based estimator inside Region~I and output $0$ outside. 
Applying a union bound over $M$ post-selected queries yields the stated $N$. 
For states with higher classicality $s>s_{\text{max}}$, adjust the $L_{\epsilon}(s_{\text{max}})$ threshold by $L_{\epsilon}(s)$ before outputting the estimator; the bound remains unchanged since it is expressed in terms of the worst classicality $s_{\text{max}}$. 
The scheme uses only single-copy heterodyne measurements and no reflected states.
\end{proof}

\textit{Comparison with the general heterodyne bound:}
Without classicality information, restricting to $|\alpha|^2\leq\kappa n$ requires
\begin{equation}
N_{\text{HD}}=O(\epsilon^{-2} e^{\kappa n} \log(M/\delta)).
\end{equation}
Knowing that the state has classicality at least $s_{\text{max}}$ refines this to
\begin{equation}
N=O(\epsilon^{-2(1+\frac{1}{s_{\text{max}}})}\log(M/\delta))
\leq O(e^{\kappa n}\epsilon^{-2}\log(M/\delta))
=N_{\text{HD}},
\end{equation}
whenever $\kappa n\geq L_{\epsilon}(s_{\text{max}})$.

\begin{corollary}[Classical states]
If $\hro$ is classical ($s_{\text{max}}=1$), then
\begin{equation}
N=O(\epsilon^{-4}\log(M/\delta)).
\end{equation}
\end{corollary}

\noindent
This matches the $O(\epsilon^{-4}\log(M/\delta))$ scaling of the efficient entanglement-assisted~(EA) BM protocol (Cor.~\ref{cor:BM-corollary}). Hence, in the classical regime, neither entanglement nor reflected states reduce the sample complexity beyond the classicality-aware heterodyne EF alternative.

\begin{figure}[t]
  \centering
  \includegraphics[width=0.48\textwidth]{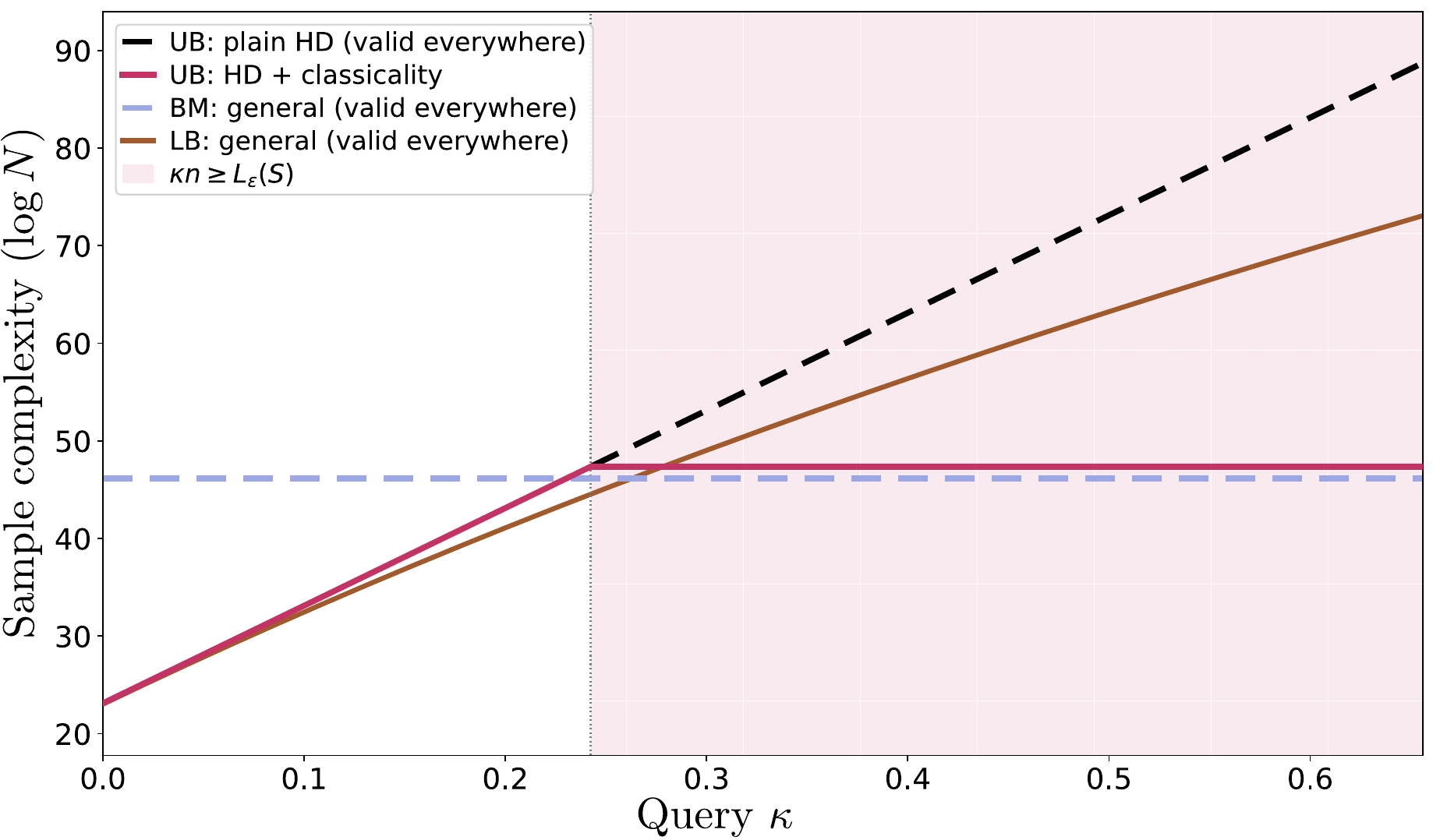}\hfill
  \includegraphics[width=0.48\textwidth]{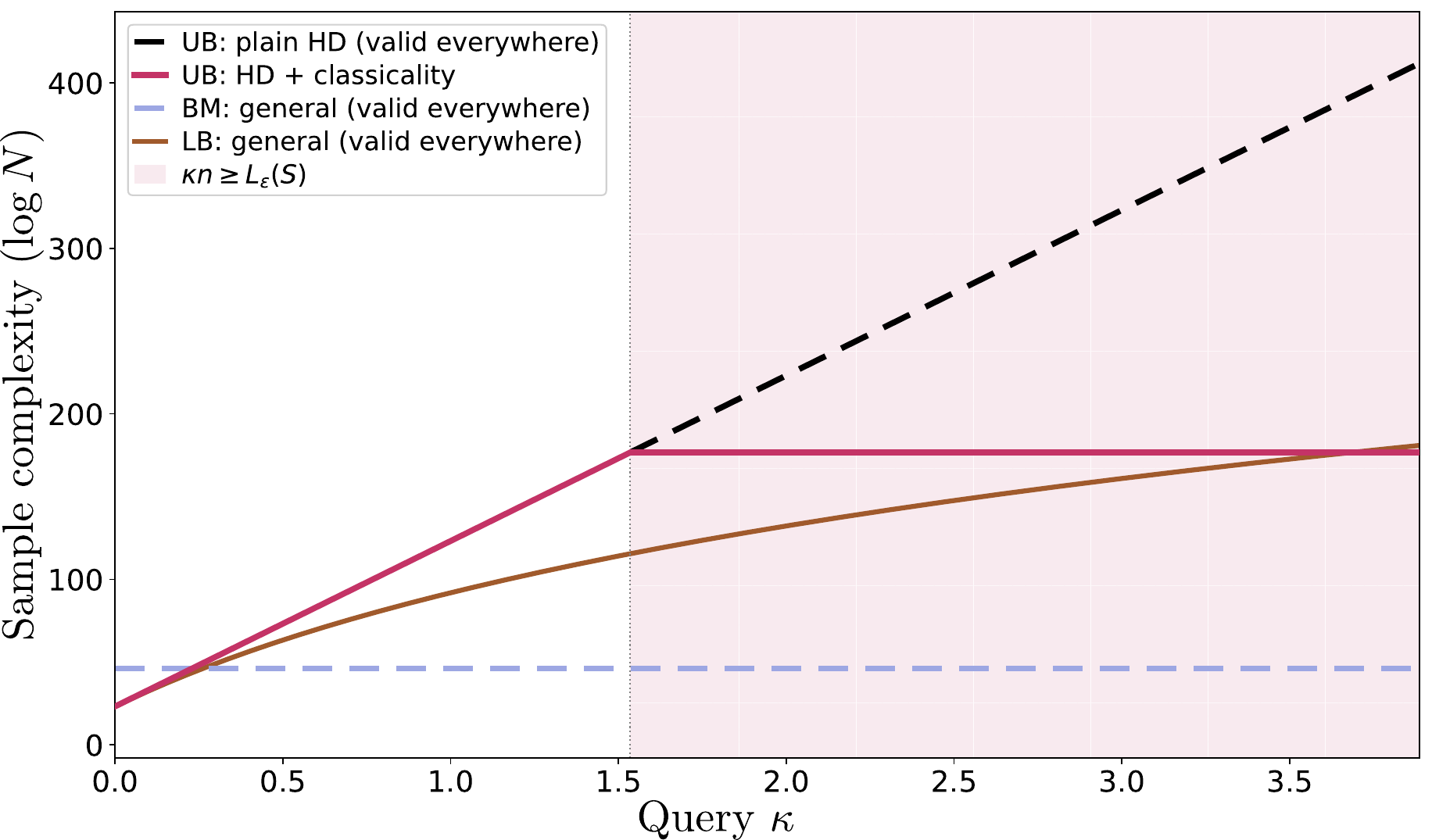}
  \caption{
    Sample complexity $N$ versus query radius $\kappa$ for EF heterodyne schemes (with/without classicality information) and the EA BM protocol.
    Each panel shows the classicality-independent lower bound~(LB), plain HD upper bound~(UB), classicality-aware HD UB, and BM scaling.
    Left: For $S=0.95$, the classicality-aware HD UB quickly flattens in $\kappa$ and significantly undercuts both the classicality-independent UB and LB; BM adds little improvement beyond this ($S=0.95$, $n=100$, $\epsilon=10^{-5}$).
    Right: For $S=0.15$, the classicality-aware HD UB still beats the classicality-agnostic HD UB, but a visible gap to BM remains, underscoring the value of entanglement and reflected states in the nonclassical regime ($S=0.15$, $n=100$, $\epsilon=10^{-5}$).
  }
  \label{fig::upper-bound-classicality}
\end{figure}

\textit{Edge cases and tradeoffs:}
If $s_{\text{max}}\le 0$ (or $s_{\text{max}}>0$ but very small), the decay bound $|\chi_{\hro}(\alpha)|\leq e^{-\frac{s_{\text{max}}}{2}|\alpha|^{2}}$ is too weak or invalid, and the heterodyne complexity reverts to $O(\epsilon^{-2} e^{\kappa n})$. 
By contrast, the EA BM protocol remains $O(\epsilon^{-4})$, independent of $n$, by learning $\chi_{\hro}$ (up to a global sign) everywhere. 
This clarifies the resource–sample–complexity tradeoff: for classical inputs ($s_{\text{max}}\approx 1$), classicality-aware heterodyne already achieves BM-like scaling, whereas for strongly nonclassical inputs ($s_{\text{max}}\approx 0$) our lower bounds show that eliminating the $e^{\kappa n}$ factor requires entanglement and reflected states.
Figure~\ref{fig::upper-bound-classicality} (Left) shows that for modest $\kappa$ and $S=s_{\text{max}}=0.95$, the classicality-aware upper bound is flat in $\kappa$, sits well below classicality-agnostic bounds, and nearly tracks BM. 
For more \textit{nonclassical} inputs, Fig.~\ref{fig::upper-bound-classicality} (Right) shows a persistent gap to BM despite improvement over the classicality-agnostic HD bound, highlighting the value of entanglement and reflected states. 
Thus, prior classicality information already lowers the EF sample cost substantially, while entanglement and reflected states are essential in the genuinely nonclassical regime.

\subsection{Lower Bound}\label{subsec::9-C}
Theorem~\ref{thm::scaling-single-copy2} gives an exponential lower bound for any EF learning scheme on arbitrary $n$-mode inputs.
It is natural to ask how this hardness varies with classicality (larger $s_{\text{max}}$ indicates more classical states).
Accordingly, we analyze the sample complexity when the learning scheme is restricted to inputs with classicality at least $S$.

We use a similar (partially) revealed hypothesis-testing task as in Sec.~\ref{sec::no-entanglement}.
In each round, Alice sampled $\gamma\in\mathbb{C}^n$ from a multivariate Gaussian (Eq.~\eqref{eq::multivariate-gaussian}) and, with equal probability, prepared either $N$ copies of the thermal state $\hro_0$ (and the corresponding reflected state $\hror_0$) or $N$ copies of the three-peak state $\hro_{s\gamma}$ (and the corresponding reflected state $\hror_{s\gamma}$).
After Bob performed all measurements, Alice revealed $\gamma$, and Bob must have correctly guessed the hypothesis.
If Bob has a learning scheme that outputs an estimator of the characteristic function with $\epsilon$-accuracy and success probability at least $2/3$, then he succeeds with probability $7/12$, i.e., $\mathbb{E}_{\gamma}\text{TVD}(p_0,p_{\pm\gamma})\geq 1/6$.
This follows because the discriminating queries, those with $2\sigma^2<|\gamma|^2\leq \kappa n$, are sampled with probability at least $1/2$.
Although a three-peak state may be parameterized by any $\gamma\in\mathbb{C}^n$, it suffices to consider $|\gamma|^2\in(2\sigma^2,\kappa n]$ to distinguish the two ensembles.

In this subsection, we restrict attention to inputs with classicality at least $S$; Alice promises Bob that if she chooses the three-peak hypothesis, then the states are at least $S$ classical.
Crucially, the sample complexity lower bound derived in Eq.~\eqref{eq:tvd-to-samples} is agnostic to the specific choice of the parameter $\nu$. 
This allows us to fix $\nu$ in the following theorem to strictly guarantee the classicality constraint ($s_{\text{max}} \ge S$) without invalidating the hardness result.
Consequently, achieving $\mathbb{E}_{\gamma}\text{TVD}(p_0,p_{\pm\gamma})\geq 1/6$ still requires exponential sample complexity, as established in Sec.~\ref{sec::no-entanglement}.
Although Bob’s effective phase space is truncated by $L_{\epsilon}(S)$ (see Eq.~\eqref{eq::phase space-truncation}), the distinguishing task remains exponentially hard within the subset of the ensemble satisfying the classicality promise.
Accordingly, this construction yields a sample-complexity lower bound for learning any $n$-mode input state satisfying $s_{\text{max}}\in[S,1)$.
\iffalse
In this subsection we restrict attention to inputs with classicality at least $S$; Alice promises Bob that if she chooses the three-peak hypothesis, then the states are at least $S$ classical.
Distinguishing the two ensembles with $\mathbb{E}_{\gamma}\text{TVD}(p_0,p_{\pm\gamma})\geq 1/6$ still requires exponential sample complexity, as shown in Sec.~\ref{sec::no-entanglement}.
Although Bob’s effective phase space is truncated by $L_{\epsilon}(S)$ (see Eq.~\eqref{eq::phase space-truncation}), the desired TVD can still be achieved by focusing on a subset of the ensemble whose states are at least $S$ classical.
Accordingly, this construction yields a sample-complexity lower bound for learning any $n$-mode input state satisfying $s_{\text{max}}\in[S,1)$.
\fi

\begin{theorem}{Classicality-aware lower bound.}\label{thm::scaling-single-copy2-classicality}
Let $\hat{\rho}$ be an arbitrary $n$-mode quantum state ($n\geq 8$) with classicality $s_{\text{max}}\in [S,1)$, and let $\hror$ denote its reflected state for some symmetric unitary $U\in\mathbb{C}^{n\times n}$. 
Assume $\epsilon\in(0,0.245)$ and $\kappa\in[\kappa_{\text{min}}(S),\kappa_{\text{max}}(S)]$, and consider any learning scheme on single copies of $\hro$ and $\hror$ (without entangled measurements). 
Suppose $\chi_{\hat{\rho}}(\alpha)$ is learned using $N$ copies of $\hro$ and $\hror$, and the scheme then receives a query $\alpha\in\mathbb{C}^n$ with $|\alpha|^2\leq \kappa n$. If the scheme outputs $\tilde{\chi}_{\hro}(\alpha)$ with $|\tilde{\chi}_{\hro}(\alpha)-\chi_{\hat{\rho}}(\alpha)|\leq \epsilon$ with probability at least $2/3$, then
\begin{align}
N=\Omega \left( \epsilon^{-2}  \big(1+0.99\kappa'\big)^n \exp\left[-2\kappa'nf(S) \right]  \right),
\end{align}

where $\kappa' = \min(\kappa, \max(\kappa^*(S), \kappa_{\text{min}}(S)) )$.
\begin{align}    
f(S) \equiv \frac{S}{1+\sqrt{1-S^2}}, 
\;\;\; \kappa_{\text{min}}(S) \equiv \frac{2}{0.99}\frac{S}{\sqrt{1-S^2}},
\;\;\; \kappa_{\text{max}} (S) \equiv \frac{1}{nf(S)} \log(\frac{0.245}{\epsilon}), 
\;\;\; \kappa^*(S) \equiv \frac{1}{2f(S)}-\frac{1}{0.99},
\end{align}

and the $\kappa$ domain is nonempty provided 
$S \leq s_{\text{cap}}(n,\epsilon) \equiv \frac{\sqrt{r^2+2r}}{r+1}, \;\; r\equiv \frac{0.99}{2n} \log(\frac{0.245}{\epsilon})$.
\end{theorem}

\begin{proof}
Three-peak states are parameterized by $(n,\gamma,\nu,\epsilon_0)$, to be specified below.
We assume $s_{\text{max}}\in[S,1)$ (verified in what follows).
The sampling distribution for $\gamma$ in the three-peak ensemble is parameterized by $\sigma_\gamma^2$, and the query $\kappa$ lies in the admissible domain above.
As recalled from Sec.~\ref{sec::no-entanglement}, setting $2\sigma_\gamma^2=0.99\kappa$ with $2\sigma_\gamma^2\geq 2\sigma^2$ yields
\begin{align}
\Pr\left[2\sigma^2<|\gamma|^2\le \kappa n\right] \geq 1/2.
\end{align}

Whenever $|\gamma|^2$ lies in that domain, the scheme distinguishes $\hro_{\tau\gamma}$ ($\tau\in{\pm1}$) from $\hro_{0}$ because
\begin{align}
\left|\chi_{\hro_{\tau\gamma}}(\gamma)-\chi_{\hro_0}(\gamma)\right| 
> 1.96 e^{-\kappa n/\Sigma^2}\epsilon_0 \geq 2\epsilon.
\end{align}

As discussed in Sec.~\ref{sec:input_states}, three-peak states parameterized by $\nu \in \left(0, \sqrt{\frac{1-S}{1+S}}~ \right] $ have a classicality of $s_{\text{max}} \in [S,1)$. 
Therefore, by setting 
\begin{align}\label{eq::state-nu}
\nu =  \sqrt{\frac{1-S}{1+S}}
\end{align}
we only allow states that are at least $S$ classical.

To rewrite a characteristic variance in terms of $s' = \frac{1-\nu^2}{1+\nu^2}$, we introduce
\begin{align}
f(s)\equiv \frac{s}{1+\sqrt{1-s^2}},
\end{align}
such that $f(S)=1/\Sigma^2$ when $S=s'=\frac{1-\nu^2}{1+\nu^2}$.
Again setting $2\sigma_\gamma^2=0.99\kappa$ with $0.99\kappa\geq 2\sigma^2$ yields
\begin{align}
\Pr\big[2\sigma^2<|\gamma|^2\leq \kappa n\big] \geq 1/2.
\end{align}
In this domain, successful distinction is guaranteed by
\begin{align}\label{eq::distinguishing-classicality}
\big|\chi_{\hro_{\tau\gamma}}(\gamma)-\chi_{\hro_0}(\gamma)\big|
> 1.96e^{-\kappa n f(S)}\epsilon_0
\geq 2\epsilon.
\end{align}
Before proceeding, introduce
\begin{align}
    &\kappa_{\text{min}}(S) \equiv \frac{2}{0.99}\frac{S}{\sqrt{1-S^2}}, 
     &\kappa_{\text{max}}(S) \equiv \frac{1}{nf(S)} \log(\frac{0.245}{\epsilon}), \label{eq::kappa-constraint}  \\
    &s_{\text{cap}}(n,\epsilon) \equiv \frac{\sqrt{r^2+2r}}{r+1}, \quad 
    & r \equiv \frac{0.99}{2n} \log(\frac{0.245}{\epsilon}),\\
    &\epsilon_0 \equiv \frac{\epsilon}{0.98} e^{\kappa n f(S)}.   
\end{align}
Therefore, for any $\kappa \in [\kappa_{\text{min}}(S), \kappa_{\text{max}}(S)]$, more than $1/2$ of the ensemble instances (with parameters $(\nu,\epsilon_0)$ and Gaussian $\gamma$) are distinguishable from $\hro_0$ using an $\epsilon$-accurate learning scheme.
Moreover, all states in the ensemble satisfy $s_{\text{max}}\in [S,1)$.
Thus a scheme that learns the characteristic function with success probability at least $2/3$ achieves $\mathbb{E}_{\gamma}\text{TVD}(p_0, p_{\pm \gamma}) \geq 1/6$.

Distinguishing the thermal and three-peak ensembles, i.e., achieving $\mathbb{E}_{\gamma}\text{TVD}(p_0,p_{\pm\gamma})\geq 1/6$ requires
\begin{align}
N \geq \frac{1}{96}\epsilon_0^{-2}(1+2\sigma_\gamma^2)^{n},
\end{align}
by Eq.~\eqref{eq:tvd-to-samples}. 
Substituting gives
\begin{align}\label{eq::labeled-equation}
N \geq \frac{0.98^2}{96}\epsilon^{-2}e^{-2\kappa n f(S)}(1+0.99 \kappa)^{n}.
\end{align}
For $S\in(0,0.795]$, the exponential base
\begin{align}
B(\kappa)\equiv e^{-2\kappa f(S)}(1+0.99 \kappa)
\end{align}
is maximized at
\begin{align}
\kappa^*(S)=\frac{1}{2f(S)}-\frac{1}{0.99}>0,
\end{align}
and decreases for $\kappa>\kappa^*$.
This reverse monotonicity reflects a limitation of the proof technique rather than of the task: learning on a larger region $|\alpha|^2\leq \kappa_1 n$ with $\kappa_1>\kappa^*$ should not be easier than learning on the smaller region $|\alpha|^2 \leq \kappa^* n$.
When $\kappa^*(S)<0$ (approximately $S\in(0.795,1)$), the exponential statement becomes trivial. We partition the classicality analysis into three domains.\\
\textbf{$S\in(0,0.459]$:} We have $\kappa^* (S)>\kappa_{\text{min}}(S)$, hence $\kappa^*>0$ and the base exceeds $1$ on a nonempty $\kappa$-interval, yielding an exponential lower bound.\\
\textbf{$S\in(0.459,0.567]$:} We have $\kappa^*(S) < \kappa_{\text{min}}(S)$ for the majority of this domain, while the base remains larger than $1$ at $\kappa = \kappa_{\text{min}}(S)$, so the bound is still exponential. \\
\textbf{$S \geq 0.568$:} On this domain, the base drops below $1$ even for the lowest admissible $\kappa = \kappa_{\text{min}}(S)$, and the resulting lower bound becomes trivial.

Asymptotically, an EF learning scheme receiving $N$ copies of an $n$-mode state $\hro$ with classicality $s_{\text{max}} \in [S,1)$ must satisfy
\begin{align}
    N = \Omega \left( \epsilon^{-2} e^{-2\kappa' n f(S)}(1+0.99\kappa')^{n} \right),
\end{align}
with $\kappa' = \min(\kappa, \max(\kappa^*(S), \kappa_{\text{min}}(S)) )$.
In Thm.~\ref{thm::scaling-single-copy2-classicality} we state the result for the full range $s_{\text{max}}\in(0,1)$, acknowledging that the bound is trivial for part of this domain. 
However, Thm.~6 in the main text strictly focuses on the exponential regime $S\in(0,0.567]$.

Eq.~\eqref{eq::kappa-constraint} also introduces a maximum value for the smallest admissible classicality, $s_{\text{cap}}(n,\epsilon)$, ensuring that $S \in (0,s_{\text{cap}}(n,\epsilon)]$.
This is a proof-technique artifact: the scheme can certainly learn states more classical than $s_{\text{cap}}(n,\epsilon)$, provided the threshold used in the analysis obeys $S\leq s_{\text{cap}}(n,\epsilon)$.
Finally, although the proof is developed for three-peak inputs with $s_{\text{max}}\in[S,1)$, learning $\chi_{\hro}$ to $\epsilon$-accuracy with success probability at least $2/3$ for any $n$-mode input of classicality $s_{\text{max}}\in[S,1)$ is at least as hard, so the lower bound applies generally.

\end{proof}

\begin{figure}[t]
\centering
\includegraphics[width=0.52\textwidth]{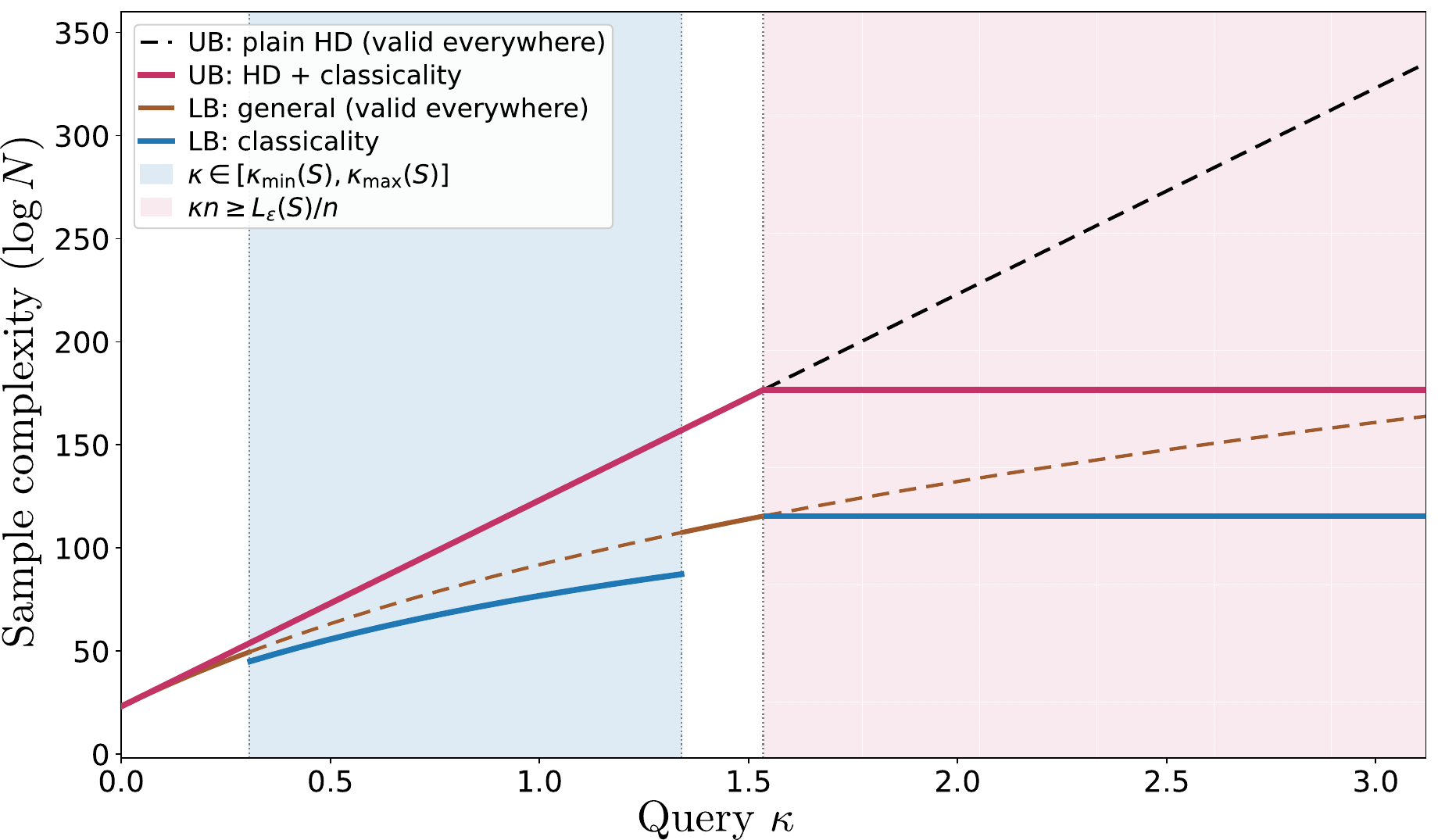}
\caption{
Sample complexity $N$ versus the query parameter $\kappa$, comparing EF schemes with and without classicality information.
For a classicality threshold $S=0.15$, we plot the classicality-independent lower bound~(LB), the classicality-aware LB, the plain heterodyne~(HD) upper bound~(UB), and the classicality-aware HD UB.
The classicality-aware LB tightens the fundamental cost on $\kappa \in [\kappa_{\min}(S),\kappa_{\max}(S)]$, while for $\kappa \geq L_{\epsilon}(S)/n$ both the classicality-aware LB and UB improve the sample complexity guarantees.
Thus, for this slice at $S=0.15$, exploiting classicality information yields strictly better sample-complexity guarantees than classicality-agnostic EF protocols ($n=100$, $\epsilon=10^{-5}$).
}
\label{fig::lower-upper-bound-classicality}
\end{figure}

\paragraph{Sample complexity comparison.}
Prior classicality information decreases the sample complexity for both the EF upper and lower bounds.
Because of finite accuracy, the threshold $L_{\epsilon}(S)$ separates a region where we output zero from the region where $\chi_{\hro}(\alpha)$ is learnable; this caps the complexity that would otherwise keep growing for classicality-agnostic schemes (see Fig.~\ref{fig::lower-upper-bound-classicality}).
Moreover, the lower bound itself drops throughout $[\kappa_{\text{min}}(S),\kappa_{\text{max}}(S)]$, quantifying how increased classicality reduces the learning cost lower bound.

As with the classicality-aware upper bound, for fixed $\kappa \geq L_{\epsilon}(S)/n$ the learning cost decreases as the state becomes more classical.
Figure~\ref{fig::lower-bound-classicality} compares the classicality-adjusted lower bound with the classicality-independent upper bound.
On the non-positive domain $S\in[-1,0]$, the general classicality-independent theorems apply.
On the positive domain $S\in(0,0.567]$, Fig.~\ref{fig::lower-bound-classicality}(left) shows a monotonic decrease in sample complexity as classicality increases.
The parameters $(n=10,\epsilon=10^{-3},\kappa=1.7)$ ensure that the classicality-adjusted lower bound is visible across the entire $S\in(0,0.567]$ range.
Figure~\ref{fig::lower-bound-classicality}(right) zooms into $S>0$ and highlights the $\kappa'$ transition as a function of classicality.

\paragraph{Recovery of the classicality-independent theorem.}
To recover Thm.~\ref{thm::scaling-single-copy2}, choose the classicality parameter
\begin{align}
    S = \frac{2L\kappa n}{\kappa^2 n^2 + L^2},
    \qquad 
    f(S) = \frac{L}{\kappa n},
\end{align}
where $L=\log(1+\eta_3)$ for an arbitrarily small $\eta_3>0$ satisfying 
$\eta_3 \in \bigl(0,\tfrac{0.245}{\epsilon}-1\bigr]$.
For this choice, the minimum admissible value of $\kappa$, namely 
$\kappa_{\text{min}}(S)$ from Eq.~\eqref{eq::kappa-constraint}, coincides with the earlier expression in Eq.~\eqref{eq::kappa_min_6}. 
The upper constraint $\kappa \le \kappa_{\text{max}}(S)$ is automatically satisfied by choosing an appropriate $\eta_3$, because $\kappa_{\text{max}}(S) \propto 1/f(S)$ diverges for arbitrarily small $\eta_3$. 
Likewise, by choosing $\eta_3$ arbitrarily small, the peak point $\kappa^{*}(S)$ can always be made larger than the queried $\kappa$, so $\kappa^{*}(S)$ does not restrict the bound in this regime. 
Finally,
\begin{align}
    e^{-2\kappa n f(S)} = e^{-2L} = (1+\eta_3)^{-2},
\end{align}
and the lower bound in Eq.~\eqref{eq::labeled-equation} therefore reduces exactly to the classicality-independent expression in Eq.~\eqref{eq::sample-complex-6}.

\begin{figure}[b]
\centering
\includegraphics[width=0.45\textwidth]{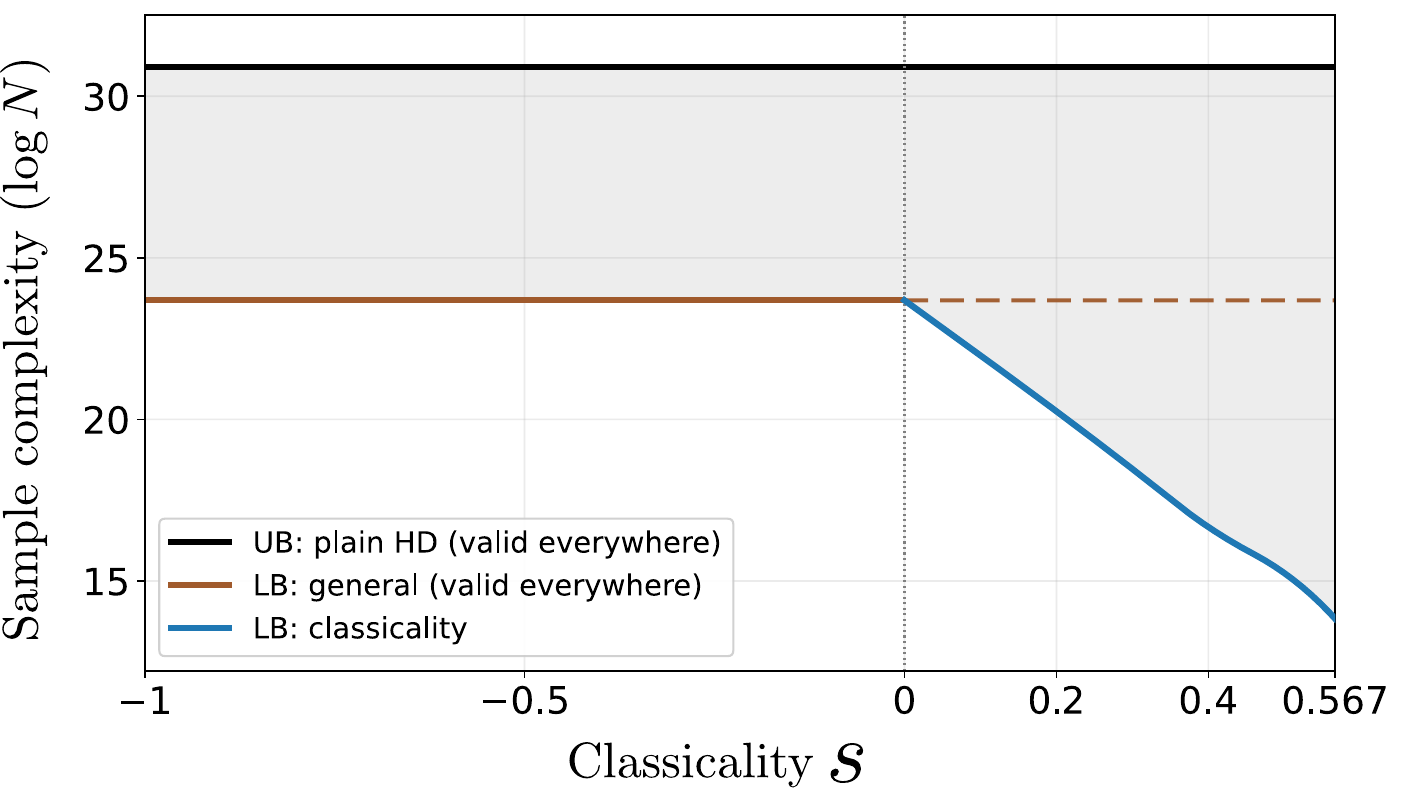}
\hspace{0.02\textwidth}
\includegraphics[width=0.45\textwidth]{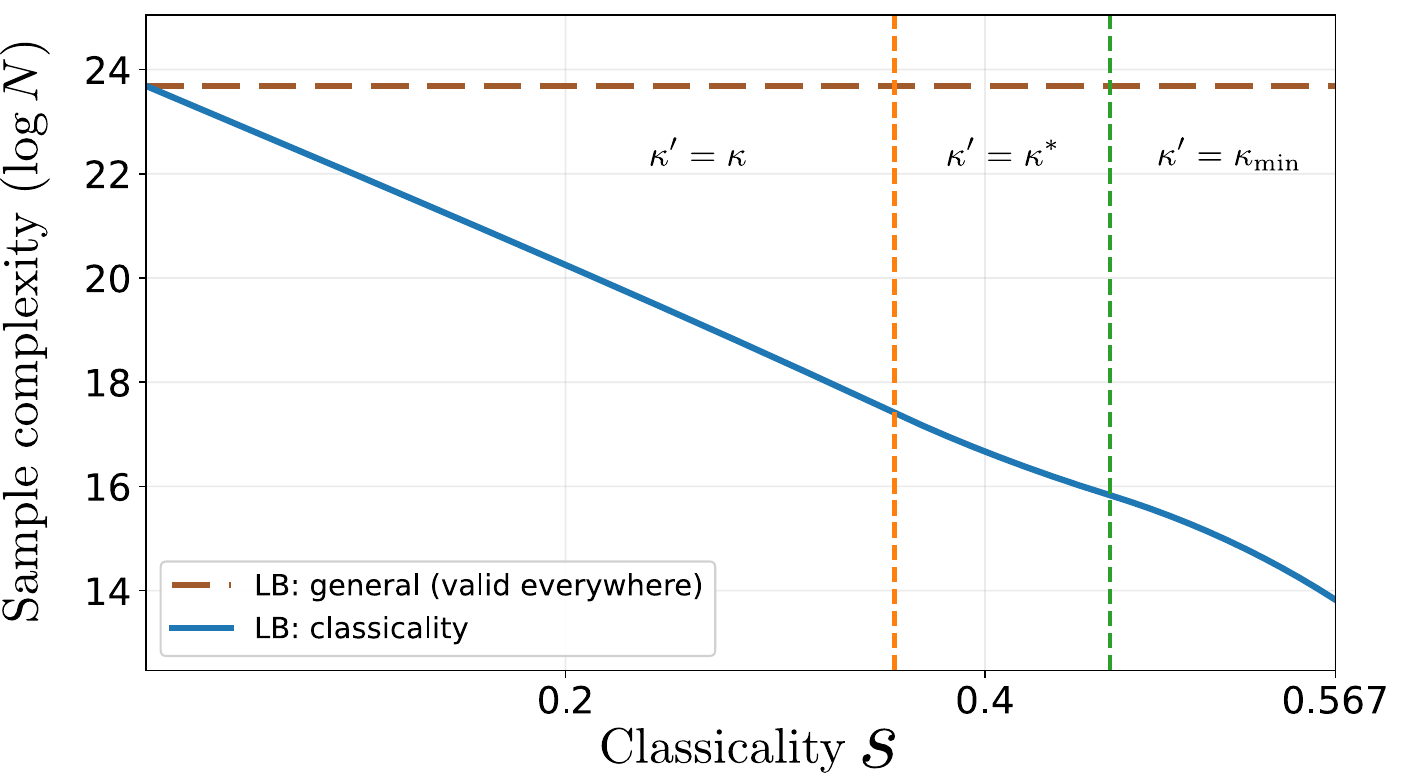}%
\caption{Left: Classicality-adjusted EF lower bound~(LB) and classicality-independent heterodyne upper bound~(UB) versus $S$.
The shaded region indicates the UB$-$LB gap.
Right: Positive-classicality zoom highlighting the transition in 
$\kappa'=\min(\kappa,\max\{\kappa^{*}(S),\kappa_{\min}(S)\})$. 
  Parameters: $n=10$, $\epsilon=10^{-3}$, $\kappa=1.7$.}
\label{fig::lower-bound-classicality}
\end{figure}

\section{Trivial lower bound for unrestricted schemes ($N=\Omega(\epsilon^{-2})$)}\label{sec::both-resources}
We recall the BM scheme in Sec.~\ref{sec::Appendix_BM}, which achieves the characteristic-function learning task with sample complexity
\begin{equation}
N_{\text{BM}}=O(\epsilon^{-4}\log(\delta^{-1})).
\end{equation}
We now ask whether this upper bound can be fundamentally improved. We prove a general lower bound that applies to any protocol with access to $N$ copies of $\hro\otimes \hror$, for some reflected state $\hror$, and arbitrary joint entangled measurements across all copies. We first present several lemmas used in the proof of the main theorem.

\begin{lemma}[Quantum Pinsker Inequality]\label{lemma::quantum-pinsker}
For density operators $\hro,\hat{\sigma}$,
\begin{equation}
\frac{1}{2\log 2}\|\hro-\hat{\sigma}\|_{1}^{2}\leq D(\hro\|\hat{\sigma}).
\end{equation}
\end{lemma}

\begin{lemma}[Umegaki relative entropy convergence]\label{lemma::umegaki-petz}
For $\alpha \in (0,1)\cup(1,\infty)$, the Petz R\'enyi relative entropy $D_{\alpha}(\hro \| \hat{\sigma})=\frac{1}{\alpha-1}\log\Tr[\hro^{\alpha}\hat{\sigma}^{1-\alpha}]$ satisfies
\begin{equation}
\lim_{\alpha\to 1}D_{\alpha}(\hro\|\hat{\sigma})=D(\hro\|\hat{\sigma}).
\end{equation}
\end{lemma}

\begin{lemma}[Petz-R\'enyi relative entropy monotonicity]\label{lemma::relative-entropy-monotonicty}
For $\alpha\in(0,1)\cup(1,\infty)$, the Petz R\'enyi relative entropy is monotonically increasing in $\alpha$~\cite{khatri2020principles}:
\begin{equation}
\alpha\leq\beta\implies D_{\alpha}(\hro\|\hat{\sigma})\leq D_{\beta}(\hro\|\hat{\sigma}).
\end{equation}
\end{lemma}

\begin{lemma}[Petz-R\'enyi $D_{2}$ bound]\label{lemma::petz-two}
In our setting $\hro_{0}$ is an infinite-rank thermal state whose support is the entire Hilbert space. As such,
\begin{equation}
D_{2}(\hro_{\gamma}\|\hro_{0})
=\log_{2}\Tr[\hro_{\gamma}^{2}\hro_{0}^{-1}]
=\log_{2}\Tr[\hro_{\gamma}\hro_{0}^{-1}\hro_{\gamma}]
\leq \log_{2}(1+8\epsilon_{0}^{2})
\leq \frac{8}{\log 2}\epsilon_{0}^{2}.
\end{equation}
The same bound holds for the reflected states, $D_{2}(\hror_{\gamma}\|\hror_{0})\leq \frac{8}{\log 2}\epsilon_{0}^{2}$.
\begin{proof}
\begin{equation}
D_{2}(\hro_{\gamma}\|\hro_{0})
=\log_{2}\Tr[\hro_{\gamma}^{2}\hro_{0}^{-1}]
=\log_{2}\Tr[\hro_{\gamma}\hro_{0}^{-1}\hro_{\gamma}],
\end{equation}
where the last equality follows by cyclicity of trace. Since $\hro_{0}=(1-\nu^{2})^{n}\nu^{2\hat{N}}$ with $0<\nu<1$, the inverse $\hro_{0}^{-1}$ is as an unbounded self-adjoint operator. Using the three-peak structure, one finds
\begin{align}
\hro_{\gamma}\hro_{0}^{-1}\hro_{\gamma}
&=(1-\nu^{2})^{n}\nu^{\hat{N}}\Bigl(\hat{D}^{\dagger}(0)+2i\epsilon_{0}\hat{D}^{\dagger}(\gamma)-2i\epsilon_{0}\hat{D}^{\dagger}(-\gamma)\Bigr)^{2}\nu^{\hat{N}}\\
&=(1+8\epsilon_{0}^{2})\hro_{0}
+4i\epsilon_{0}(1-\nu^{2})^{n}\nu^{\hat{N}}\Bigl(\hat{D}^{\dagger}(\gamma)-\hat{D}^{\dagger}(-\gamma)\Bigr)\nu^{\hat{N}}
-4\epsilon_{0}^{2}(1-\nu^{2})^{n}\nu^{\hat{N}}\Bigl(\hat{D}^{\dagger}(2\gamma)+\hat{D}^{\dagger}(-2\gamma)\Bigr)\nu^{\hat{N}}.
\end{align}
Taking the trace (see Eq.~\eqref{eq::displacement-helper}) yields
\begin{equation}\label{eq::D2-invariance}
\Tr[\hro_{\gamma}\hro_{0}^{-1}\hro_{\gamma}]
=1+8\epsilon_{0}^{2}\bigl(1-e^{-4a|\gamma|^{2}}\bigr)
\leq 1+8\epsilon_{0}^{2}.
\end{equation}
Thus
\begin{equation}
D_{2}(\hro_{\gamma}\|\hro_{0})\leq \log_{2}(1+8\epsilon_{0}^{2})\leq \frac{8}{\log 2}\epsilon_{0}^{2},
\end{equation}
where the second inequality uses $\log(1+x)<x$ for $x>0$. 
For reflected states, the same bound holds by mapping $\gamma \mapsto U\gamma^*$ in Eq.~\eqref{eq::D2-invariance}, hence $D_{2}(\hror_{\gamma}\|\hror_{0})=D_{2}(\hro_{\gamma}\|\hro_{0})$.
\end{proof}
\end{lemma}

\begin{theorem}
Let $\hat{\rho}$ be an arbitrary $n$-mode quantum state ($n\geq 8$) and let $\hror$ denote its reflected state for some symmetric unitary $U\in\mathbb{C}^{n\times n}$. Assume $\epsilon\in(0,0.245)$ and $\kappa>0$. Suppose $\chi_{\hat{\rho}}(\alpha)$ is learned using $N$ copies of $\hro$ and $\hror$, and the scheme then receives a query $\alpha\in\mathbb{C}^n$ with $|\alpha|^2\leq \kappa n$. Consider any learning scheme with access to arbitrary entangled measurements on the entire joint state. If the scheme outputs an estimate $\tilde{\chi}_{\hro}(\alpha)$ satisfying $|\tilde{\chi}_{\hro}(\alpha)-\chi_{\hat{\rho}}(\alpha)|\leq \epsilon$ with probability at least $2/3$, then
\begin{align}
N=\Omega\left(\epsilon^{-2}\right).
\end{align}
\end{theorem}

\begin{proof}
We consider the hypothesis-testing task from Sec.~\ref{sec::no-entanglement}, except that Bob now performs an entangled measurement on the joint state $(\hro \otimes \hror )^{\otimes N}$. We again use the family of three-peak states parametrized by $\nu \in (0,1)$, $\epsilon_0\in(0,1/4]$, and $\gamma \in \mathbb{C}^{n}$ (specified below). In Sec.~\ref{sec::revealed_hyp_test}, we showed that access to a learning scheme that estimates $\chi_{\hat{\rho}}(\gamma)$ for all $\gamma\in\mathbb{C}^n$ with $|\gamma|^2\le \kappa n$, to $\epsilon = 0.98\eta_1 e^{-\frac{\kappa n}{\Sigma^2}} \epsilon_0$ accuracy and success probability at least $2/3$, guarantees
\begin{align}\label{eq::star-tvd-lowerbound10}
    \mathbb{E}_{\gamma}\text{TVD}\left(p_0, p_{\pm\gamma}\right)\geq 1/6,
\end{align}
where $p_{0}$ and $p_{\pm\gamma}$ are the outcome distributions under single-copy measurements. The same statement holds when allowing an entangled measurement on $(\hro_{0}\otimes\hror_{0})^{\otimes N}$ or $(\hro_{\gamma}\otimes\hror_{\gamma})^{\otimes N}$. For a POVM $\{M_{o}\}_{o}$, define
\begin{align}
p_{0}(o)&=\Tr\bigl[M_{o}(\hro_{0}\otimes\hror_{0})^{\otimes N}\bigr],\\
p_{\gamma}(o)&=\Tr\bigl[M_{o}(\hro_{\gamma}\otimes\hror_{\gamma})^{\otimes N}\bigr],
\end{align}
and proceed without introducing an additional Bernoulli variable $s$.

As in Sec.~\ref{sec::revealed_hyp_test}, choose
\begin{equation}
\epsilon_{0}=\frac{\epsilon}{0.98}(1+\eta_{3}),\qquad
\Sigma^{2}=\frac{\kappa n}{\log(1+\eta_{3})},
\end{equation}
for some arbitrarily small $\eta_{3} \in (0, \frac{0.245}{\epsilon}-1]$ satisfying Eq.~\eqref{eq::eta-3-constraint}. 
This choice reduces the distinguishing task to the characteristic-function learning task for $2\sigma^2 <|\gamma|^{2}\leq \kappa n$. The reduction yields
\begin{align}\label{eq::tvd-general-scheme}
\frac{1}{6}\le \mathbb{E}_{\gamma}\text{TVD}(p_{0},p_{\gamma})
\le \frac{1}{2}\mathbb{E}_{\gamma}
\norm{(\hro_{\gamma}\otimes\hror_{\gamma})^{\otimes N}-(\hro_{0}\otimes\hror_{0})^{\otimes N}}_{1},
\end{align}
by Eq.~\eqref{eq::star-tvd-lowerbound10} and the Helstrom inequality. One can also verify that Eq.~\eqref{eq::star-tvd-lowerbound10} still holds after eliminating the expectation over the Bernoulli variable $s$ by following Eqs.~\eqref{eq::sec6-tvd-1/6-one}–\eqref{eq::sec6-tvd-1/6-three}. By Lemma~\ref{lemma::quantum-pinsker},
\begin{equation}
\norm{\hro_{\gamma}-\hro_{0}}_{1}\le \sqrt{2\log 2}\sqrt{D(\hro_{\gamma}\|\hro_{0})}.
\end{equation}
Taking the expectation value, applying Jensen's inequality for the concave square root, and then using Lemma~\ref{lemma::umegaki-petz} with Lemma~\ref{lemma::relative-entropy-monotonicty} gives
\begin{align}
\mathbb{E}_{\gamma}
\norm{(\hro_{\gamma}\otimes\hror_{\gamma})^{\otimes N}-(\hro_{0}\otimes\hror_{0})^{\otimes N}}_{1}
&\le \sqrt{2\log 2}\sqrt{\mathbb{E}_{\gamma}D\bigl((\hro_{\gamma}\otimes\hror_{\gamma})^{\otimes N}\|(\hro_{0}\otimes\hror_{0})^{\otimes N}\bigr)}\\
&\le \sqrt{2\log 2}\sqrt{\mathbb{E}_{\gamma}D_{2}\bigl((\hro_{\gamma}\otimes\hror_{\gamma})^{\otimes N}\|(\hro_{0}\otimes\hror_{0})^{\otimes N}\bigr)}.
\end{align}

By additivity of the Petz-R\'enyi $D_{2}$ on tensor products~\cite{khatri2020principles},
\begin{align}
\mathbb{E}_{\gamma}D_{2}\bigl((\hro_{\gamma}\otimes\hror_{\gamma})^{\otimes N}\|(\hro_{0}\otimes\hror_{0})^{\otimes N}\bigr)
&=N\mathbb{E}_{\gamma}\bigl(D_{2}(\hro_{\gamma}\|\hro_{0})+D_{2}(\hror_{\gamma}\|\hror_{0})\bigr).
\end{align}
Using Lemma~\ref{lemma::petz-two} and $D_{2}(\hror_{\gamma}\|\hror_{0})=D_{2}(\hro_{\gamma}\|\hro_{0})$,
\begin{align}
\mathbb{E}_{\gamma}
\norm{(\hro_{0}\otimes\hror_{0})^{\otimes N}-(\hro_{\gamma}\otimes\hror_{\gamma})^{\otimes N}}_{1}
&\le \sqrt{2\log 2}\sqrt{N\cdot \frac{16}{\log 2}\epsilon_{0}^{2}}
= \sqrt{32N\epsilon_{0}^{2}}.
\end{align}
Substituting into Eq.~\eqref{eq::tvd-general-scheme} gives
\begin{equation}
\frac{1}{6}\le \frac{1}{2}\sqrt{32N\epsilon_0^2}
\quad\Longrightarrow\quad
N\ge \frac{1}{288}\epsilon_{0}^{-2}.
\end{equation}
Relating the sample complexity to the learning parameters by reduction,
\begin{equation}
N\ge \frac{(0.98)^{2}}{288(1+\eta_{3})^{2}}\epsilon^{-2},
\end{equation}
which yields the asymptotic lower bound
\begin{equation}
N=\Omega(\epsilon^{-2}).
\end{equation}
\end{proof}

As in the previous proofs, any general learning scheme must succeed for any $n$-mode input; hence analyzing the three-peak state is without loss of generality. Moreover, the single-copy Petz-R\'enyi relative entropy in Lemma~\ref{lemma::petz-two} is invariant under reflection, i.e., $D_{2}(\hro_{\gamma}\|\hro_{0})=D_{2}(\hror_{\gamma}\|\hror_{0})$, so the lower bound is independent of how many copies of reflected states are used or in what order. Finally, $N$ independent BM on $N$ copies of $\hro \otimes \hror$ are a special case of an entangled measurement on $(\hro \otimes \hror)^{\otimes N}$. Thus, we have derived a lower bound to the corresponding $O(\epsilon^{-4})$ upper bound.

\bibliography{reference.bib}